\documentclass[11pt]{article}

\usepackage{xcolor}
\usepackage{amsfonts,amsmath,amssymb,mathtools,amsthm}
\usepackage{graphicx}
\usepackage{stmaryrd}
\usepackage{thmtools}
\DeclareGraphicsExtensions{.eps}

\renewcommand{\theequation}{\thesection.\arabic{equation}}

\newcommand\encadremath[1]{\vbox{\hrule\hbox{\vrule\kern8pt
\vbox{\kern8pt \hbox{$\displaystyle #1$}\kern8pt}
\kern8pt\vrule}\hrule}}
\def\enca#1{\vbox{\hrule\hbox{
\vrule\kern8pt\vbox{\kern8pt \hbox{$\displaystyle #1$}
\kern8pt} \kern8pt\vrule}\hrule}}

\newcommand\framefig[1]{
\begin{figure}[bth]
\hrule\hbox{\vrule\kern8pt
\vbox{\kern8pt \vbox{
\begin{center}
{#1}
\end{center}
}\kern8pt}
\kern8pt\vrule}\hrule
\end{figure}
}

\newcommand\figureframex[3]{
\begin{figure}[bth]
\hrule\hbox{\vrule\kern8pt
\vbox{\kern8pt \vbox{
\begin{center}
{\mbox{\epsfxsize=#1.truecm\epsfbox{#2}}}
\end{center}
\caption{#3}
}\kern8pt}
\kern8pt\vrule}\hrule
\end{figure}
}
\newcommand\figureframey[3]{
\begin{figure}[bth]
\hrule\hbox{\vrule\kern8pt
\vbox{\kern8pt \vbox{
\begin{center}
{\mbox{\epsfysize=#1.truecm\epsfbox{#2}}}
\end{center}
\caption{#3}
}\kern8pt}
\kern8pt\vrule}\hrule
\end{figure}
}

\renewcommand{\thesection}{\arabic{section}}
\renewcommand{\theequation}{\arabic{section}-\arabic{equation}}
\makeatletter
\@addtoreset{equation}{section}
\makeatother
\newtheorem{theorem}{Theorem}[section]

\newtheorem{proposition}{Proposition}[section]
\newtheorem{lemma}{Lemma}[section]
\newtheorem{corollary}{Corollary}[section]

\theoremstyle{definition}
\newtheorem{remark}{Remark}[section]
\newtheorem{definition}{Definition}[section]

\def\br{\begin{remark}\rm\small}
\def\er{\end{remark}}
\def\bt{\begin{theorem}}
\def\et{\end{theorem}}
\def\bd{\begin{definition}}
\def\ed{\end{definition}}
\def\bp{\begin{proposition}}
\def\ep{\end{proposition}}
\def\bl{\begin{lemma}}
\def\el{\end{lemma}}
\def\bc{\begin{corollary}}
\def\ec{\end{corollary}}
\def\beaq{\begin{eqnarray}}
\def\eeaq{\end{eqnarray}}

\theoremstyle{definition}

\newcommand{\be}{\begin{equation}}
\newcommand{\ee}{\end{equation}}
\newcommand{\beq}{\begin{equation}}
\newcommand{\eeq}{\end{equation}}
\newcommand{\bea}{\begin{eqnarray}}
\newcommand{\eea}{\end{eqnarray}}
\newcommand{\beqq}{\begin{equation*}}
\newcommand{\eeqq}{\end{equation*}}
\newcommand{\beaa}{\begin{eqnarray*}}
\newcommand{\eeaa}{\end{eqnarray*}}

\newcommand{\Tr}{{\operatorname {Tr}}}

\newcommand{\diag}{{\operatorname{diag}}}

\newcommand{\td}{\tilde}

\newcommand\blfootnote[1]{%
  \begingroup
  \renewcommand\thefootnote{}\footnote{#1}%
  \addtocounter{footnote}{-1}%
  \endgroup
}

\usepackage[pdftex]{hyperref}
\hypersetup{colorlinks,urlcolor=magenta,citecolor=red,linkcolor=blue,filecolor=black}

\title{\bf{Hamiltonian representation of isomonodromic deformations of twisted rational connections: The Painlev\'{e} $1$ hierarchy}}

\date{\vspace{-5ex}}
 
\author{$_{1}$Olivier Marchal\footnote{Universit\'e de  Lyon, Universit\'{e} Jean Monnet Saint-\'{E}tienne, CNRS UMR 5208, Institut Camille Jordan, Institut Universitaire de France, F-42023 Saint-Etienne, France.}\,\,,
$_{2}$Mohamad Alameddine\footnote{Universit\'e de  Lyon, Universit\'{e} Jean Monnet Saint-\'{E}tienne, CNRS UMR 5208, Institut Camille Jordan, F-42023 Saint-Etienne, France.}
}

\begin{document}

\maketitle

\vspace{1.5cm}

\begin{abstract}
In this paper, we build the Hamiltonian system and the corresponding Lax pairs associated to a twisted connection in $\mathfrak{gl}_2(\mathbb{C})$ admitting an irregular and ramified pole at infinity of arbitrary degree, hence corresponding to the Painlev\'{e} $1$ hierarchy. We provide explicit formulas for these Lax pairs and Hamiltonians in terms of the irregular times and standard $2g$ Darboux coordinates associated to the twisted connection. Furthermore, we obtain a map that reduces the space of irregular times to only $g$ non-trivial isomonodromic deformations. In addition, we perform a symplectic change of Darboux coordinates to obtain a set of symmetric Darboux coordinates in which Hamiltonians and Lax pairs are polynomial. Finally, we apply our general theory to the first cases of the hierarchy: the Airy case $(g=0)$, the Painlev\'{e} $1$ case $(g=1)$ and the next two elements of the Painlev\'{e} $1$ hierarchy.

\blfootnote{\textit{Email Addresses:}$_{1}$\textsf{olivier.marchal@univ-st-etienne.fr}, $_{2}$\textsf{mohamad.alameddine@univ-st-etienne.fr}}
\end{abstract}

\tableofcontents

\newpage

\section{Introduction and summary of the results}
Isomonodromic deformations have been studied since the beginning of the twentieth century \cite{Picard,Fuchs,Garnier,Painleve,Gambier,schlesinger1912klasse} and is still currently an active domain in modern mathematics. If the initial restriction to Fuchsian singularities is now very well understood, many questions in the case of irregular singularities remain open. If a geometrical understanding of the Hamiltonian representation of the isomonodromic equations for a generic meromorphic connection in $\mathfrak{gl}_d(\mathbb{C})$ with $d\geq 2$ is now well understood \cite{HURTUBISE20081394,BertolaHarnadHurtubise2022,KomyoLoray2023,Komyo2024}, the explicit expressions for the Hamiltonians and Lax pairs were derived on a case by case basis until some very recent results. In \cite{MartaPaper2022}, the authors obtained explicit expressions for the Hamiltonians using confluences of isomonodromic deformations of Fuchsian systems. In particular, this method may only obtain results for deformations obtained by confluences of simple poles limiting its range of application for $d\geq 3$. Independently, in \cite{MarchalOrantinAlameddine2022}, the authors proposed a generic construction and some explicit formulas for the Lax pairs and Hamiltonians associated to meromorphic connections in $\mathfrak{gl}_2(\mathbb{C})$ such that all leading orders at each pole are assumed to be diagonalizable with distinct eigenvalues.\footnote{For clarity in the exposition, results of \cite{MarchalOrantinAlameddine2022} are restricted to the stronger assumption that the matrices defining the Lax matrix are assumed to have distinct eigenvalues at all orders, but as remarked in the paper, the results may easily be generalized to the assumption that only leading orders at each pole are assumed to be diagonalizable and admitting an untwisted local formal fundamental form.} In addition, they proposed an explicit map from the geometric set of irregular times (defined in \cite{Boalch2001,Boalch2012,Boalch2022}) to a smaller set of isomonodromic times complemented by a set of trivial times and showed that the Darboux coordinates are independent of the trivial times. Thus, it is natural to wonder if the method of \cite{MarchalOrantinAlameddine2022} can be extended to the case where a meromorphic connection in $\mathfrak{gl}_2(\mathbb{C})$ exhibits a pole whose leading order cannot be diagonalized.

\medskip

The main purpose of this article is to obtain explicit Hamiltonians and Lax pairs for the generic ramified stratum of polynomial $\mathfrak{gl}_2(\mathbb{C})$ connections with a single pole at infinity. More precisely, we assume a non-diagonalizable leading coefficient and a nonzero coefficient of $\lambda^{2r_\infty-5}$ in $(\Tr\hat L)^2-4\det\hat L$. Equivalently, the leading odd irregular time satisfies $t_{\infty,2r_\infty-3}\neq0$. Non-diagonalizability alone does not imply this formal type. Together with \cite{MarchalOrantinAlameddine2022}, the present construction treats the corresponding generic untwisted and ramified strata; lower ramified strata and configurations with several poles require separate analysis. Let us also mention that the present paper has some important geometric interpretation. Indeed, the connection with the twisted flat symplectic Ehresmann connections of Boalch--Yamakawa \cite{boalch2015twistedwildcharactervarieties} is established by the irregular Riemann-Hilbert correspondence. In particular, Boalch--Yamakawa \cite{boalch2015twistedwildcharactervarieties} constructed the algebraic extension of the character variety to the twisted setting, henceforth describing the twisted connection. One of the main results of the article \cite{MarchalOrantinAlameddine2022} in the untwisted case is to provide an explicit bi-rational map from the Jimbo-Miwa-Ueno/Boalch symplectic Ehresmann connection to the isomonodromic connection. The present twisted case is analogous, one may interpret the construction of this paper as an explicit bi-rational map from the Boalch-Yamakawa twisted connection to the isomonodromic connection considered above the Riemann sphere with one ramified pole at $\infty$. Up to the knowledge of the authors, some explicit expressions for the Hamiltonians and Lax pairs for an arbitrary order twisted connection in $\mathfrak{gl}_2$ were never derived before. Our strategy is also different from several existing methods in the literature. Indeed, in \cite{MartaPaper2022}, the authors could derive some explicit formulas for the untwisted case using confluences of simple poles but this method seems hard to generalize to the twisted case and shall not provide all possible irregular cases for arbitrary rank. In \cite{Yamakawa2019FundamentalTwoForms} and in \cite{HarnadHurtubise,BertolaHarnadHurtubise2022}, D. Yamakawa and the Montr\'{e}al school first define isospectral Hamiltonians that satisfy isospectral equations. These isospectral equations differ from the isomonodromy equations in the case of irregular singularities and in order to make them coincide, Yamakawa chooses some special local diagonalization matrices whose exterior derivatives have special properties (Lemma $4.1$ of \cite{Yamakawa2019FundamentalTwoForms}). Unfortunately, the construction is done by recursion and is quite involved and hence neither explicit nor practical. The approach of the Montr\'{e}al school in \cite{BertolaHarnadHurtubise2022} is similar in the sense that they prove the existence of some special Darboux coordinates, known as isospectral Darboux coordinates for which the isospectral equations identify with the isomonodromic equations. However, getting the isospectral coordinates is very complicated since one needs to solve involved PDEs to grasp them. On the contrary, our construction is more straightforward since we directly solve the isomonodromic equations using some given sets of Darboux coordinates. In the untwisted case, the authors provided in \cite{marchal2023isomonodromicisospectraldeformationsmeromorphic} the connection between both approaches by relating our Darboux coordinates and their associated Hamiltonian systems with the isospectral Darboux coordinates. In the twisted case, the work of the Montr\'{e}al school \cite{BertolaHarnadHurtubise2022} would require adaptations but we believe that some explicit formulas for isospectral Darboux coordinates may be derived and that some relations similar to the one derived in \cite{marchal2023isomonodromicisospectraldeformationsmeromorphic} in the untwisted case could also be obtained.

\medskip

Let us emphasize that the ``twisted case'' requires a specific and non-trivial analysis. Indeed, the underlying geometry, in particular the definition of the irregular times at a ramified pole, is more difficult and less understood than the non-ramified case. Main results in this area are \cite{Boalch2001,biquard_boalch_2004,Boalch2012,Boalch2022,boalch2015twistedwildcharactervarieties} which extends the algebraic construction of the wild character variety to the twisted setting, we shall rely on these works throughout the article. One of the main difference in the twisted case is the necessity to introduce a ramified cover around each pole with some associated local coordinates in order to be able to ``diagonalize'' the singular part of the connection around the pole. 
Moreover, the definition of irregular times differs since for example the eigenvalues of the leading order of the Lax matrix are necessarily the same (because the matrix is assumed to be non-diagonalizable) so that the dimension of the space of irregular times and associated deformations drastically change. All these important changes require a detailed analysis of the twisted case that we propose in the present paper.

\medskip

In particular, our main results that can be seen as a summary and plan of the article are:
\begin{itemize} 
\item For any isomonodromic deformation, characterized by a vector $\boldsymbol{\alpha}$ in the tangent space, we provide an explicit gauge transformation  between the geometric Lax pair $\left(\td{L}(\lambda),\td{A}_{\boldsymbol{\alpha}}(\lambda)\right)$ and the companion-like Lax pair $\left(L(\lambda),A_{\boldsymbol{\alpha}}(\lambda)\right)$ in terms of apparent singularities $\left(q_i\right)_{1\leq i\leq g}$, their dual coordinates $\left(p_i\right)_{1\leq i\leq g}$ and the irregular times $\mathbf{t}$ in Proposition \ref{PropExplicitGaugeTransfo}. 
\item A general expression of the companion-like Lax pair $\left(L(\lambda),A_{\boldsymbol{\alpha}}(\lambda)\right)$ in terms of the Darboux coordinates $\left(q_i,p_i\right)_{1\leq i\leq g}$ is given in Propositions \ref{PropLaxMatrix}, \ref{PropA12Form} and \ref{Propcalpha}, complemented with equation \eqref{TrivialEntriesA}. These results follow from the local asymptotics at infinity of the wave matrix obtained in Proposition \ref{PropPsiAsymp} following the geometric construction of the twisted meromorphic connection.
\item Explicit expressions of the evolutions of the Darboux coordinates relatively to irregular times and a proof that these evolutions are indeed Hamiltonian with an explicit expression of the latter are given in Theorem \ref{HamTheorem}.
\item A symplectic change from Darboux coordinates $\left(q_i,p_i\right)_{1\leq i\leq g}$ to the set of symmetric Darboux coordinates $\left(Q_i,P_i\right)_{1\leq i\leq g}$ for which the geometric Lax pair $\left(\td{L}(\lambda),\td{A}_{\boldsymbol{\alpha}}(\lambda)\right)$ and the Hamiltonians are polynomial is provided in Definition \ref{DefNewCoord}. The explicit polynomial expressions are given in Theorem \ref{HamiltonianSymmetricCoordinates} and Propositions \ref{PropTdL} and \ref{PropTdA}.
\item A natural reduction of the space of deformations relatively to irregular times (of dimension $2g+4$) to a subspace of non-trivial isomonodromic deformations (of dimension $g$) complemented by a space of trivial deformations is detailed in Section \ref{SectionReduc}. Note in particular that this map is explicit both at the level of the tangent space (Definition \ref{TrivialVectors}) and at the level of times (Definition \ref{Times}). The terminology ``trivial deformations'' stands for the fact that the evolutions of the (shifted) Darboux coordinates $\left(\check{q}_i,\check{p}_i\right)_{1\leq i\leq g}$ relatively to these times are proven trivial in Theorem \ref{TheoReduction}. 
\item Simpler formulas in terms of the Darboux coordinates for the companion-like Lax pair $\left(L(\lambda),A_{\boldsymbol{\alpha}}(\lambda)\right)$ and Hamiltonians, after a canonical choice of the trivial times (given in Definition \ref{TrivialTimesChoice}), are provided in Proposition \ref{ProptdLtdAReduced} and Theorem \ref{HamTheoremReduced}.
\item Some simpler polynomial expressions in terms of the symmetric Darboux coordinates of the geometric Lax pair $\left(\td{L}(\lambda),\td{A}_{\boldsymbol{\alpha}}(\lambda)\right)$ and Hamiltonians, after a canonical choice of the trivial times (defined in Definition \ref{TrivialTimesChoice}), are provided in Proposition \ref{CheckLAEquationsReduced} and Theorem \ref{HamTheoremReduced}. 
\item The connection with the quantization of classical spectral curves via the topological recursion of \cite{EO07} is presented as a by-product in Section \ref{SectionTR}.
\end{itemize} 

Let us emphasize the factorization in Theorem \ref{HamTheoremReduced}: the Hamiltonians are time-dependent linear combinations of the companion coefficients $H_{\infty,k}$. The coefficients of this linear combination are determined by the lower triangular Toeplitz system in Proposition \ref{Propnureduced} and depend only on the times and the deformation direction. The functions $H_{\infty,k}$ are independent of the choice of deformation direction, but contain explicit $\hbar$ corrections and are not, in general, an involutive family of isospectral invariants. Their $\hbar=0$ parts give the corresponding characteristic coefficients on the canonical slice. The symmetric Darboux coordinates make both the Hamiltonians and the polynomial Lax matrices polynomial. Relating this description to isospectral Darboux coordinates in the ramified case remains a separate question, analogous to the untwisted construction in \cite{marchal2023isomonodromicisospectraldeformationsmeromorphic}.

\medskip

Some Maple files regarding the examples done in Section \ref{SectionExamples} are available on O.M. website at 
\url{https://math.univ-lyon1.fr/~marchal/AdditionalRessources/index.html}. They include verifications regarding the formulas proposed in this article.

\section{Twisted meromorphic connections at infinity}

\subsection{Twisted meromorphic connections and irregular times}

The space of $\mathfrak{gl}_2(\mathbb{C})$ meromorphic connections has been studied from many different perspectives. In the present article, we shall mainly follow the point of view of the Montr\'{e}al group \cite{Darboux_coord93,HarnadHurtubise1997} together with some insight from the work of Boalch and Yamakawa \cite{Boalch2001,boalch2012geometrybraidingstokesdata,boalch2015twistedwildcharactervarieties}. Let us first define the space we shall be studying.

\begin{definition}[Space of meromorphic connections with a pole at infinity]
Let $r_\infty\geq 3$ be a given integer. We shall consider
\beq
F_{\infty, r_\infty}:= \left\{\hat{L}(\lambda) = \sum_{k=1}^{r_\infty-1} \hat{L}^{[\infty,k]} \lambda^{k-1} \,\,/\,\, \{\hat{L}^{[\infty,k]}\} \in \left(\mathfrak{gl}_2(\mathbb{C})\right)^{r_\infty-1}\right\}/{GL}_2 
\eeq
where ${GL}_2$ acts simultaneously by conjugation on all the coefficients $\{\hat{L}^{[\infty,k]}\}_{1\leq k\leq r_\infty-1}$. The corresponding meromorphic connection is defined by
\beq d \hat{\Psi}= \hat{L}(\lambda) d\lambda \hat{\Psi} \,\, \Leftrightarrow \, \, \partial_\lambda \hat{\Psi}= \hat{L}(\lambda) \hat{\Psi}\eeq
where $\hat{\Psi}$ is referred to as the wave matrix.
\end{definition}

The space $F_{\infty,r_\infty}$ parametrizes polynomial representatives whose connection one-form has a pole of order at most $r_\infty$ at infinity. Its formal type determines whether the singularity is irregular; below we restrict to a specified generic ramified irregular locus. The generic untwisted case where the leading order $\hat{L}^{[\infty,r_\infty-1]}$ is diagonalizable with distinct eigenvalues has been studied in \cite{MarchalOrantinAlameddine2022} where a complete construction of the associated Hamiltonian systems is provided.  Note that most of the construction of \cite{MarchalOrantinAlameddine2022} extends to the larger untwisted case where
the leading coefficient is diagonalizable, provided that the irregular part admits an unramified formal diagonalization. We first introduce the space with non-diagonalizable leading coefficient, and then specify the generic ramified locus considered in this article.

\begin{definition}[Connections with non-diagonalizable leading coefficient]
Let $r_\infty\geq3$. We denote by $\hat F_{\infty,r_\infty}$ the subset of $F_{\infty,r_\infty}$ consisting of the classes for which $\hat L^{[\infty,r_\infty-1]}$ is non-diagonalizable.
\end{definition}

\noindent\textbf{Standing genericity assumption.}
Throughout the construction below, we restrict to the open subset on which
\beq\label{GenericRamificationAssumption}
[\lambda^{2r_\infty-5}]
\bigl((\Tr\hat L(\lambda))^2-4\det\hat L(\lambda)\bigr)\neq0,
\eeq
where $[\lambda^m]$ denotes coefficient extraction. This is a genericity assumption selecting the ramified formal type used in the paper; it is not a definition of ramification for an arbitrary connection. Unless stated otherwise, subsequent statements involving $\hat F_{\infty,r_\infty}$ are understood on this open subset.

To explain the assumption, put $d=r_\infty-2$ and use a constant conjugation to write
\beq
\hat L(\lambda)=
\begin{pmatrix}a&0\\c&a\end{pmatrix}\lambda^d
+\begin{pmatrix}u&b\\v&w\end{pmatrix}\lambda^{d-1}
+O(\lambda^{d-2}),\qquad c\neq0.
\eeq
Then
\beq
(\Tr\hat L)^2-4\det\hat L
=4bc\,\lambda^{2d-1}+O(\lambda^{2d-2}).
\eeq
Thus \eqref{GenericRamificationAssumption} is equivalent to $b\neq0$, whereas non-diagonalizability of the leading coefficient only ensures $c\neq0$. When $b\neq0$, conjugation by $\operatorname{diag}(b^{-1},1)$ sets the subleading upper-right entry to $1$ and the leading lower-left entry to $bc$. A further constant lower triangular conjugation makes the two subleading diagonal entries equal. With the irregular-time convention introduced below, $4bc=t_{\infty,2r_\infty-3}^{\,2}\neq0$, yielding precisely the normalization in \eqref{NormalizationInfty}. If $b=0$, a different formal type or normalization can occur, including an unramified irregular type or a lower ramified type; those strata are not covered by the present construction.

$\hat{F}_{\infty,r_\infty}$ can be given a Poisson structure inherited from the Poisson structure of a corresponding loop algebra and this space has been intensively studied from the point of view of isospectral and isomonodromic deformations. Following P. Boalch's works, this space is a Poisson space parametrizing isomorphism classes of meromorphic connections over the Riemann sphere admitting an irregular pole at $\infty$ and whose local Turritin-Levelt fundamental form at infinity is controlled by the irregular times (that shall be defined below in \eqref{Psi}). Deforming in an isomonodromic way implies that the moduli space becomes a flat Poisson fibre bundle over the base of deformation parameters (in our present case, only the irregular times at infinity). Symplectic leaves are specified by the appropriate local formal conjugacy data. This must be distinguished from fixing the full generalized monodromy data, including Stokes data: the latter specifies an isomonodromic solution as the irregular times vary. The general theory for the untwisted or twisted case has been described in \cite{Boalch2001,Boalch2012,Boalch2022}. Let us briefly review this perspective in our twisted setting and use it to define local coordinates on $\hat{F}_{\infty,r_\infty}$ trivializing the fibration.

\medskip

The main difference when dealing with the twisted case at infinity is that one needs to introduce a two-sheeted cover above infinity and define the local coordinate at infinity by
\beq \label{Defz}
z_\infty(\lambda)= \lambda^{-\frac{1}{2}} 
\eeq

The general theory of \cite{Boalch2022} implies the following proposition.

\begin{proposition}\label{PropDiago} Let $z\overset{\text{def}}{:=} \lambda^{\frac{1}{2}}$. For any given $\hat{L}(\lambda)$ in an orbit of $\hat{F}_{\infty,r_\infty}$, there exists a formal gauge matrix $G_\infty(z)$ around $\infty$ such that 
\beq G_\infty(z)=G_{\infty,-1}z+G_{\infty,0}+\sum_{k=1}^{\infty} G_{\infty,k}z^{-k} \,\text{ with }\, G_{\infty,-1} \,\text{ of rank 1}
\eeq
and
\begin{itemize} \item $\Psi_\infty(z)\overset{\text{def}}{=}G_\infty(z) \hat{\Psi}$ is a formal fundamental solution, also known as a Turritin-Levelt fundamental form (or Birkhoff factorization):
\small{\bea \label{Psi} \Psi_\infty(\lambda)&=&\Psi_{\infty}^{(\text{reg})}(z) \,\diag\left(\exp\left(-\sum_{k=1}^{2r_\infty-2} \frac{t_{\infty,k}}{k} z^k + \frac{1}{2} \ln z \right), \exp\left(-\sum_{k=1}^{2r_\infty-2} (-1)^{k}\frac{t_{\infty,k}}{k} z^k + \frac{1}{2}\ln z\right)  \right)\cr
&=&\Psi_{\infty}^{(\text{reg})}(z) \,\diag\left( \exp\left(-\sum_{k=1}^{2r_\infty-2} \frac{t_{\infty,k}}{k} \lambda^{\frac{k}{2}} + \frac{1}{4} \ln \lambda\right) , \exp\left(-\sum_{k=1}^{2r_\infty-2} (-1)^{k}\frac{t_{\infty,k}}{k} \lambda^{\frac{k}{2}} + \frac{1}{4}\ln \lambda \right) \right)\cr
&&
\eea}\normalsize
where $\Psi_{\infty}^{(\mathrm{reg})}(z)\in\mathrm{GL}_2(\mathbb C[[z^{-1}]])$ is a formal regular factor with invertible constant term; convergence is not asserted. Absorbing this factor into $G_\infty$, we may and shall choose $\Psi_{\infty}^{(\mathrm{reg})}=I_2$ in the formal normal form below.
\item The associated Lax matrix $L_\infty= G_\infty \hat{L} G_{\infty}^{-1}+ (\partial_\lambda G_\infty)G_\infty^{-1}$ is then in exact formal diagonal normal form:
\bea L_\infty(\lambda)&=&\diag\left(-\frac{1}{2}\sum_{k=1}^{2r_\infty-2} t_{\infty,k} z^{k-2} + \frac{1}{4z^2}, -\frac{1}{2}\sum_{k=1}^{2r_\infty-2}(-1)^{k} t_{\infty,k} z^{k-2} + \frac{1}{4z^2}  \right)\cr
&=&\diag\left(-\frac{1}{2}\sum_{k=1}^{2r_\infty-2} t_{\infty,k} \lambda^{\frac{k}{2}-1} + \frac{1}{4\lambda}, -\frac{1}{2}\sum_{k=1}^{2r_\infty-2} (-1)^{k} t_{\infty,k} \lambda^{\frac{k}{2}-1} + \frac{1}{4\lambda}  \right)
\eea
\end{itemize}
The complex numbers $\left(t_{\infty,k}\right)_{1\leq k\leq 2r_\infty-2}$ define the ``irregular times'' at infinity that we shall denote $\mathbf{t}=\{(t_{\infty,k})_{1\leq k\leq 2r_\infty-2}\}$ the irregular type of $\hat{L}\in \hat{F}_{\infty,r_\infty}$. On the stratum fixed above, $t_{\infty,2r_\infty-3}\neq0$.
\end{proposition}

\begin{remark}The coefficients $\mathbf{t}=\left(t_{\infty,k}\right)_{1\leq k\leq 2r_\infty-2}$ are also referred to as ``spectral times'' or ``KP times'' in part of the literature. Their topology has been studied in \cite{RembadoDoucot2023}.
\end{remark}
\begin{remark}
The locus $t_{\infty,2r_\infty-3}=0$ is outside the generic ramified stratum considered here. On that locus the formal type, the normalization and the number of effective deformation parameters may change. In particular, the singularity of $M_\infty$ cannot in general be resolved merely by deleting rows and columns. Extensions to lower ramified or unramified strata require a separate formal classification and are not claimed in this paper.
\end{remark}

The deck transformation $z\mapsto-z$ exchanges the two formal eigenlines. With the chosen sheet labels this gives $t_{\infty^{(2)},k}=(-1)^kt_{\infty^{(1)},k}$; this conclusion uses the ramified formal type, rather than the trace alone.
\begin{remark}\label{RemarkNorma}
In the formal diagonal gauge just fixed,
\beq
\Tr L_\infty=-\sum_{k=1}^{r_\infty-1}t_{\infty,2k}\lambda^{k-1}
+\frac{1}{2\lambda}.
\eeq
The shearing construction has $\det G_\infty=c z(1+O(z^{-1}))$ with $c\neq0$. Hence $\partial_\lambda\log\det G_\infty=1/(2\lambda)+O(\lambda^{-3/2})$. Since $\Tr\hat L$ is polynomial, the exact trace gauge identity gives
\beq
\Tr\hat L=-\sum_{k=1}^{r_\infty-1}t_{\infty,2k}\lambda^{k-1}.
\eeq
This also identifies the constant trace coefficient when $r_\infty=3$. In this diagonal formal gauge the determinant only involves even powers of $z$ and satisfies
\beq \det L_\infty=\frac{1}{4}t_{\infty,2r_\infty-2}^2\lambda^{2r_\infty-4}+\left(-\frac{1}{4}t_{\infty,2r_\infty-3}^2+\frac{1}{2}t_{\infty,2r_\infty-2}t_{\infty,2r_\infty-4} \right)\lambda^{2r_\infty-5}+O\left(\lambda^{2r_\infty-6}\right)\eeq
The nonzero subleading upper-right coefficient may be normalized to $1$ by the remaining constant diagonal conjugation. In this normalization the lower-left leading coefficient is $t_{\infty,2r_\infty-3}^{\,2}/4$. Thus the coefficient normalized to $1$ belongs to $\hat L^{[\infty,r_\infty-2]}$, not to the leading matrix.
\end{remark}

\subsection{Choice of representative normalized at infinity}

Fixing the irregular type of $\hat{L}(\lambda)$ does not fix it uniquely. In fact, the space 
\beq
\hat{\mathcal{M}}_{\infty,r_\infty,\mathbf{t}} :=\left\{\hat{L}(\lambda) \in \hat{F}_{\infty,r_\infty}\,\,/\,\, \hat{L}(\lambda) \,\text{ has irregular type } \mathbf{t} \right\}
\eeq
is, on its smooth stable locus, a symplectic manifold, seen as the symplectic quotient of a product of coadjoint orbits \cite{Boalch2012}, of dimension
\beq
\dim \hat{\mathcal{M}}_{\infty,r_\infty,\mathbf{t}} = 2r_\infty-6 = 2g
\eeq
where
\beq
\label{GenusDef} g:= r_\infty-3
\eeq
is the geometric genus of the generic smooth spectral curve $\det(yI_2-\hat L(\lambda))=0$. Singular members can have smaller normalization genus.

For any value of the irregular times, the Montr\'{e}al group introduced a set of local Darboux coordinates $\left(q_i,p_i\right)_{1\leq i\leq g}$ on $ \hat{\mathcal{M}}_{\infty,r_\infty,\mathbf{t}}$. Indeed, in each orbit in $\hat{F}_{\infty,r_\infty}$, the global action of $GL_2(\mathbb{C})$ implies that we may choose the leading coefficient $\hat{L}^{[\infty,r_\infty-1]}$ as a lower triangular matrix with identical coefficients on the diagonal (which is the standard form for a non-diagonalizable matrix of size $2$). Furthermore, the remaining action allows to fix the coefficients on the diagonal of the subleading order $\hat{L}^{[\infty,r_\infty-2]}$ at equal values. Combining this choice with Remark \ref{RemarkNorma}, we obtain the existence of a unique element for which $\hat{L}(\lambda)$ is of the form
\beq
\hat{L}(\lambda) =\begin{pmatrix} -\frac{1}{2}t_{\infty,2r_\infty-2} & 0\\ \frac{1}{4} (t_{\infty,2r_\infty-3})^2 & -\frac{1}{2}t_{\infty,2r_\infty-2} \end{pmatrix}\lambda^{r_\infty-2}+ \begin{pmatrix}-\frac{1}{2}t_{\infty,2r_\infty-4}&1\\X&-\frac{1}{2}t_{\infty,2r_\infty-4} \end{pmatrix}\lambda^{r_\infty-3}+ O\left(\lambda^{r_\infty-4}\right).
\eeq

One may thus identify $\hat{\mathcal{M}}_{\infty,r_\infty, \mathbf{t}}$ with the space of such representatives
\beq\label{NormalizationInfty}
\begin{array}{ll}
\hat{\mathcal{M}}_{\infty,r_\infty,\mathbf{t}} \sim &\Big\{  \tilde{L}(\lambda) = {\displaystyle \sum_{k=1}^{r_\infty-1}} \tilde{L}^{[\infty,k]} \lambda^{k-1}
\,/\, \tilde{L}^{[\infty,r_\infty-1]}=\begin{pmatrix} -\frac{1}{2}t_{\infty,2r_\infty-2} & 0\\ \frac{1}{4} (t_{\infty,2r_\infty-3})^2 & -\frac{1}{2}t_{\infty,2r_\infty-2}\end{pmatrix}\cr
&\text{ and } \tilde{L}^{[\infty,r_\infty-2]} =\begin{pmatrix}-\frac{1}{2}t_{\infty,2r_\infty-4}&1\\ \delta_\infty&-\frac{1}{2}t_{\infty,2r_\infty-4} \end{pmatrix}, \, \delta_\infty \in \mathbb{C},\quad \tilde L\text{ has irregular type }\mathbf t
\Big\}
\end{array} .
\eeq
In the following, we shall use the notation $\tilde{L}(\lambda)$ whenever we consider such a representative and call it a representative ``normalized at infinity''. 

\subsection{Darboux coordinates}

 On the smooth stable locus, the work of the Montr\'{e}al group implies that the space $\hat{\mathcal{M}}_{\infty, r_\infty,\mathbf{t}}$ is a symplectic manifold of dimension $2r_\infty-6=2g$. Consequently, one may define a set of Darboux coordinates $(q_i,p_i)_{1\leq i\leq g}$ on $\hat{\mathcal{M}}_{\infty, r_\infty,\mathbf{t}}$ that we shall present in this section. Let $\tilde{L}(\lambda) \in \hat{\mathcal{M}}_{\infty, r_\infty,\mathbf{t}}$ be a representative of the form described above in eq. \eqref{NormalizationInfty}. By definition, the entry $\left[\tilde{L}(\lambda)\right]_{1,2}$ is a monic polynomial function of $\lambda$ of degree $r_\infty-3=g$. We work on the open chart where its roots are pairwise distinct, $q_i\neq q_j$ for $i\neq j$, and define $(q_i)_{1\leq i\leq g}$ as the $g$ zeroes of $\left[\tilde{L}(\lambda)\right]_{1,2}$
\beq\label{Conditionqi}
\forall\, i\in \llbracket 1,g\rrbracket \, : \; \left[\tilde{L}(q_i)\right]_{1,2} = 0.
\eeq
This defines half of the spectral Darboux coordinates. The second half is obtained by evaluating the entry $\left[\tilde{L}(\lambda)\right]_{1,1}$  at $\lambda = q_i$,
\beq \label{Conditionpi}
\forall\, i\in \llbracket 1,g\rrbracket \, : \; p_i:=\left[\tilde{L}(q_i)\right]_{1,1}.
\eeq

Let us remark that, for any $i \in \llbracket 1,g \rrbracket$, the pair $\left(q_i,p_i\right)$ is by definition a point on the spectral curve defined by $\det(y I_2 -\td{L}(\lambda))=0$. In other words, we have
\beq 
\forall\, i\in \llbracket 1,g\rrbracket \, : \; \det(p_i I_2-\tilde{L}(q_i)) = 0.
\eeq

As in the untwisted case,the previous construction provides a local description of the space $\hat{F}_{\infty,r_\infty}$ as a trivial bundle $\hat{F}_{\infty, r_\infty} \to B$ where the base $B =\mathbf{t}$ is the set of irregular times. The fiber above a point $\mathbf{t} \in B$ is $\hat{\mathcal{M}}_{\infty, r_\infty,\mathbf{t}}$ that we equip with spectral Darboux coordinates $\left(q_i,p_i\right)_{1\leq i\leq g}$.

The space $B$ is a space of isomonodromic deformations meaning that any vector field $\partial_t \in T_{\mathbf t}B$ gives rise to a deformation of $\tilde{L}(\lambda)$ preserving its generalized monodromy data. There exist different equivalent ways to characterize the property of being an isomonodromic vector field. The one that we shall use in this article is the existence of a compatible system of the form
\beq \label{LaxPairDef}
\left\{
\begin{array}{l}
\partial_\lambda \tilde{\Psi}(\lambda,\mathbf{t}) = \tilde{L}(\lambda) \tilde{\Psi}(\lambda,\mathbf{t})  \cr
\partial_t \tilde{\Psi}(\lambda,\mathbf{t}) = \tilde{A}_t(\lambda) \tilde{\Psi}(\lambda,\mathbf{t})  \cr
\end{array}
\right.
\eeq
where $ \tilde{A}_t(\lambda) $ is a polynomial function of $\lambda$ with a pole at infinity lower or equal to $r_\infty-1$. The highest even-time deformation contributes the scalar term $-\lambda^{r_\infty-1}I_2/(2r_\infty-2)$. Equations \eqref{LaxPairDef} are referred to as a Lax pair whose compatibility condition is
\beq
\partial_\lambda \tilde{A}_t(\lambda) - \partial_t \tilde{L}(\lambda) + \left[\tilde{L}(\lambda) , \tilde{A}_t(\lambda) \right] = 0.
\eeq

\subsection{Scalar differential equation and companion gauge}

Let us now consider an orbit in $\hat{F}_{\infty,r_\infty}$ and a representative $\td{L}(\lambda)$ of this orbit normalized at infinity as above. Let $\td{\Psi}(\lambda)$ be a wave matrix solution to the linear system
\beq  \partial_\lambda \td{\Psi}(\lambda)=\td{L}(\lambda)\td{\Psi}(\lambda).\eeq

The differential system $\partial_\lambda \td{\Psi}(\lambda)=\td{L}(\lambda)\td{\Psi}(\lambda)$ may be rewritten into a scalar differential equation for $\td{\Psi}_{1,1}$ that is equivalent to a companion like matrix system. More precisely, defining 
\beq \label{GaugeGexpr}  \Psi(\lambda)=G(\lambda) \td{\Psi}(\lambda) \,\,\text{with}\,\, G(\lambda)=\begin{pmatrix} 1&0\\ \td{L}_{1,1}& \td{L}_{1,2}\end{pmatrix}\eeq
we get that $\Psi$ is a solution of the companion-like system
\beq \label{CompanionMatrix}\partial_\lambda \Psi(\lambda)=L(\lambda)\Psi(\lambda)\,\,\text{with}\,\, L(\lambda)=\begin{pmatrix}0&1\\ L_{2,1}&L_{2,2}\end{pmatrix}\eeq 
given by
\bea \label{LInTermsOfTdL} L_{2,1}&=&-\det \td{L}+\partial_\lambda\td{L}_{1,1}-\td{L}_{1,1}\frac{\partial_\lambda\td{L}_{1,2}}{\td{L}_{1,2}},\cr
L_{2,2}&=&\Tr\, \td{L} +\frac{\partial_\lambda\td{L}_{1,2}}{\td{L}_{1,2}}.
\eea
Note in particular that the first line of $\Psi$ and $\td{\Psi}$ is obviously the same: $\Psi_{1,1}=\td{\Psi}_{1,1}\overset{\text{def}}{=}\psi_1$ and $\Psi_{1,2}=\td{\Psi}_{1,2}\overset{\text{def}}{=}\psi_2$ so that we immediately get
\beq \Psi(\lambda)=\begin{pmatrix}\td{\Psi}_{1,1}(\lambda)& \td{\Psi}_{1,2}(\lambda)\\ \partial_\lambda \td{\Psi}_{1,1}(\lambda)& \partial_\lambda \td{\Psi}_{1,2}(\lambda) \end{pmatrix}=\begin{pmatrix}\psi_{1}(\lambda)& \psi_2(\lambda)\\ \partial_\lambda \psi_{1}(\lambda)& \partial_\lambda \psi_{2}(\lambda) \end{pmatrix}.\eeq
The companion-like system \eqref{CompanionMatrix} is equivalent to say that $\psi_1$ and $\psi_2$ satisfy the linear ODE:
\beq \left(\left[\partial_{\lambda}\right]^2 -L_{2,2}(\lambda)\partial_\lambda -L_{2,1}(\lambda)\right)\psi_i=0\eeq 
which is sometimes referred to as the ``quantum curve''.

\subsection{Introduction of a scaling parameter $\hbar$}\label{SectionIntrohbar}

In order to make the connection with formal $\hbar$-transseries appearing in the quantization of classical spectral curves via topological recursion of \cite{EO07}, we shall also introduce a formal $\hbar$ parameter by a simple rescaling of the irregular times.
\bea t_{\infty,k} &\to& \hbar^{\frac{k}{2}-1} t_{\infty,k} \,\,,\,\,  \forall \, k\in  \llbracket 1, 2r_\infty-2\rrbracket ,\cr
\lambda&\to& \hbar^{-1}\lambda 
\eea
This very simple rescaling implies that the differential system reads
\beq \hbar\partial_\lambda \td{\Psi}(\lambda,\hbar)=\td{L}(\lambda,\hbar)\td{\Psi}(\lambda,\hbar).\eeq
To preserve the normalized representative under this rescaling, also apply the constant diagonal gauge $D_\hbar=\operatorname{diag}(\hbar^g,1)$. Indeed, the old monic upper-right polynomial evaluated at $\hbar^{-1}\lambda$ has leading coefficient $\hbar^{-g}$, which this conjugation removes. If $q_i^{\mathrm{old}}$ denotes an old root, the new root is $q_i=\hbar q_i^{\mathrm{old}}$. After this normalization we reuse the notation $\td L$ and $\td\Psi$. For readers uneasy with the additional parameter, \textbf{$\hbar$ may be fixed to $1$ in the rest of the paper}, except when discussing formal $\hbar$ expansions.

\subsection{Explicit expressions of the gauge transformation}\label{SectionGaugeTransfo}
Using the Darboux coordinates $(q_i,p_i)_{1\leq i\leq g}$ and the irregular times $\mathbf{t}$, one may obtain the explicit expression of the gauge transformation relating $\Psi$ and $\td{\Psi}$. In order to do so, we shall introduce an intermediate wave matrix $\check{\Psi}$ for the following proposition.

\begin{proposition}\label{PropExplicitGaugeTransfo} The matrices $\td{\Psi}$ and $\Psi$ are related by the gauge transformations
\bea \label{GaugeTransfo} 
\td{\Psi}(\lambda,\hbar)&=&G_1(\lambda,\hbar)\check{\Psi}(\lambda,\hbar) \, \text{ with }\, \label{Gdef} G_1(\lambda,\hbar)=\begin{pmatrix} 1&0\\ \frac{1}{2}t_{\infty,2r_\infty-2}\lambda+ g_0
& 1\end{pmatrix}\cr
\check{\Psi}(\lambda,\hbar)&=&J(\lambda,\hbar) \Psi(\lambda,\hbar)\, \text{ with }\, J(\lambda,\hbar)=\begin{pmatrix}1 &0\\
\frac{Q(\lambda,\hbar)}{\underset{j=1}{\overset{g}{\prod}}(\lambda-q_j)}& \frac{1}{\underset{j=1}{\overset{g}{\prod}}(\lambda-q_j)}
\end{pmatrix}
\eea
where $Q$ is the unique polynomial in $\lambda$ of degree at most $g-1$ such that (with the convention that empty products are set to $1$; for $g=0$ we set $Q=0$) 
\beq Q(q_i,\hbar)=-p_i \,,\, \forall \, i\in \llbracket 1,g\rrbracket
\eeq
i.e.
\beq \label{DefQ2}  Q(\lambda,\hbar)= -\sum_{i=1}^g p_i \prod_{j\neq i}\frac{\lambda-q_j}{q_i-q_j}\eeq
and the coefficient $g_0$ is given by
\beq g_0=\frac{1}{2}t_{\infty,2r_\infty-4}+\frac{1}{2}t_{\infty,2r_\infty-2}\underset{j=1}{\overset{g}{\sum}} q_j\eeq
\end{proposition}

\begin{proof} The proof consists in observing that 
\beq \td{G}(\lambda)=\left(G_1(\lambda,\hbar)J(\lambda,\hbar)\right)^{-1}=\begin{pmatrix} 1&0\\ -Q(\lambda)-\left(\frac{1}{2}t_{\infty,2r_\infty-2}\lambda+g_0 \right) \underset{j=1}{\overset{g}{\prod}}(\lambda-q_j)&\underset{j=1}{\overset{g}{\prod}}(\lambda-q_j)\end{pmatrix}\eeq
recovers the matrix \eqref{GaugeGexpr}. Indeed, we first have that $\td{G}_{2,2}(\lambda)=\td{L}_{1,2}(\lambda)$ and by definition $\td{L}_{1,2}$ is a monic polynomial of degree $g=r_\infty-3$ with zeroes given by $(q_i)_{1\leq i\leq g}$. Similarly, the entry $\td{G}_{2,1}(\lambda)$ is polynomial in $\lambda$ of degree $g+1=r_\infty-2$. Moreover, it satisfies $\td{G}_{2,1}(q_i)=p_i$ for all $i\in \llbracket 1,g\rrbracket$ because of \eqref{Conditionpi}. Finally its leading coefficients at infinity are
\beq \td{G}_{2,1}(\lambda)=-\frac{1}{2}t_{\infty,2r_\infty-2}\lambda^{r_\infty-2}+\left(\frac{1}{2}t_{\infty,2r_\infty-2}\underset{j=1}{\overset{g}{\sum}} q_j-g_0\right)\lambda^{r_\infty-3}+O\left(\lambda^{r_\infty-4}\right)\eeq 
so that taking
\beq g_0=\frac{1}{2}t_{\infty,2r_\infty-4}+\frac{1}{2}t_{\infty,2r_\infty-2}\underset{j=1}{\overset{g}{\sum}} q_j\eeq
provides $\td{G}_{2,1}(\lambda)=-\frac{1}{2}t_{\infty,2r_\infty-2}\lambda^{r_\infty-2}-\frac{1}{2}t_{\infty,2r_\infty-4}\lambda^{r_\infty-3}+O\left(\lambda^{r_\infty-4}\right)$. Hence, with this choice of $g_0$, it is equal to $\td{L}_{1,1}(\lambda)$. Consequently, $\td{G}(\lambda)$ recovers the matrix $G(\lambda)$ of equation \eqref{GaugeGexpr} ending the proof. 
\end{proof}

\begin{remark}
By definition, the matrix $\check{\Psi}(\lambda,\hbar)$ satisfies the Lax system:
\bea \hbar \partial_\lambda \check{\Psi}(\lambda,\hbar)&=&\check{L}(\lambda,\hbar) \check{\Psi}(\lambda,\hbar)\cr
\hbar \partial_t \check{\Psi}(\lambda,\hbar)&=&\check{A}_{t}(\lambda,\hbar) \check{\Psi}(\lambda,\hbar)
\eea
for any irregular time $t\in \mathbf{t}$. In particular, the corresponding Lax matrices $\check{L}(\lambda,\hbar)$ and $\check{A}_{t}(\lambda,\hbar)$ are given by
\bea \check{L}(\lambda,\hbar)&=&J(\lambda,\hbar) L(\lambda,\hbar)J^{-1}(\lambda,\hbar) +\hbar (\partial_\lambda J(\lambda,\hbar)) J^{-1}(\lambda,\hbar)\cr
\check{A}_{t}(\lambda,\hbar)&=&J(\lambda,\hbar) A_{t}(\lambda,\hbar)J^{-1}(\lambda,\hbar) +\hbar \partial_t(J(\lambda,\hbar)) J^{-1}(\lambda,\hbar)
\eea
and are polynomial functions of $\lambda$ with no singularities at $\lambda\in\{q_1,\dots,q_g\}$.
\end{remark}

Note that by definition, the entries of $\check{L}$ are related to those of $L$ by
\bea \label{CheckLEquations}\check{L}_{1,1}(\lambda,\hbar)&=&-Q(\lambda,\hbar),\cr
\check{L}_{1,2}(\lambda,\hbar)&=&\underset{j=1}{\overset{g}{\prod}}(\lambda-q_j),\cr
\check{L}_{2,2}(\lambda,\hbar)&=&L_{2,2}(\lambda,\hbar)+Q(\lambda,\hbar)-\sum_{j=1}^g\frac{\hbar}{\lambda-q_j},\cr
\check{L}_{2,1}(\lambda,\hbar)&=& \frac{\hbar}{\underset{j=1}{\overset{g}{\prod}} (\lambda-q_j)}\frac{ \partial Q(\lambda,\hbar)}{\partial \lambda} +\frac{L_{2,1}(\lambda,\hbar) }{\underset{j=1}{\overset{g}{\prod}} (\lambda - q_j)} - L_{2,2}(\lambda,\hbar) \frac{Q(\lambda,\hbar)}{\underset{j=1}{\overset{g}{\prod}} (\lambda-q_j)} - \frac{Q(\lambda,\hbar)^2}{\underset{j=1}{\overset{g}{\prod}}(\lambda-q_j)}\cr 
&&
\eea 

Similarly, the entries of $\td{L}$ are related to those of $\check{L}$ by
\bea\label{TdLEquations} \td{L}_{1,1}(\lambda,\hbar)&=&\check{L}_{1,1}(\lambda,\hbar)-\left(\frac{1}{2}t_{\infty,2r_\infty-2}\lambda+g_0\right)\check{L}_{1,2}(\lambda,\hbar)\cr 
\td{L}_{1,2}(\lambda,\hbar)&=&\check{L}_{1,2}(\lambda,\hbar)\cr
\td{L}_{2,1}(\lambda,\hbar)&=&\check{L}_{2,1}(\lambda,\hbar)-\left(\frac{1}{2}t_{\infty,2r_\infty-2}\lambda+g_0\right)^2\check{L}_{1,2}(\lambda,\hbar)\cr
&&+\left(\frac{1}{2}t_{\infty,2r_\infty-2}\lambda+g_0\right)\left(\check{L}_{1,1}(\lambda,\hbar)-\check{L}_{2,2}(\lambda,\hbar)\right) +\frac{1}{2}\hbar t_{\infty,2r_\infty-2}\cr
 \td{L}_{2,2}(\lambda,\hbar)&=&\check{L}_{2,2}(\lambda,\hbar)+\left(\frac{1}{2}t_{\infty,2r_\infty-2}\lambda+g_0\right)\check{L}_{1,2}(\lambda,\hbar)
\eea

\subsection{Wronksians and asymptotics of the wave functions}
Combining the gauge transformations $G_\infty$, $G_1$ and $J$, we obtain the following proposition.

\begin{proposition}\label{PropPsiAsymp} The scalar wave functions $\psi_1=\Psi_{1,1}=\td{\Psi}_{1,1}=\check{\Psi}_{1,1}$ and $\psi_2=\Psi_{1,2}=\td{\Psi}_{1,2}=\check{\Psi}_{1,2}$ have the following expansions around $\infty$.
\bea \label{PsiAsymptotics0}\psi_1(\lambda)&\overset{\lambda\to \infty}{=}&\exp\left(-\frac{1}{\hbar}\sum_{k=1}^{2r_\infty-2} \frac{t_{\infty,k}}{k} \lambda^{\frac{k}{2}} - \frac{1}{4} \ln \lambda +O(1)\right), \cr
\psi_2(\lambda)&\overset{\lambda\to \infty}{=}&\exp\left(-\frac{1}{\hbar}\sum_{k=1}^{2r_\infty-2} (-1)^{k}\frac{t_{\infty,k}}{k} \lambda^{\frac{k}{2}} - \frac{1}{4}\ln \lambda  +O(1)\right).
\eea
\end{proposition}

\begin{proof}The proof is done in Appendix \ref{AppendixAsymptoticsWaveFunctions}.
\end{proof}

For convenience, we shall also define the Wronskians associated to the Lax systems and provide their explicit expressions that follow from the previous proposition:

\begin{definition}[Wronskians]\label{DefWronskian}Let us define $W(\lambda,\hbar)=\det \Psi(\lambda,\hbar)$, $\check{W}(\lambda,\hbar)=\det \check{\Psi}(\lambda,\hbar)$ and $\td{W}(\lambda,\hbar)=\det \td{\Psi}(\lambda,\hbar)$ the Wronskians associated to the corresponding wave matrices. They are given by 
\bea \td{W}(\lambda)&=&\td{W}_0 \exp\left(\frac{1}{\hbar}\int_0^\lambda \td{P}_1(s)ds \right),\cr
\check{W}(\lambda)&=&\td{W}(\lambda)=\td{W}_0 \exp\left(\frac{1}{\hbar}\int_0^\lambda \td{P}_1(s)ds \right)\cr
W(\lambda)&=&W_0\left(\underset{i=1}{\overset{g}{\prod}} (\lambda-q_i)\right)\exp\left(\frac{1}{\hbar}\int_0^\lambda \td{P}_1(s)ds  
\right). \eea 
where $W_0=\td W_0$ is independent of $\lambda$ (but may depend on deformation parameters). Indeed, the exact gauge relations imply $\check W=\td W$ and $W=(\underset{i=1}{\overset{g}{\prod}}(\lambda-q_i))\td W$. We have defined
\beq \label{P1Coeffs}\td{P}_{\infty,k}^{(1)}= -t_{\infty,2k+2} \,\,,\,\, \forall\, k\in\llbracket 0,r_\infty-2 \rrbracket,
\eeq
and regrouped them into the polynomial $\td{P}_1$:
\beq \label{DefP1}\td{P}_1(\lambda)=\sum_{j=0}^{r_\infty-2} \td{P}_{\infty,j}^{(1)}\lambda^j=-\sum_{j=0}^{r_\infty-2}t_{\infty,2j+2}\lambda^j.\eeq
\end{definition}

\begin{proof}The proof starts with $W(\lambda)=\hbar(\psi_1(\lambda)\partial_\lambda \psi_2(\lambda)- \psi_2(\lambda)\partial_\lambda\psi_1(\lambda))$. From the general construction and Proposition \ref{PropPsiAsymp}, $W(\lambda)\exp\left(-\frac{1}{\hbar}\int_0^\lambda \td{P}_1(s)ds\right)$ is a polynomial function of $\lambda$ of degree $r_\infty-3=g$. Moreover, since $L_{2,2}(\lambda)=\hbar \frac{\partial_\lambda W(\lambda)}{W(\lambda)}$, we get that the zeroes of $W(\lambda)$ are simple poles of $L_{2,2}(\lambda)$. Since $L_{2,2}(\lambda)$ may only have poles at infinity or at $(q_i)_{1\leq i\leq g}$, we get that $W(\lambda)\exp\left(-\frac{1}{\hbar}\int_0^\lambda \td{P}_1(s)ds\right)=W_0\underset{i=1}{\overset{g}{\prod}} (\lambda-q_i)$ for some constant $W_0$ (i.e. independent of $\lambda$).
Formulas for $\check{W}(\lambda)$ and $\td{W}(\lambda)$ follow from $W(\lambda)$ and the gauge transformations using the determinants of $G_1$ and $J$.
\end{proof}

\subsection{Explicit expression for the Lax matrix $L$}
In this section we shall provide an explicit expression for the matrix $L(\lambda)$ in terms of irregular times and Darboux coordinates. Only $g$ coefficients of the matrix shall remain undetermined at this stage. These coefficients will be put in one-to-one correspondence with the upcoming Hamiltonians. In order to write down the Lax matrix in a compact form, we shall introduce the following definition.

\begin{definition}\label{DefP2}
We include the quantum correction in the notation $\td P_2$ throughout the paper. Define
\bea
\td P_{\infty,k}^{(2)}
&=&\frac14\sum_{j=2k-2r_\infty+6}^{2r_\infty-2}
(-1)^j t_{\infty,j}t_{\infty,2k-j+4},
\qquad r_\infty-2\leq k\leq2r_\infty-4,\cr
\td P_{\infty,r_\infty-3}^{(2)}
&=&\frac14\sum_{j=1}^{2r_\infty-3}
(-1)^j t_{\infty,j}t_{\infty,2r_\infty-j-2}
+\frac{\hbar}{2}t_{\infty,2r_\infty-2}.
\eea
Set
\beq
\td P_2(\lambda,\hbar)
=\sum_{k=r_\infty-3}^{2r_\infty-4}\td P_{\infty,k}^{(2)}\lambda^k,
\qquad
\hat P_2(\lambda,\hbar)
=\td P_2(\lambda,\hbar)-\sum_{k=0}^{r_\infty-4}H_{\infty,k}\lambda^k.
\eeq
The $g=r_\infty-3$ coefficients $H_{\infty,k}$ are determined below by interpolation. We write $\td P_2^{(0)}$ for the time polynomial obtained by omitting the explicit $\hbar$ term; thus
\beq
\td P_2(\lambda,\hbar)=\td P_2^{(0)}(\lambda)
+\frac{\hbar}{2}t_{\infty,2r_\infty-2}\lambda^g.
\eeq
Arguments $\hbar$ are suppressed when no ambiguity is possible. In particular, every subsequent occurrence of $\td P_2$ or $\td P_{\infty,k}^{(2)}$ in the general companion, interpolation and Hamiltonian formulas uses this corrected definition. On the canonical slice the correction vanishes.
\end{definition}

Using the previous definition, we obtain the following proposition.

\begin{proposition}\label{PropLaxMatrix} The Lax matrix $L(\lambda,\hbar)$ is given by
\beq L(\lambda,\hbar)=\begin{pmatrix} 0&1\\ L_{2,1}(\lambda,\hbar)&L_{2,2}(\lambda,\hbar)\end{pmatrix}\eeq
with
\bea L_{2,2}(\lambda,\hbar)&=&\td{P}_1(\lambda)+ \sum_{j=1}^g \frac{\hbar}{\lambda-q_j}\cr
L_{2,1}(\lambda,\hbar)&=&-\td{P}_2(\lambda) +\sum_{k=0}^{r_\infty-4} H_{\infty,k}\lambda^k - \sum_{j=1}^g\frac{\hbar p_j}{\lambda-q_j}
\eea
Coefficients $\left(H_{\infty,k}\right)_{0\leq k\leq r_\infty-4}$ shall be determined later in Proposition \ref{PropDefCi2}.
\end{proposition}

\begin{proof}The proof is based on the fact that the entries of $L$ are rational functions of $\lambda$ with poles only at $\infty$ or at apparent singularities $(q_i)_{1\leq i\leq g}$. Using the knowledge of the asymptotics expansion at $\infty$ provides the result. This is detailed in Appendix \ref{AppendixLForm}. 
\end{proof}

\section{Classical spectral curve and connection with topological recursion}\label{SectionTR}
Before turning to deformations relatively to the irregular times, let us briefly mention the connection of the present setup with the classical spectral curve and the topological recursion. This section being independent of the others, we stress that readers with no interest in topological recursion or in WKB expansions may skip the content of this section.  

\medskip

Suppose the Lax matrix has a formal leading coefficient $L^{(0)}(\lambda)$ in $\hbar$. We define the classical spectral curve by
\beq
\det(yI_2-L^{(0)}(\lambda))=0.
\eeq
An analytic limit may be used when it exists; a general transseries need not admit such a limit. A gauge $G=G^{(0)}+O(\hbar)$ with $G^{(0)}$ invertible preserves this curve on the common domain of definition, since its leading action is conjugation and $\hbar(\partial_\lambda G)G^{-1}=O(\hbar)$. The gauges relating $L$, $\check L$ and $\td L$ therefore give the same leading characteristic polynomial away from their apparent singularities.
\beq\label{ClassicalSpectralCurve}
y^2-\td P_1(\lambda)\big|_{\hbar=0}\,y
+\hat P_2(\lambda,\hbar)\big|_{\hbar=0}=0.
\eeq
Its generic smooth member has genus $g=r_\infty-3$; singular members can have smaller normalization genus. After choosing cycles, the remaining $g$ coefficients may locally be parametrized by filling fractions. Infinity is a ramification point on the generic ramified stratum. The two classical branches satisfy
\bea
y_1(z)&=&-\frac12\sum_{k=1}^{2r_\infty-2}t_{\infty,k}
x(z)^{k/2-1}+O(x(z)^{-3/2}),\cr
y_2(z)&=&-\frac12\sum_{k=1}^{2r_\infty-2}(-1)^k t_{\infty,k}
x(z)^{k/2-1}+O(x(z)^{-3/2}).
\eea
The term $-\hbar/(4x)$ belongs to the quantum logarithmic derivatives $\hbar\partial_x\log\psi_i$, not to these classical branches.

The constructions of \cite{MOsl2,Quantization_2021} associate formal wave functions to suitable spectral-curve data by topological recursion. In positive genus, these data include the cycle and Bergman-kernel normalization, integration data and nonperturbative parameters. Identifying the resulting quantum system with a prescribed Lax system requires matching these choices and the full quantum coefficients; equality of the classical curves alone is insufficient. Whenever this identification has been made, write
\beq
\Psi^{\mathrm{TR}}=
\begin{pmatrix}
\psi_1^{\mathrm{TR}}&\psi_2^{\mathrm{TR}}\\
\hbar\partial_\lambda\psi_1^{\mathrm{TR}}&
\hbar\partial_\lambda\psi_2^{\mathrm{TR}}
\end{pmatrix}.
\eeq
Two invertible fundamental matrices of the same left differential system obey
\beq
\Psi(\lambda,\hbar)=\Psi^{\mathrm{TR}}(\lambda,\hbar)C.
\eeq
Here $C$ is independent of $\lambda$ and, for the same compatible deformation system, of all deformation times; it may depend on $\hbar$. These are formal reconstruction statements. An analytic interpretation additionally requires control of summation and Stokes data.

\section{General isomonodromic deformations and auxiliary matrices}\label{SectionAuxi}
\subsection{Definition of general isomonodromic deformations}
The previous sections provide a natural set of parameters for which we may consider deformations, namely the irregular times $\left(t_{\infty,k}\right)_{1\leq k\leq 2r_\infty-2}$. In order to study deformations relatively to these parameters we introduce the following definition.

\begin{definition}\label{DefGeneralDeformationsDefinition} We define the following general deformation operators.
\beq \label{GeneralDeformationsDefinition}\mathcal{L}_{\boldsymbol{\alpha}}=\hbar \sum_{k=1}^{2r_\infty-2} \alpha_{\infty,k} \partial_{t_{\infty,k}}\eeq
where we define the vector $\boldsymbol{\alpha}\in \mathbb{C}^{2r_{\infty}-2}=\mathbb{C}^{2g+4}$ by
\beq \boldsymbol{\alpha}= \sum_{k=1}^{2r_\infty-2} \alpha_{\infty,k}\mathbf{e}_{k}.\eeq
\end{definition}

Deformations defined by Definition \ref{DefGeneralDeformationsDefinition} shall be seen as general isomonodromic deformations in $\hat{F}_{\infty,r_\infty}$. 

\medskip
 
Associated to a vector $\boldsymbol{\alpha}$ are general auxiliary Lax matrices $\td{A}_{\boldsymbol{\alpha}}(\lambda)$, $\check{A}_{\boldsymbol{\alpha}}(\lambda)$  and $A_{\boldsymbol{\alpha}}(\lambda)$ defined by
\bea \td{A}_{\boldsymbol{\alpha}}(\lambda)=\mathcal{L}_{\boldsymbol{\alpha}}[\td{\Psi}(\lambda)] \td{\Psi}^{-1}(\lambda)
 \,\, &\Leftrightarrow&\,\, \mathcal{L}_{\boldsymbol{\alpha}}[\td{\Psi}(\lambda)] =\td{A}_{\boldsymbol{\alpha}}(\lambda) \td{\Psi}(\lambda)\cr
\check{A}_{\boldsymbol{\alpha}}(\lambda)=\mathcal{L}_{\boldsymbol{\alpha}}[\check{\Psi}(\lambda)] \check{\Psi}^{-1}(\lambda)
 \,\, &\Leftrightarrow&\,\, \mathcal{L}_{\boldsymbol{\alpha}}[\check{\Psi}(\lambda)] =\check{A}_{\boldsymbol{\alpha}}(\lambda) \check{\Psi}(\lambda)\cr
A_{\boldsymbol{\alpha}}(\lambda)=\mathcal{L}_{\boldsymbol{\alpha}}[\Psi(\lambda)] \Psi^{-1}(\lambda) \,\, &\Leftrightarrow&\,\, \mathcal{L}_{\boldsymbol{\alpha}}[\Psi(\lambda)] =A_{\boldsymbol{\alpha}}(\lambda) \Psi(\lambda)
\eea

In particular, $\td{A}_{\boldsymbol{\alpha}}(\lambda)$ and $\check{A}_{\boldsymbol{\alpha}}(\lambda)$ are polynomial functions of $\lambda$ while $A(\lambda)$ may also have additional poles at $\{q_1,\dots,q_g\}$. Note that $\left(L(\lambda),A_{\boldsymbol{\alpha}}(\lambda)\right)$, $\left(\check{L}(\lambda),\check{A}_{\boldsymbol{\alpha}}(\lambda)\right)$ and $\left(\td{L}(\lambda),\td{A}_{\boldsymbol{\alpha}}(\lambda)\right)$ provide equivalent Lax pairs but expressed in three different gauges. The corresponding compatibility equations are

\bea \label{CompatibilityEquation}\mathcal{L}_{\boldsymbol{\alpha}}[L]&=&[A_{\boldsymbol{\alpha}},L]+\hbar\partial_\lambda A_{\boldsymbol{\alpha}}\cr
\mathcal{L}_{\boldsymbol{\alpha}}[\check{L}]&=&[\check{A}_{\boldsymbol{\alpha}},\check{L}]+\hbar\partial_\lambda \check{A}_{\boldsymbol{\alpha}}\cr
\mathcal{L}_{\boldsymbol{\alpha}}[\td{L}]&=&[\td{A}_{\boldsymbol{\alpha}},\td{L}]+\hbar\partial_\lambda \td{A}_{\boldsymbol{\alpha}}.
\eea

We shall now use the asymptotic expansions of the wave matrices in order to obtain information on the general form of the auxiliary matrices. Then, we shall use the compatibility equations in order to determine the evolutions of the Darboux coordinates under general isomonodromic deformations and prove that these evolutions are Hamiltonian as performed in a similar way for the untwisted case in \cite{MarchalOrantinAlameddine2022}. 

\subsection{General form of the auxiliary matrix $A_{\boldsymbol{\alpha}}(\lambda,\hbar)$}

Using compatibility conditions one may easily obtain two of the entries of $A_{\boldsymbol{\alpha}}(\lambda)$. Indeed, since $L$ is a companion-like matrix, compatibility equations \eqref{CompatibilityEquation} imply that
\bea \label{TrivialEntriesA}\left[A_{\boldsymbol{\alpha}}(\lambda)\right]_{2,1}&=&\hbar \partial_{\lambda} \left[A_{\boldsymbol{\alpha}}(\lambda)\right]_{1,1}+\left[A_{\boldsymbol{\alpha}}(\lambda)\right]_{1,2}L_{2,1}(\lambda),\cr
\left[A_{\boldsymbol{\alpha}}(\lambda)\right]_{2,2}&=&\hbar \partial_{\lambda} \left[A_{\boldsymbol{\alpha}}(\lambda)\right]_{1,2}+\left[A_{\boldsymbol{\alpha}}(\lambda)\right]_{1,1}+\left[A_{\boldsymbol{\alpha}}(\lambda)\right]_{1,2}L_{2,2}(\lambda),
\eea
so that only the first line of $A_{\boldsymbol{\alpha}}(\lambda)$ remains unknown at this stage. The other two entries of the compatibility equation \eqref{CompatibilityEquation} leads to
\bea \label{Compat}\mathcal{L}_{\boldsymbol{\alpha}}[L_{2,1}(\lambda)]&=&\hbar^2 \frac{\partial^2 \left[A_{\boldsymbol{\alpha}}(\lambda)\right]_{1,1}}{\partial \lambda^2} +2\hbar L_{2,1}(\lambda)\, \partial_\lambda \left[A_{\boldsymbol{\alpha}}(\lambda)\right]_{1,2}+\hbar \left[A_{\boldsymbol{\alpha}}(\lambda)\right]_{1,2} \, \partial_\lambda L_{2,1}(\lambda)\cr
&&- \hbar L_{2,2}(\lambda)\, \partial_{\lambda} \left[A_{\boldsymbol{\alpha}}(\lambda)\right]_{1,1},\cr
\mathcal{L}_{\boldsymbol{\alpha}}[L_{2,2}(\lambda)]&=&\hbar^2 \frac{\partial^2 \left[A_{\boldsymbol{\alpha}}(\lambda)\right]_{1,2}}{\partial \lambda^2} +2\hbar \partial_\lambda \left[A_{\boldsymbol{\alpha}}(\lambda)\right]_{1,1}+\hbar L_{2,2}(\lambda)\,\partial_\lambda \left[A_{\boldsymbol{\alpha}}(\lambda)\right]_{1,2}\cr
&&+ \hbar \left[A_{\boldsymbol{\alpha}}(\lambda)\right]_{1,2}\,\partial_{\lambda}L_{2,2}(\lambda)
\eea
that shall be used later to determine the evolution equations for $\left(q_i,p_i\right)_{1\leq i\leq g}$. Before studying the compatibility equations, let us observe that the asymptotic expansions of the wave matrix $\Psi$ at infinity allows to determine the general form of the auxiliary matrix $A_{\boldsymbol{\alpha}}(\lambda,\hbar)$. This leads us to the following results.

\begin{proposition}\label{PropAsymptoticExpansionA12} The asymptotic expansion of entry $\left[A_{\boldsymbol{\alpha}}(\lambda)\right]_{1,2}$ at infinity is given by
\beq \label{A12Asymptoticsinf}\forall \,M\geq 1\,:\, \left[A_{\boldsymbol{\alpha}}(\lambda)\right]_{1,2}\overset{\lambda\to \infty}{=}\sum_{i=-1}^{M} \frac{\nu^{(\boldsymbol{\alpha})}_{\infty,i}}{\lambda^i} +O\left(\lambda^{-M-1}\right).
\eeq
Moreover, coefficients $\left(\nu^{(\boldsymbol{\alpha})}_{\infty,k}\right)_{-1\leq k\leq r_\infty-3}$ are determined by 
\beq \label{RelationNuAlphaInfty} M_\infty\begin{pmatrix} \nu^{(\boldsymbol{\alpha})}_{\infty,-1}\\ \nu^{(\boldsymbol{\alpha})}_{\infty,0}\\  \vdots \\ \nu^{(\boldsymbol{\alpha})}_{\infty,r_\infty-3}\end{pmatrix}=\begin{pmatrix}\frac{2\alpha_{\infty,2r_\infty-3}}{(2r_\infty-3)}\\ \frac{2\alpha_{\infty,2r_\infty-5}}{(2r_\infty-5)}\\ \vdots \\ \frac{2\alpha_{\infty,1}}{1} \end{pmatrix}\eeq
where $M_\infty$ is a lower triangular Toeplitz matrix of size $(r_\infty-1)\times(r_\infty-1)$ independent of the deformation $\boldsymbol{\alpha}$:
\beq\label{MatrixMInfty} M_\infty=\begin{pmatrix}t_{\infty,2r_\infty-3}&0&\dots& &0\\
t_{\infty,2r_\infty-5}& t_{\infty,2r_\infty-3}&0& \ddots& \vdots\\
\vdots& \ddots& \ddots&&\vdots \\
t_{\infty,3}& & \ddots &\ddots&0 \\
t_{\infty,1}& t_{\infty,3} & \dots & &t_{\infty,2r_\infty-3}\end{pmatrix}.
\eeq
\end{proposition}

\begin{proof}The proof is presented in Appendix \ref{AppendixExpansionA}.
\end{proof}

The previous proposition may be used to determine the general form of the entry $\left[A_{\boldsymbol{\alpha}}(\lambda)\right]_{1,2}$.

\begin{proposition}\label{PropA12Form} Entry $\left[A_{\boldsymbol{\alpha}}(\lambda)\right]_{1,2}$ is given by
\beq \label{ExpressionA12} \left[A_{\boldsymbol{\alpha}}(\lambda)\right]_{1,2}=\nu^{(\boldsymbol{\alpha})}_{\infty,-1}\lambda+\nu^{(\boldsymbol{\alpha})}_{\infty,0} + \sum_{j=1}^g \frac{\mu^{(\boldsymbol{\alpha})}_j}{\lambda-q_j}.\eeq
Coefficients $\left(\mu^{(\boldsymbol{\alpha})}_j\right)_{1\leq j\leq g}$ are determined by the linear system
\beq \label{RelationNuMuMatrixForm} V_\infty\begin{pmatrix}\mu^{(\boldsymbol{\alpha})}_1\\ \vdots\\ \vdots\\\mu^{(\boldsymbol{\alpha})}_g\end{pmatrix}= \begin{pmatrix}\nu^{(\boldsymbol{\alpha})}_{\infty,1}\\ \nu^{(\boldsymbol{\alpha})}_{\infty,2}\\\vdots \\ \nu^{(\boldsymbol{\alpha})}_{\infty,r_\infty-3}\end{pmatrix}\eeq
where $V_\infty$ is a $(r_{\infty}-3)\times g$ matrix 
\beq \label{DefVinfty}V_\infty=\begin{pmatrix}1&1 &\dots &\dots &1\\
q_1& q_2&\dots &\dots& q_{g}\\
\vdots & & & & \vdots\\
\vdots & & & & \vdots\\
q_1^{r_\infty-4}& q_2^{r_\infty-4} &\dots & \dots& q_{g}^{r_\infty-4}\end{pmatrix}\,,\, 
\eeq
\end{proposition}

\begin{proof}We know that $\left[A_{\boldsymbol{\alpha}}(\lambda)\right]_{1,2}$ is rational in $\lambda$ with only simple poles in $\{q_1,\dots,q_g\}$ and a pole at infinity. Proposition \ref{PropAsymptoticExpansionA12} provides the asymptotics at infinity so that \eqref{ExpressionA12} holds. Moreover the expansion at infinity of 
\beq \sum_{j=1}^g \frac{\mu^{(\boldsymbol{\alpha})}_j}{\lambda-q_j}=\sum_{k=1}^{\infty}\sum_{j=1}^g \mu^{(\boldsymbol{\alpha})}_j q_j^{k-1} \lambda^{-k}\eeq
identifies with \eqref{A12Asymptoticsinf} only with \eqref{RelationNuMuMatrixForm}.
\end{proof}

Note that we may determine coefficients $\left(\nu^{(\boldsymbol{\alpha})}_{\infty,k}\right)_{k\geq r_\infty-2}$ by the fact that
\bea \sum_{j=1}^g\mu_j^{(\boldsymbol{\alpha})}\prod_{i\neq j}(\lambda-q_i)&=&\left([A_{\boldsymbol{\alpha}}(\lambda)]_{1,2} -\nu^{(\boldsymbol{\alpha})}_{\infty,-1}\lambda -\nu^{(\boldsymbol{\alpha})}_{\infty,0}\right)\left(\underset{i=1}{\overset{g}{\prod}}(\lambda-q_i)\right)\cr
&=&\left(\sum_{k=1}^{\infty}\nu^{(\boldsymbol{\alpha})}_{\infty,k}\lambda^{-k}\right)\left(\prod_{i=1}^g(\lambda-q_i)\right)\cr
&=&\left(\sum_{k=1}^{\infty}\nu^{(\boldsymbol{\alpha})}_{\infty,k}\lambda^{-k}\right)\left(\sum_{i=0}^g(-1)^{g-i}e_{g-i}(\{q_1,\dots,q_g\}) \lambda^i\right)
\eea
where the l.h.s. is a polynomial in $\lambda$  and $\left(e_k(\{q_1,\dots,q_g\})\right)_{0\leq k\leq g}$ are the elementary symmetric polynomials. Thus, for all $m\geq 1$:
\small{\beq 0=\sum_{k=m}^{g+m}(-1)^{g+m-k}\nu^{(\boldsymbol{\alpha})}_{\infty,k}e_{g+m-k}(\{q_1,\dots,q_g\})=\nu^{(\boldsymbol{\alpha})}_{\infty,g+m}+\sum_{k=m}^{g+m-1}(-1)^{g+m-k}\nu^{(\boldsymbol{\alpha})}_{\infty,k}e_{g+m-k}(\{q_1,\dots,q_g\}) \eeq}
\normalsize{so} that we obtain the recursive relations
\beq \label{recursivenus} \forall\, m\geq 1\,:\, \nu^{(\boldsymbol{\alpha})}_{\infty,g+m}=\sum_{k=m}^{g+m-1}(-1)^{g+m-1-k}\nu^{(\boldsymbol{\alpha})}_{\infty,k}e_{g+m-k}(\{q_1,\dots,q_g\})\eeq
In particular, we get for $m=1$:
\beq \label{nuinftyrinftyminus2} \nu^{(\boldsymbol{\alpha})}_{\infty,r_\infty-2}=\sum_{k=1}^{g}(-1)^{g-k}\nu^{(\boldsymbol{\alpha})}_{\infty,k}e_{g+1-k}(\{q_1,\dots,q_g\})\eeq

Let us now perform similar computation for $\left[A_{\boldsymbol{\alpha}}(\lambda)\right]_{1,1}$. We obtain the following proposition.

\begin{proposition}\label{Propcalpha}The entry $\left[A_{\boldsymbol{\alpha}}(\lambda)\right]_{1,1}$ is given by
\beq \label{ExpressionA11} \left[A_{\boldsymbol{\alpha}}(\lambda)\right]_{1,1}=\sum_{i=0}^{r_\infty-1}c^{(\boldsymbol{\alpha})}_{\infty,i}\lambda^i+\sum_{j=1}^g\frac{\rho^{(\boldsymbol{\alpha})}_j}{\lambda-q_j}.\eeq
with 
\beq \forall\, j\in \llbracket 1,g\rrbracket \,:\, \rho^{(\boldsymbol{\alpha})}_j=-\mu^{(\boldsymbol{\alpha})}_j p_j\eeq
Coefficients $\left(c^{(\boldsymbol{\alpha})}_{\infty,k}\right)_{1\leq k\leq r_\infty-1}$  are determined by
\beq \label{calphaexpr} M_\infty \begin{pmatrix} c^{(\boldsymbol{\alpha})}_{\infty,r_\infty-1}\\ \vdots\\c^{(\boldsymbol{\alpha})}_{\infty,k}  \\\vdots \\ c^{(\boldsymbol{\alpha})}_{\infty, 1}\end{pmatrix}=\begin{pmatrix}\frac{\alpha_{\infty,2r_\infty-3}}{2r_\infty-3}t_{\infty,2r_\infty-2}-\frac{\alpha_{\infty,2r_\infty-2}}{2r_\infty-2}t_{\infty,2r_\infty-3}\\ \vdots\\ \underset{m=k}{\overset{r_\infty-1}{\sum}}\left(\frac{\alpha_{\infty,2k+2r_\infty-2m-3}}{2k+2r_\infty-2m-3}t_{\infty,2m}-\frac{\alpha_{\infty,2k+2r_\infty-2m-2}}{2k+2r_\infty-2m-2}t_{\infty,2m-1}\right)  \\ \vdots\\ \underset{m=1}{\overset{r_\infty-1}{\sum}}\left(\frac{\alpha_{\infty,2r_\infty-2m-1}}{2r_\infty-2m-1}t_{\infty,2m}-\frac{\alpha_{\infty,2r_\infty-2m}}{2r_\infty-2m}t_{\infty,2m-1}\right) \end{pmatrix}
\eeq
with the matrix $M_\infty$ given by \eqref{MatrixMInfty}.
\end{proposition}

\begin{proof}The proof is done in Appendix \ref{AppendixA11}. 
\end{proof}

Note that $c^{(\boldsymbol{\alpha})}_{\infty,0}$ is not determined but will play no role in the rest of the paper. In the previous propositions, one may easily observe that $M_\infty$, $\left(\nu^{(\boldsymbol{\alpha})}_{\infty,k}\right)_{-1\leq k\leq r_\infty-3}$ and $\left(c^{(\boldsymbol{\alpha})}_{\infty,k}\right)_{1\leq k\leq r_\infty-1}$ are  independent of the Darboux coordinates and depend only on irregular times and the deformation $\boldsymbol{\alpha}$. On the contrary, $\left(\mu^{(\boldsymbol{\alpha})}_j\right)_{1\leq j\leq g}$ and $V_\infty$ depend on the Darboux coordinates.

\section{General Hamiltonian evolutions}
The previous sections provide the general form of the matrices $L(\lambda,\hbar)$ and $A_{\boldsymbol{\alpha}}(\lambda)$ through Propositions \ref{PropLaxMatrix}, \ref{PropAsymptoticExpansionA12}, \ref{PropA12Form}, \ref{Propcalpha} and equation \eqref{TrivialEntriesA}. As we shall see below, inserting the previous knowledge into the compatibility equations \eqref{Compat} provides the evolutions of the Darboux coordinates.

\medskip 
The first step is to look at order $(\lambda-q_j)^{-2}$ in $\mathcal{L}_{\boldsymbol{\alpha}}[L_{2,2}(\lambda)]$. We obtain, for all $j\in \llbracket 1, g\rrbracket$:
\beq \label{Lqj}\mathcal{L}_{\boldsymbol{\alpha}}[q_j]=2\mu^{(\boldsymbol{\alpha})}_j\left(p_j -\frac{1}{2}\td{P}_1(q_j)\right)-\hbar \nu^{(\boldsymbol{\alpha})}_{\infty,0} -\hbar \nu^{(\boldsymbol{\alpha})}_{\infty,-1} q_j-\hbar \sum_{i\neq j}\frac{\mu^{(\boldsymbol{\alpha})}_j+\mu^{(\boldsymbol{\alpha})}_i}{q_j-q_i}.\eeq

The next step is to determine the coefficients $\left(H_{\infty,j}\right)_{0\leq j\leq r_\infty-4}$ that remain unknown in $L_{2,1}(\lambda)$. To achieve this task, we look at order $(\lambda-q_j)^{-2}$ in $\mathcal{L}_{\boldsymbol{\alpha}}[L_{2,1}(\lambda)]$ using \eqref{Compat}. We obtain

\footnotesize{
\bea -\hbar p_j \mathcal{L}_{\boldsymbol{\alpha}}[q_j]&=&-2\hbar \mu^{(\boldsymbol{\alpha})}_j\left(
-\td{P}_2(q_j) +\sum_{k=0}^{r_\infty-4}H_{\infty,k}q_j^k-\sum_{i\neq j} \frac{\hbar p_i}{q_j-q_i}\right)\cr
&&+\hbar^2 p_j\left( \nu^{(\boldsymbol{\alpha})}_{\infty,-1} q_j+\nu^{(\boldsymbol{\alpha})}_{\infty,0}+\sum_{i\neq j} \frac{\mu^{(\boldsymbol{\alpha})}_i}{q_j-q_i}\right)-\hbar \mu^{(\boldsymbol{\alpha})}_j p_j\left(\td{P}_1(q_j)+\sum_{i\neq j} \frac{\hbar}{q_j-q_i}\right).\cr 
&&
\eea}
\normalsize{}
Inserting \eqref{Lqj} first gives the following identity multiplied by $\mu_j^{(\boldsymbol\alpha)}$. For each fixed $j$, one may choose a deformation with $\mu_j^{(\boldsymbol\alpha)}\neq0$: the triangular system at infinity and the Vandermonde system are invertible on the chosen chart. Since the resulting expression is independent of the deformation direction, this proves, for every $j\in\llbracket1,g\rrbracket$,
\beq \label{DefCi}\sum_{k=0}^{r_\infty-4}H_{\infty,k}q_j^k=p_j^2 -\td{P}_1(q_j)p_j +\td{P}_2(q_j)+\hbar \sum_{i\neq j}\frac{p_i-p_j}{q_j-q_i}
\eeq
where it is obvious that the r.h.s. is independent of the deformation vector $\boldsymbol{\alpha}$. The last relation can be rewritten into a matrix form.

\begin{proposition}\label{PropDefCi2} We have 
\beq \label{DefCi2}
(V_\infty)^{t}\begin{pmatrix}H_{\infty,0}\\ \vdots\\ H_{\infty,r_\infty-4} \end{pmatrix}=\begin{pmatrix} p_1^2- \td{P}_1(q_1)p_1 +\td{P}_2(q_1)+\hbar \underset{i\neq 1}{\sum}\frac{p_i-p_1}{q_1-q_i}\\
\vdots\\
p_g^2- \td{P}_1(q_g)p_g+\td{P}_2(q_g)+\hbar \underset{i\neq g}{\sum}\frac{p_i-p_g}{q_g-q_i}
\end{pmatrix}
\eeq
\end{proposition}

Finally, in order to obtain the evolution equation for $(p_j)_{1\leq j\leq g}$ we look at order $(\lambda-q_j)^{-1}$ of the entry $\mathcal{L}_{\boldsymbol{\alpha}}[L_{2,1}(\lambda)]$. We get, for all $j\in \llbracket 1,g\rrbracket$,
\bea \label{Lpj} \mathcal{L}_{\boldsymbol{\alpha}}[p_j]&=&\hbar \sum_{i\neq j}\frac{(\mu^{(\boldsymbol{\alpha})}_i+\mu^{(\boldsymbol{\alpha})}_j)(p_i-p_j)}{(q_j-q_i)^2} +\mu^{(\boldsymbol{\alpha})}_j\left(p_j \td{P}_1'(q_j)-\td{P}_2'(q_j)+\sum_{k=1}^{r_\infty-4}kH_{\infty,k}q_j^{k-1}\right) \cr
&&+\hbar \nu^{(\boldsymbol{\alpha})}_{\infty,-1}p_j+\hbar \sum_{k=1}^{r_\infty-1}kc^{(\boldsymbol{\alpha})}_{\infty,k}q_j^{k-1}.
\eea

Thus, we have obtained the general evolutions for $\left(p_j,q_j\right)_{1\leq j\leq g}$ through \eqref{Lqj} and \eqref{Lpj}. We may now formulate our first main Theorem showing that these evolutions are Hamiltonian. 

\begin{theorem}[Hamiltonian evolution] \label{HamTheorem} Defining 
\beq \label{DefHam}\text{Ham}^{(\boldsymbol{\alpha})}(\mathbf{q},\mathbf{p}) =\sum_{k=0}^{r_\infty-4} \nu_{\infty,k+1}^{\boldsymbol{(\alpha)}}H_{\infty,k}-\hbar \sum_{j=1}^g\sum_{k=1}^{r_\infty-1}c^{(\boldsymbol{\alpha})}_{\infty,k}q_j^{k}-\hbar \nu^{(\boldsymbol{\alpha})}_{\infty,0}\sum_{j=1}^{g} p_j-\hbar \nu^{(\boldsymbol{\alpha})}_{\infty,-1}\sum_{j=1}^g q_jp_j,
\eeq
the evolutions for $j\in \llbracket 1,g\rrbracket$,
\small{\bea 
\mathcal{L}_{\boldsymbol{\alpha}}[q_j]&=&2\mu^{(\boldsymbol{\alpha})}_j\left(p_j -\frac{1}{2}\td{P}_1(q_j)\right)-\hbar \nu^{(\boldsymbol{\alpha})}_{\infty,0} -\hbar \nu^{(\boldsymbol{\alpha})}_{\infty,-1} q_j-\hbar \sum_{i\neq j}\frac{\mu^{(\boldsymbol{\alpha})}_j+\mu^{(\boldsymbol{\alpha})}_i}{q_j-q_i},\cr
\mathcal{L}_{\boldsymbol{\alpha}}[p_j]&=&\hbar \sum_{i\neq j}\frac{(\mu^{(\boldsymbol{\alpha})}_i+\mu^{(\boldsymbol{\alpha})}_j)(p_i-p_j)}{(q_j-q_i)^2} +\mu^{(\boldsymbol{\alpha})}_j\left(p_j \td{P}_1'(q_j)-\td{P}_2'(q_j)+\sum_{k=1}^{r_\infty-4}kH_{\infty,k}q_j^{k-1}\right)\cr
&&+\hbar \nu^{(\boldsymbol{\alpha})}_{\infty,-1}p_j+\hbar \sum_{k=1}^{r_\infty-1}kc^{(\boldsymbol{\alpha})}_{\infty,k}q_j^{k-1}
\eea}
\normalsize{are} Hamiltonian in the sense that 
\beq \forall\, j\in \llbracket 1,g\rrbracket \,:\, \mathcal{L}_{\boldsymbol{\alpha}}[q_j]=\frac{\partial \text{Ham}^{(\boldsymbol{\alpha})}(\mathbf{q},\mathbf{p})}{\partial p_j}\, \text{ and }\, \mathcal{L}_{\boldsymbol{\alpha}}[p_j]=-\frac{\partial \text{Ham}^{(\boldsymbol{\alpha})}(\mathbf{q},\mathbf{p})}{\partial q_j}.\eeq
Quantities involved in the Hamiltonian evolution are defined by Propositions \ref{PropAsymptoticExpansionA12}, \ref{PropA12Form}, \ref{Propcalpha} and \ref{PropDefCi2}.
\end{theorem}

\begin{proof}Proof is done in Appendix \ref{AppendixHamiltonian}.\end{proof}

\begin{remark}Note that there is an alternative expression for the Hamiltonian \eqref{DefHam}:
\bea\text{Ham}^{(\boldsymbol{\alpha})}(\mathbf{q},\mathbf{p})&=&-\frac{\hbar}{2}\displaystyle{\sum_{\substack{(i,j)\in \llbracket 1,g\rrbracket^2 \\ i\neq j }}} \frac{(\mu^{(\boldsymbol{\alpha})}_i+\mu^{(\boldsymbol{\alpha})}_j)(p_i-p_j)}{q_i-q_j} -\hbar \sum_{j=1}^{g} (\nu^{(\boldsymbol{\alpha})}_{\infty,0} p_j+\nu^{(\boldsymbol{\alpha})}_{\infty,-1}q_jp_j) \cr
&&+\sum_{j=1}^{g}\mu^{(\boldsymbol{\alpha})}_j\left[p_j^2-\td{P}_1(q_j)p_j +\td{P}_2(q_j)\right]-\hbar \sum_{j=1}^g\sum_{k=1}^{r_\infty-1}c^{(\boldsymbol{\alpha})}_{\infty,k}q_j^{k}\cr
&&
\eea
\end{remark}

Theorem \ref{HamTheorem} shows that the Hamiltonian expression for a general isomonodromic deformation may be split into several contributions
\begin{itemize}\item A linear combination of the $\left(H_{\infty,k}\right)_{0\leq k\leq r_\infty-4}$ whose coefficients are given by $\left(\nu_{\infty,k+1}^{\boldsymbol{(\alpha)}}\right)_{0\leq k\leq r_\infty-4}$. Note that the coefficients $\left(H_{\infty,k}\right)_{0\leq k\leq r_\infty-4}$ do not depend on the choice of deformation direction (they can vary along its flow) and correspond to the unknown coefficients of $L_{2,1}(\lambda)$.
\item A linear combination of $\left(\underset{j=1}{\overset{g}{\sum}}q_j^{k}\right)_{1\leq k\leq r_\infty-1}$ whose coefficients are given by $\left(c^{(\boldsymbol{\alpha})}_{\infty,k}\right)_{1\leq k\leq r_\infty-1}$. As we will see in the next sections, these terms will vanish for a suitable choice of non-trivial isomonodromic deformations.
\item Two additional terms $\underset{j=1}{\overset{g}{\sum}} p_j$ and $\underset{j=1}{\overset{g}{\sum}} q_jp_j$ that are respectively proportional to $-\hbar \nu^{(\boldsymbol{\alpha})}_{\infty,0}$ and $-\hbar \nu^{(\boldsymbol{\alpha})}_{\infty,-1}$. These terms shall be removed after a suitable symplectic rescaling of $(q_i,p_i)_{1\leq i\leq g}$.
\end{itemize}

\section{Expression of the Hamiltonian and Lax matrices in terms of symmetric Darboux coordinates}\label{SecNew}
In this section, we show that we may use the symmetric polynomials $\left(e_i(\{\check{q}_1,\dots,\check{q}_g\})\right)_{1\leq i\leq g}$ to obtain polynomial Hamiltonians and polynomial explicit formulas for the matrices $\td{L}$ and $\td{A}^{(\boldsymbol{\alpha}_\tau)}$. 

\subsection{Notations and identities regarding symmetric polynomials}
In the rest of the paper we need to introduce elementary symmetric polynomials and other families of symmetric polynomials. 

\begin{definition}[Families of symmetric polynomials]\label{DefSymmetricPoly}We shall introduce the following families of symmetric polynomials:
\begin{itemize}
\item Elementary symmetric polynomials are denoted by $\left(e_i(\{x_1,\dots,x_n\})\right)_{i\geq 0}$ with the convention that $e_0(\{x_1,\dots,x_n\})=1$ and $e_k(\{x_1,\dots,x_n\})=0$ if $k<0$ or $k>n$. By definition we have:
\beq e_k(\{x_1,\dots,x_n\})=\sum_{1\leq i_1<\dots<i_k\leq n} x_{i_1}\dots x_{i_k} \,\,,\,\,  \forall \, k\in \llbracket 1, n\rrbracket\eeq
\item Complete homogeneous symmetric polynomial are denoted by $\left(h_i(\{x_1,\dots,x_n\})\right)_{i\geq 0}$ with the convention that $h_0(\{x_1,\dots,x_n\})=1$ and $h_k=0$ for $k<0$. By definition we have:
\beq h_k(\{x_1,\dots,x_n\})=\sum_{1\leq i_1\leq \dots\leq i_k\leq n}x_{i_1}\dots x_{i_k} \,\,,\,\,  \forall \, k\geq1\eeq
\item $k^{\text{th}}$ symmetric power sum polynomials are denoted by $\left(S_k(\{x_1,\dots,x_n\})\right)_{k\geq 0}$. By definition, we have:
\bea S_0(\{x_1,\dots,x_n\})&=&n\cr
S_k(\{x_1,\dots,x_n\})&=&\sum_{j=1}^n x_j^k \,,\, \forall \, k\geq 1
\eea
\end{itemize}
\end{definition}

Each of the families $(e_1,\ldots,e_n)$, $(h_1,\ldots,h_n)$ and $(S_1,\ldots,S_n)$ generates the algebra of symmetric polynomials over $\mathbb C$; these finite families are not vector-space bases of that algebra. We also have the relations
\bea \label{SymmPoly} \prod_{j=1}^n (\lambda-x_j)&=&\sum_{k=0}^n (-1)^{n-k} e_{n-k}(\{x_1,\dots,x_n\})\lambda^k=\sum_{k=0}^n (-1)^{k} e_{k}(\{x_1,\dots,x_n\})\lambda^{n-k}\cr
\frac{1}{\underset{j=1}{\overset{n}{\prod}} (\lambda-x_j)}&=& \sum_{k=0}^{\infty}h_k(\{x_1,\dots,x_n\})\lambda^{-n-k}
\eea
The relation between the various sets are given by
\bea \label{Relationhe}h_0(\{x_1,\dots,x_n\})&=&e_0(\{x_1,\dots,x_n\})\cr
h_k(\{x_1,\dots,x_n\})&=&\sum_{j=1}^k (-1)^{j}\sum_{\substack{(b_1,\dots,b_j)\in \llbracket 1,k\rrbracket^j \\ b_1+\dots+b_j=k}}\,\,\prod_{m=1}^j (-1)^{b_m}e_{b_m}(\{x_1,\dots,x_n\})\,\,,\,\, \forall \, k\geq1\cr
&&\eea
and $\forall\, m\geq 1$:
\bea\label{RelationSe} &&S_m(\{x_1,\dots,x_n\})=(-1)^m m\sum_{k=1}^m \frac{1}{k}\hat{B}_{m,k}(-e_1(\{x_1,\dots,x_n\}),\dots,-e_{m-k+1}(\{x_1,\dots,x_n\}))\cr
&&=(-1)^m m\sum_{\substack{b_1+2b_2+\dots+mb_m=m\\ b_1\geq 0,\dots,b_m\geq 0}} \frac{(-1)^{b_1+\dots+b_m}}{(b_1+\dots+b_m)} \binom{b_1+\dots +b_m}{b_1,\dots,b_m } \prod_{i=1}^m e_i(\{x_1,\dots,x_n\})^{b_i} \cr
&&
\eea
where $\left(\hat{B}_{m,k}\right)_{m\geq k\geq 0}$ are the ordinary Bell polynomials. Finally, we also have the identities 
\bea \label{Identitites}
(n-k)e_{k}(\{x_1,\dots,x_n\})&=&\sum_{i=0}^k (-1)^{i} e_{k-i}(\{x_1,\dots,x_n\})S_{i}(\{x_1,\dots,x_n\})\,\,,\,\, \forall \,k\in \llbracket 0, n\rrbracket\cr 
S_{k}(\{x_1,\dots,x_n\})&=&\sum_{i=k-n}^{k-1} (-1)^{k-1+i} e_{k-i}(\{x_1,\dots,x_n\})S_{i}(\{x_1,\dots,x_n\})\,\,,\,\, \forall \,k\geq n\cr
&& 
\eea

Elementary symmetric polynomials satisfy some useful relations. In all interpolation and Vandermonde formulas below, the variables $x_i$ are assumed pairwise distinct. Polynomial identities obtained after clearing denominators extend to coincident variables:
\begin{lemma}\label{LemmaESP1}For any $(i,m)\in \llbracket 1,g\rrbracket^2$:
\beq \frac{\partial e_i(\{x_1,\dots,x_g\})}{\partial x_m}=\sum_{j=0}^{i-1} (-1)^j e_{i-1-j}(\{x_1,\dots,x_g\})x_m^j\eeq
\end{lemma}
\begin{proposition}\label{PropESP1} For any $i\in \llbracket 1,g\rrbracket$, we have
\beq \sum_{k=1}^g \frac{\partial e_i(\{x_1,\dots,x_g\})}{\partial x_k} \prod_{j\neq k}^g\frac{\lambda-x_j}{x_k-x_j}=\sum_{j=0}^{i-1}(-1)^je_{i-j-1}(\{x_1,\dots,x_g\})\lambda^j\eeq
\end{proposition}
These relations allow to express $Q(\lambda)$.
\begin{corollary}\label{QCorollary}We have
\beq Q(\lambda)=\sum_{j=0}^{g-1} (-1)^{j-1}\left(\sum_{i=j+1}^g P_i Q_{i-j-1}\right)\lambda^j\eeq
\end{corollary}
Moreover, the elementary symmetric polynomials satisfy:
\small{\bea \label{IdentityPoly1}\forall\, (i,j)\in \llbracket 1,n\rrbracket^2\,:\, e_{n-i}(\{x_1,\dots,x_n\}\setminus\{x_j\})&=&\sum_{m=i}^n (-1)^{m-i}e_{n-m}(\{x_1,\dots,x_n\})x_j^{m-i}\cr
\forall\, j\in \llbracket 1,n\rrbracket\,:\,  0&=&\sum_{m=0}^n (-1)^{n-m}e_{n-m}(\{x_1,\dots,x_n\})x_j^{m}
\eea}
\normalsize{so} that we obtain
\begin{lemma}\label{LemmaInversionVandermonde1} For any $i\in \llbracket 1,n\rrbracket$ and any $M\geq 0$ we have:
\beq x_i^M=\sum_{j=1}^n\sum_{m=\text{Max}(j,j+n-1-M)}^n (-1)^{n-m}e_{n-m}(\{x_1,\dots,x_n\})h_{M+m-j-n+1}(\{x_1,\dots,x_n\})x_i^{j-1}\eeq
\end{lemma}
In particular, we may invert the Vandermonde matrix with the following Proposition.
\begin{proposition}\label{PropInversionVandermonde2}For any $i \in \llbracket 0,n\rrbracket$ and $M\geq 0$ we have:
\small{\beq \sum_{j=1}^n \frac{(-1)^{n-i}e_{n-i}(\{x_1,\dots,x_n\}\setminus\{x_j\})}{\underset{m\neq j}{\prod}(x_j-x_m)}x_j^M=\sum_{m=\text{Max}(i,i+n-1-M)}^{n} (-1)^{n-m}e_{n-m}(\{x_1,\dots,x_n\})h_{M+m-i-n+1}(\{x_1,\dots,x_n\})\eeq}
\normalsize{In} particular, for $M\leq n-1$ we get:
\beq \forall \, M\in \llbracket 0,n-1\rrbracket\,:\, \sum_{j=1}^n \frac{(-1)^{n-i}e_{n-i}(\{x_1,\dots,x_n\}\setminus\{x_j\})}{\underset{m\neq j}{\prod}(x_j-x_m)}x_j^M=\delta_{i,M+1}\eeq
\end{proposition}

\begin{proof}For completeness, the proofs are presented in Appendix \ref{AppendixH}.
\end{proof}

\subsection{Symmetric Darboux coordinates}
Let us first recall the well-known result from symplectic geometry.

\begin{lemma}\label{LemmaSymplectic}
Let $\alpha\neq0$ and $\beta$ be fixed constants, put $x_i=\alpha^{-1}Q_i+\beta$, and let $f=(f_1,\ldots,f_g)$ be locally invertible. In the complex setting assume $f$ and $h$ holomorphic; in the real setting assume $f$ of class $\mathcal C^2$ and $h$ of class $\mathcal C^1$. Define
\bea
\td Q_k&=&f_k(x_1,\ldots,x_g),\cr
\alpha P_i+h(x_i)&=&\sum_{k=1}^g\td P_k
\frac{\partial f_k}{\partial x_i}(x_1,\ldots,x_g),
\qquad \forall\, 1\leq i\leq g.
\eea
Then $(Q_i,P_i)\mapsto(\td Q_i,\td P_i)$ is a local symplectic change of coordinates.
\end{lemma}
\begin{proof}
Since $dQ_i=\alpha\,dx_i$, the defining equations give
\beq
\sum_i P_i\,dQ_i
=\sum_k\td P_k\,d\td Q_k-\sum_i h(x_i)\,dx_i.
\eeq
The last one-form is locally exact. Exterior differentiation therefore yields
$\underset{i}{\sum} dQ_i\wedge dP_i=\underset{i}{\sum} d\td Q_i\wedge d\td P_i$.
\end{proof}

We apply the lemma to the elementary symmetric polynomials, whose Jacobian is invertible when the roots are pairwise distinct.

\begin{definition}[Symmetric Darboux coordinates]\label{DefNewCoord}
At fixed deformation times, define
\bea
Q_i&=&e_i(q_1,\ldots,q_g)
=e_i\bigl(T_2^{-1}(\check q_1-T_1),\ldots,
T_2^{-1}(\check q_g-T_1)\bigr),\cr
p_i&=&\sum_{k=1}^gP_k\frac{\partial e_k(q_1,\ldots,q_g)}{\partial q_i}
=T_2\check p_i+\frac12\td P_1\bigl(T_2^{-1}(\check q_i-T_1)\bigr).
\eea
We set $Q_0=1$ and $Q_k=0$ for $k<0$ or $k>g$.
\end{definition}
The map from $(q,p)$ uses Lemma \ref{LemmaSymplectic} with $\alpha=1$, $\beta=0$, $h=0$. The map from $(\check q,\check p)$ uses $\alpha=T_2$, $\beta=-\frac{T_1}{T_2}$ and $h(x)=\frac{\td P_1(x)}{2}$. These statements are fiberwise, with the deformation times held fixed.

\subsection{Polynomial expression of the Hamiltonian in the symmetric Darboux coordinates}
The transformation $(q_i,p_i)\mapsto(Q_i,P_i)$ is symplectic and has no explicit time dependence. Hence the Hamiltonians in Theorem \ref{HamTheorem} transform by substitution into these coordinates. Fiberwise symplecticity of a general time-dependent transformation, such as the map involving $(\check q,\check p)$, would additionally require the usual correction to the Hamiltonian.
\begin{theorem}[Expression of the general Hamiltonian in terms of symmetric Darboux coordinates]\label{HamiltonianSymmetricCoordinates}We have:
\small{\bea\label{GeneralHamiltonianNewVariables} ~&&\text{Ham}^{(\boldsymbol{\alpha})}(\mathbf{Q},\mathbf{P})=-\hbar \sum_{j=1}^{g} (\nu^{(\boldsymbol{\alpha})}_{\infty,0} p_j+\nu^{(\boldsymbol{\alpha})}_{\infty,-1}q_jp_j)-\hbar \sum_{j=1}^g\sum_{k=1}^{r_\infty-1}c^{(\boldsymbol{\alpha})}_{\infty,k}q_j^{k}+\sum_{i=1}^g \nu^{(\boldsymbol{\alpha})}_{\infty,i}H_{\infty,i-1}\cr
&=&-\hbar \nu^{(\boldsymbol{\alpha})}_{\infty,0}\sum_{k=0}^{g-1}(g-k)Q_kP_{k+1}-\hbar \nu^{(\boldsymbol{\alpha})}_{\infty,-1}\sum_{k=1}^g kQ_kP_k-\hbar \sum_{k=1}^{r_\infty-1}c^{(\boldsymbol{\alpha})}_{\infty,k}S_k(\{q_1,\dots,q_g\})\cr
&&-\hbar\sum_{i=1}^g\nu^{(\boldsymbol{\alpha})}_{\infty,i}\sum_{k=i+1}^g\left((-1)^{i}(g-i)P_k Q_{k-1-i}+\sum_{m=i+1}^{k-1}(-1)^{m}P_k Q_{k-1-m}S_{m-i}(\{q_1,\dots,q_g\})\right)\cr
&&+\sum_{i=1}^g\nu^{(\boldsymbol{\alpha})}_{\infty,i}\sum_{k_1=1}^g\sum_{k_2=1}^g P_{k_1}P_{k_2}\Big[(-1)^{i-1}\sum_{r_1=\text{Max}(0,i-k_2)}^{\text{Min}(k_1-1,i-1)}Q_{k_1-1-r_1}Q_{k_2-i+r_1}\cr
&&+ \displaystyle \sum_{\substack{0\leq r_1\leq k_1-1 \\ 0\leq r_2\leq k_2-1 \\ r_1+r_2\geq g}}(-1)^{r_1+r_2}Q_{k_1-1-r_1}Q_{k_2-1-r_2}\sum_{m=i}^g (-1)^{g-m}Q_{g-m}h_{r_1+r_2+m-i-g+1}(\{q_1,\dots,q_g\})\Big]\cr
&&+\sum_{i=1}^g\nu^{(\boldsymbol{\alpha})}_{\infty,i}\sum_{k=1}^g\Big[ \sum_{r=0}^{\text{Min}(k-1,i-1)} (-1)^r t_{\infty,2i-2r} P_k Q_{k-1-r} \cr
&&+ \sum_{r=0}^{k-1}\sum_{s=g-r}^{g+1}\sum_{m=i}^{g} (-1)^{g+r-m}t_{\infty,2s+2}  P_k Q_{k-1-r}  Q_{g-m}h_{r+s+m-i-g+1}(\{q_1,\dots,q_g\})\Big]\cr
&&+\sum_{i=1}^g\nu^{(\boldsymbol{\alpha})}_{\infty,i}\sum_{r=g}^{2r_\infty-4}\sum_{m=i}^{g} (-1)^{g-m} \td{P}_{\infty,r}^{(2)} Q_{g-m}h_{r+m-i-g+1}(\{q_1,\dots,q_g\})\cr
&&
\eea}
\normalsize{where} $\left(\td{P}_{\infty,k}^{(2)}\right)_{g\leq k\leq 2g+2}$ are given by Definition \ref{DefP2}, $\left(\nu^{(\boldsymbol{\alpha})}_{\infty,i}\right)_{-1\leq i\leq g}$ are given by \eqref{RelationNuAlphaInfty}, and coefficients $\left(S_k(\{q_1,\dots,q_g\})\right)_{k\geq 0}$, $\left(h_k(\{q_1,\dots,q_g\})\right)_{k\geq 0}$ are given by Definition \ref{DefSymmetricPoly}.
\end{theorem}

Expression \eqref{GeneralHamiltonianNewVariables} shows that the general Hamiltonian is polynomial in $(Q_i,P_i)_{1\leq i\leq g}$ and at most quadratic in the momenta. Its explicit time dependence enters through the coefficients $\nu_{\infty,i}^{(\boldsymbol\alpha)}$, $c_{\infty,k}^{(\boldsymbol\alpha)}$, $\td P_{\infty,k}^{(2)}$ and the displayed even-time terms. The coefficients $\td P_{\infty,k}^{(2)}$ include the $\hbar$ correction of Definition \ref{DefP2}.

\begin{proof}The proof is done in Appendix \ref{ProofTermsHamiltonians}.
\end{proof}

\subsection{Expressing the Lax matrices with the symmetric Darboux coordinates}

Symmetric Darboux coordinates $(Q_1,\dots,Q_g,P_1,\dots,P_g)$ are well-suited for the matrix $\td{L}$ given by \eqref{TdLEquations} as the following proposition shows

\begin{proposition}\label{PropTdL}
Put
\bea
R(\lambda)&=&\prod_{j=1}^g(\lambda-q_j)
=\sum_{k=0}^g(-1)^kQ_k\lambda^{g-k},\cr
Q(\lambda)&=&\sum_{j=0}^{g-1}(-1)^{j-1}
\left(\sum_{i=j+1}^gP_iQ_{i-j-1}\right)\lambda^j,\cr
b(\lambda)&=&\frac12t_{\infty,2r_\infty-2}\lambda+g_0,
\qquad g_0=\frac12t_{\infty,2r_\infty-4}
+\frac12t_{\infty,2r_\infty-2}Q_1.
\eea
For $g=0$, use $R=1$, $Q=0$ and $Q_1=0$. Let $[\cdot]_+$ denote the polynomial part at $\lambda=\infty$. Then
\bea
\td L_{1,1}&=&-Q-bR,\cr
\td L_{1,2}&=&R,\cr
\td L_{2,2}&=&\td P_1+Q+bR,\cr
\td L_{2,1}&=&
\left[\frac{-\td P_2-\td P_1Q-Q^2}{R}\right]_+
-b^2R-b(2Q+\td P_1)+\hbar b'.
\eea
Equivalently, the last entry is
\beq
\td L_{2,1}=
\left[\frac{-\td P_2^{(0)}-\td P_1Q-Q^2}{R}\right]_+
-b^2R-b(2Q+\td P_1).
\eeq
These entries are polynomial in the symmetric Darboux coordinates and at most quadratic in the momenta.
\end{proposition}
\begin{proof}
Substitute Proposition \ref{PropLaxMatrix} in \eqref{CheckLEquations} and then use \eqref{TdLEquations}. The apparent poles cancel by \eqref{DefCi}. Taking polynomial parts gives the displayed formulas. The difference between $\td P_2$ and $\td P_2^{(0)}$ contributes $-\hbar t_{\infty,2r_\infty-2}/2$ to the quotient, which cancels $\hbar b'$.

Polynomial division by the monic polynomial $R$ is polynomial in its coefficients. More explicitly, if $F(\lambda)=\sum_{j=0}^N F_j\lambda^j$, then
\beq
\left[\frac{F}{R}\right]_+
=\sum_{i=0}^{N-g}\left(\sum_{j=g+i}^N F_j
h_{j-g-i}(q_1,\ldots,q_g)\right)\lambda^i,
\eeq
with an empty sum if $N<g$. All $h_k$ are polynomials in the $Q_i$, determined by
\beq
\sum_{k\geq0}h_k z^k
=\left(\sum_{k=0}^g(-1)^kQ_kz^k\right)^{-1}.
\eeq
\end{proof}

\medskip

We may also obtain the entries of $\td{A}_{\boldsymbol{\alpha}}$ in terms of the symmetric Darboux coordinates. 

\begin{proposition}\label{PropTdA}
Use $R,Q,b$ from Proposition \ref{PropTdL} and set
\beq
K=Q+bR,\qquad
E=\sum_{k=0}^{g-1}H_{\infty,k}\lambda^k,\qquad
c=\sum_{k=0}^{r_\infty-1}c_{\infty,k}^{(\boldsymbol\alpha)}\lambda^k,
\qquad v=[A_{\boldsymbol\alpha}]_{1,2}.
\eeq
For readability write $\nu_k=\nu_{\infty,k}^{(\boldsymbol\alpha)}$,
$a=t_{\infty,2r_\infty-2}$,
$\alpha_a=\alpha_{\infty,2r_\infty-2}$ and
$\alpha_d=\alpha_{\infty,2r_\infty-4}$. Then
\bea
[\td A_{\boldsymbol\alpha}]_{1,1}
&=&c-[vQ]_+-bRv,\cr
[\td A_{\boldsymbol\alpha}]_{1,2}&=&Rv,\cr
\Tr\td A_{\boldsymbol\alpha}
&=&2c+[v\td P_1]_++\hbar(g+1)\nu_{-1},\cr
[\td A_{\boldsymbol\alpha}]_{2,2}
&=&\Tr\td A_{\boldsymbol\alpha}-[\td A_{\boldsymbol\alpha}]_{1,1},\cr
[\td A_{\boldsymbol\alpha}]_{2,1}
&=&\left[\frac{v(-\td P_2+E-\td P_1K-K^2)+\hbar c'}{R}\right]_+\cr
&&-\hbar\left((g+1)\nu_{-1}b
+\frac a2(g\nu_0+\nu_{-1}Q_1)\right)
+\frac{\hbar}{2}\bigl(\alpha_a(\lambda+Q_1)+\alpha_d\bigr).
\eea
Here $\td P_2$ is the corrected polynomial of Definition \ref{DefP2}, $E$ is determined by Proposition \ref{PropDefCi2}, and $v$ by Proposition \ref{PropA12Form}. In particular $Rv$ is polynomial. Its Laurent coefficients are given by \eqref{RelationNuAlphaInfty} and \eqref{recursivenus}; thus the displayed polynomial parts require only finitely many known coefficients. For the scalar normalization $c_{\infty,0}^{(\boldsymbol\alpha)}=0$, all entries are polynomial in the symmetric Darboux coordinates. Other scalar normalizations add a scalar multiple of $I_2$.
\end{proposition}
\begin{proof}
The exact gauge computation and the polynomial-part calculation are given in Appendix \ref{ProofAtildeSymmetric}.
\end{proof}

\section{Decomposition and reduction of the space of isomonodromic deformations}\label{SectionReduc}
The second goal of this paper is now to provide a decomposition of the space of isomonodromic deformations (of dimension $2g+4$) into a subspace of trivial deformations associated to trivial times (i.e. for which rescaled Darboux coordinates are independent) and a subspace of non-trivial deformations of dimension $g$ associated to non-trivial times while providing the corresponding Hamiltonian evolutions.

\subsection{Subspaces of trivial and non-trivial deformations}
In the previous section we considered general isomonodromic deformations relatively to all irregular times by considering $\mathcal{L}_{\boldsymbol{\alpha}}$ characterized by a general vector $\boldsymbol{\alpha}\in \mathbb{C}^{2g+4}$. However, as we will see below, there exists a subspace of deformations of dimension $g+4$ for which the evolutions of the Darboux coordinates are trivial, thus leaving only a non-trivial subspace of deformations of dimension $g$. These non-trivial deformations shall later be mapped to $g$ isomonodromic times whose expressions will be explicit in terms of the initial irregular times. Trivial deformations shall correspond to the fact that only odd irregular times $\left(t_{\infty,2k-1}\right)_{1\leq k\leq r_\infty-1}$ are relevant whereas even irregular times $\left(t_{\infty,2k}\right)_{1\leq k\leq r_\infty-1}$ do not appear in the Hamiltonians. In other words, \textbf{considering meromorphic connections in $\mathfrak{gl}_2(\mathbb{C})$ or in $\mathfrak{sl}_2(\mathbb{C})$ is essentially the same at the level of the Hamiltonian systems}. This remark shall provide a subspace of trivial deformations of dimension $g+2$. Finally, the remaining $2$ trivial deformations correspond to the remaining two degrees of freedom in the action of the M\"{o}bius transformations. As we will see below, this choice encodes the necessity of a symplectic rescaling (translation and dilatation) of the Darboux coordinates. These two degrees of freedom shall be used to fix the first two leading non-trivial coefficients at infinity: $t_{\infty, 2r_\infty-3}$ (conventionally set to $2$) and $t_{\infty,2r_\infty-5}$ (conventionally set to $0$). 

\medskip

Let us first recall that the space of isomonodromic deformations, denoted $\mathcal{T}$, is given by:
\beq \mathcal{L}_{\boldsymbol{\alpha}}=\hbar \sum_{k=1}^{2r_\infty-2} \alpha_{\infty,k} \partial_{t_{\infty,k}}\eeq
We make the identification with $\mathbb{C}^{2r_\infty-2}$ by identifying an isomonodromic deformation $\mathcal{L}_{\boldsymbol{\alpha}}$ with its vector $\boldsymbol{\alpha}\in \mathbb{C}^{2r_\infty-2}$: 

\beq \mathcal{L}_{\boldsymbol{\alpha}}=\hbar \sum_{k=1}^{2r_\infty-2} \alpha_{\infty,k} \partial_{t_{\infty,k}} \,\,\Leftrightarrow\,\, \boldsymbol{\alpha}= \sum_{k=1}^{2r_\infty-2} \alpha_{\infty,k}  \mathbf{e}_k\eeq
where we shall denote $\left(\mathbf{e}_k\right)_{1\leq k\leq 2r_\infty-2}$ the canonical basis of $\mathbb{C}^{2r_\infty-2}$.

\begin{definition}\label{TrivialVectors} We define the following vectors of $\mathbb{C}^{2r_\infty-2}$ and their corresponding deformations.
\bea \mathbf{w}_k&=&\mathbf{e}_{2k} \,\,,\,\, \forall \, k\in \llbracket 1,r_\infty-1\rrbracket\cr
 \mathbf{u}_{k}&=&\frac{1}{2}\sum_{m=1}^{r_\infty-k-2} (2m-1)t_{\infty,2m+1+2k}\mathbf{e}_{2m-1} +\frac{1}{2}\sum_{s=1}^{r_\infty-k-2}2s\, t_{\infty,2s+2k+2}\mathbf{e}_{2s}\cr
&=&\frac{1}{2}\sum_{r=1}^{2r_\infty-2k-4} r\,t_{\infty,r+2k+2} \mathbf{e}_r  \,\,,\,\,  \forall \, k\in \llbracket -1,r_\infty-3\rrbracket
\eea
and we shall denote:
\bea \mathcal{U}_{\text{trivial}}&=&\text{Span}\left\{\mathbf{w}_1,\dots,\mathbf{w}_{r_\infty-1},\mathbf{u}_{-1},\mathbf{u}_0\right\}\cr
\mathcal{U}_{\text{iso}}&=&\text{Span}\left\{\mathbf{u}_1,\dots,\mathbf{u}_{r_\infty-3}\right\}
\eea
\end{definition}

Note in particular that $\mathcal{U}_{\text{iso}}$ is of dimension $g=r_\infty-3$ and that $(\mathbf{w}_1,\dots,\mathbf{w}_{r_\infty-1},\mathbf{u}_{-1},\dots,\mathbf{u}_{r_\infty-3})$ is a basis of $\mathbb{C}^{2r_\infty-2}$. The choice of basis is such that the following proposition holds.

\begin{proposition}\label{PropReduction}We have for all $k\in\llbracket 1, r_\infty-1\rrbracket$:
\bea 
\nu_j^{(\mathbf{w}_k)}&=&0 \,\,,\,\, \forall \, j\in \llbracket -1, r_\infty-3\rrbracket\cr
\mu_j^{(\mathbf{w}_k)}&=&0 \,\,,\,\, \forall \, j\in \llbracket 1, g\rrbracket\cr
c_{\infty,j}^{(\mathbf{w}_k)}&=&-\frac{1}{2k}\delta_{j,k} \,\,,\,\, \forall \, j\in \llbracket 1, r_\infty-1\rrbracket
\eea
and for all $k\in\llbracket -1, r_\infty-3\rrbracket$:
\bea 
\nu_j^{(\mathbf{u}_k)}&=&\delta_{j,k} \,\,,\,\, \forall \, j\in \llbracket -1, r_\infty-3\rrbracket\cr
c_{\infty,j}^{(\mathbf{u}_k)}&=&0 \,\,,\,\, \forall \, j\in \llbracket 1, r_\infty-1\rrbracket
\eea
\end{proposition}

\begin{proof}The proof is presented in Appendix \ref{AppendixB}.
\end{proof}
\begin{remark}
    The vectors of deformations $(\mathbf{u}_{-1},\mathbf{u}_0)$ can be obtained from the action of translations and dilatations on $\lambda$ on the irregular times. These transformations correspond to the subset of M\"{o}bius transformations on $\lambda$ fixing infinity which is prescribed to be a pole in the present setting). It is expected that the Hamiltonian system is invariant under these transformations since it does not depend on the choice of parametrization of the meromorphic Lax matrix and its spectral parameter.  
\end{remark}

Note that Proposition \ref{PropReduction} and \eqref{RelationNuMuMatrixForm} imply that 
\beq\label{mujuspecial} \mu_j^{(\mathbf{u}_{-1})}=0=\mu_j^{(\mathbf{u}_{0})}\,\,,\,\, \forall \, j\in \llbracket 1, g\rrbracket\eeq
so that for all $j\in \llbracket 1,g\rrbracket$:

Proposition \ref{PropReduction} combined with Theorem \ref{HamTheorem} provides the following theorem.

\begin{theorem}\label{TheoSplitTangentSpace}For any $j\in \llbracket 1,g\rrbracket$, we have:
\bea  \mathcal{L}_{\mathbf{w}_k}[q_j]&=&0 \,\, ,\,\, \forall\, k\in \llbracket 1,r_\infty-1\rrbracket \cr
\mathcal{L}_{\mathbf{w}_k}[p_j]&=&-\frac{\hbar}{2} q_j^{k-1}\,\, ,\,\, \forall\, k\in \llbracket 1,r_\infty-1\rrbracket \cr
\mathcal{L}_{\mathbf{u}_{-1}}[q_j]&=&-\hbar q_j\cr
\mathcal{L}_{\mathbf{u}_{-1}}[p_j]&=&\hbar p_j\cr
\mathcal{L}_{\mathbf{u}_{0}}[q_j]&=&-\hbar\cr
\mathcal{L}_{\mathbf{u}_{0}}[p_j]&=&0
\eea
\end{theorem}

\begin{proof}
Let $j\in \llbracket 1,g\rrbracket$ and $k\in \llbracket 1,r_\infty-1\rrbracket$. From Proposition \ref{PropReduction} we have $\nu^{(\mathbf{w}_k)}_i=0$ for all $i\in \llbracket -1,r_\infty-3\rrbracket$. Consequently, from \eqref{RelationNuMuMatrixForm}, $\mu^{(\mathbf{w}_k)}_i=0$ for all $i\in \llbracket 1,g\rrbracket$. Therefore $\mathcal{L}_{\mathbf{w}_k}[q_j]=0$ from Theorem \ref{HamTheorem}. Similarly, we also have from Proposition \ref{PropReduction} that $c_{\infty,i}^{(\mathbf{w}_k)}=-\frac{1}{2k}\delta_{i,k}$ for all  $i\in \llbracket 1, r_\infty-1\rrbracket$ so that Theorem \ref{HamTheorem} provides
\beq \mathcal{L}_{\mathbf{w}_k}[p_j]=\hbar \sum_{i=1}^{r_\infty-1}ic^{(\mathbf{w}_k)}_{\infty,i}q_j^{i-1}=-\frac{\hbar}{2} q_j^{k-1}\eeq

\medskip

Let us now consider $\mathcal{L}_{u_{-1}}$. From Proposition \ref{PropReduction} we have $\nu^{(\mathbf{u}_{-1})}_i=\delta_{i,-1}$ for all $i\in \llbracket -1,r_\infty-2\rrbracket$. Consequently, from \eqref{RelationNuMuMatrixForm}, $\mu^{(\mathbf{u}_{-1})}_i=0$ for all $i\in \llbracket 1,g\rrbracket$. From Theorem \ref{HamTheorem}, we get that
\beq \mathcal{L}_{\mathbf{u}_{-1}}[q_j]=-\hbar\nu^{(\mathbf{u}_{-1})}_{\infty,-1} q_j=-\hbar q_j\eeq
Similarly, since from Proposition \ref{PropReduction} we have $c_{\infty,i}^{(\mathbf{u}_{-1})}=0$ for all  $i\in \llbracket 1, r_\infty-1\rrbracket$, Theorem \ref{HamTheorem} provides:
\beq \mathcal{L}_{\mathbf{u}_{-1}}[p_j]=\hbar\nu^{(\mathbf{u}_{-1})}_{\infty,-1} p_j=\hbar p_j\eeq

\medskip

Let us now consider $\mathcal{L}_{u_{0}}$. From Proposition \ref{PropReduction} we have $\nu^{(\mathbf{u}_{0})}_i=\delta_{i,0}$ for all $i\in \llbracket -1,r_\infty-2\rrbracket$. Consequently, from \eqref{RelationNuMuMatrixForm}, $\mu^{(\mathbf{u}_0)}_i=0$ for all $i\in \llbracket 1,g\rrbracket$. From Theorem \ref{HamTheorem}, we get that
\beq \mathcal{L}_{\mathbf{u}_{0}}[q_j]=-\hbar\nu^{(\mathbf{u}_{0})}_{\infty,0} =-\hbar\eeq
Similarly, since from Proposition \ref{PropReduction} we have $c_{\infty,i}^{(\mathbf{u}_{0})}=0$ for all  $i\in \llbracket 1, r_\infty-1\rrbracket$, Theorem \ref{HamTheorem} provides:
\beq \mathcal{L}_{\mathbf{u}_{0}}[p_j]=0\eeq

\end{proof}

Note that $\mathcal{L}_{\mathbf{u}_{-1}}$ and $\mathcal{L}_{\mathbf{u}_{0}}$ do not act trivially on $(q_j,p_j)_{1\leq j\leq g}$. As we will see below, one needs to rescale the Darboux coordinates in order to have a trivial action. The purpose of the next section is to define trivial and isomonodromic times that are dual to the previous deformations. However, it is not possible to define some times $\left(\tau_{1},\dots, \tau_{r_\infty-3}\right)$ such that $\hbar \partial_{\tau_k}=\mathcal{L}_{\mathbf{u}_{k}}$ since the system becomes non compatible for $g\geq 4$.

\subsection{Definition of trivial times and isomonodromic times}
The split in the tangent space between trivial and non-trivial subspaces may be translated at the level of coordinates. This corresponds to choosing $g$ non-trivial times and $g+4$ trivial times for which the evolutions of the shifted Darboux coordinates are trivial. In particular, one may then choose the values of these trivial times to any arbitrary values without changing the Hamiltonian evolutions. However, it s important to notice that the choice of trivial times and isomonodromic times is not unique since, for example, one may use any arbitrary combination of isomonodromic times to provide a new one. We propose the following set of trivial and non-trivial times that are particularly convenient in our context.

\begin{definition}[Trivial and non-trivial deformation times]\label{Times}
Work locally on $t_{\infty,2r_\infty-3}\neq0$ and fix a branch of its logarithm. All fractional powers below use this same branch. Define the trivial times by
\bea T_{\infty,k}&=&t_{\infty,2k} \,\, ,\,\, \forall\, k\in \llbracket 1,r_\infty-1\rrbracket\cr
T_2&=&\left(\frac{1}{2}t_{\infty,2r_\infty-3}\right)^{\frac{2}{2r_\infty-3}}\cr
T_1&=&\frac{t_{\infty,2r_\infty-5}}{2r_\infty-5} \left(\frac{1}{2}t_{\infty,2r_\infty-3}\right)^{-\frac{2r_\infty-5}{2r_\infty-3}}
\eea
We also define the $g=r_\infty-3$ ``isomonodromic'' times $\left(\tau_{k}\right)_{1\leq k\leq g}$,for all $k\in \llbracket 1, g\rrbracket$, by: 
\small{\bea \tau_k&=&\sum_{i=0}^{k-1}\frac{(-1)^i\left(\underset{s=1}{\overset{i}{\prod}} (2r_\infty-2k+2s-7)\right)  \left(\frac{1}{2}t_{\infty,2r_\infty-5}\right)^i \left(\frac{1}{2}t_{\infty,2r_\infty-3}\right)^{-\frac{(2r_\infty-3)i+2r_\infty-5-2k}{2r_\infty-3}} \frac{1}{2}t_{\infty, 2r_\infty-5-2k+2i} }{i!(2r_\infty-5)^i}\cr
&&+ \frac{(-1)^{k}\left(\underset{s=1}{\overset{k}{\prod}}(2r_\infty-2k+2s-7)\right) \left(\frac{1}{2}t_{\infty, 2r_\infty-5}\right)^{k+1} \left(\frac{1}{2}t_{\infty,2r_\infty-3}\right)^{-\frac{(k+1)(2r_\infty-5)}{2r_\infty-3} }}{(k+1)(k-1)!(2r_\infty-5)^{k}}\cr
&&
\eea}
\normalsize{We} shall denote $\mathcal{T}_{\text{trivial}}$ the set of trivial times and $\mathcal{T}_{\text{iso}}$ the set of isomonodromic times:
\beq
\mathcal T_{\mathrm{trivial}}
=\{T_{\infty,1},\ldots,T_{\infty,r_\infty-1},T_1,T_2\},
\qquad
\mathcal T_{\mathrm{iso}}=\{\tau_1,\ldots,\tau_g\}.
\eeq
\end{definition}

On the branch fixed in Definition \ref{Times}, these variables form a local coordinate system in one-to-one correspondence with the irregular times $(t_{\infty,k})_{1\leq k\leq2r_\infty-2}$. The local inverse is given by the following proposition.

\begin{proposition}\label{InverseCoordinates}One may recover the irregular times $\left(t_{\infty,k}\right)_{1\leq k\leq 2r_\infty-2}$ from $\mathcal{T}_{\text{trivial}}\cup\mathcal{T}_{\text{iso}}$ with the following formulas:
\bea t_{\infty,2r_\infty-3}&=&2T_2^{\frac{2r_\infty-3}{2}}\cr
t_{\infty,2r_\infty-5}&=&(2r_\infty-5)T_1\,T_2^{\frac{2r_\infty-5}{2}}\cr
t_{\infty,2i}&=&T_{\infty,i} \,\,,\,\, \forall \, i\in \llbracket 1,r_\infty-1\rrbracket
\eea
and for all $k\in \llbracket 1,r_\infty-3\rrbracket$:
\beq t_{\infty,2k-1}=2T_2^{\frac{2k-1}{2}}\left(\sum_{p=1}^{r_\infty-k-2}  \frac{\underset{m=p+1}{\overset{r_\infty-k-2}{\prod}}(2r_\infty-2m-5)}{2^{r_\infty-k-p-2}(r_\infty-k-p-2)!}T_1^{r_\infty-k-p-2}\tau_p + T_1^{r_\infty-1-k}\frac{\underset{m=0}{\overset{r_\infty-k-2}{\prod}}(2r_\infty-2m-5)}{2^{r_\infty-1-k}(r_\infty-1-k)!}
\right)\eeq
\end{proposition}

\begin{proof}The proof is computational and is proposed in Appendix \ref{AppendixF}.
\end{proof}

The last proposition allows to obtain immediately the expression of derivatives relatively to trivial and non-trivial times using the chain rule.

\begin{proposition}\label{Derivativetau}For all $k\in \llbracket 1,r_\infty-3 \rrbracket$, we have:
\beq \partial_{\tau_{k}}=2\sum_{i=1}^{r_\infty-2-k}    \frac{\underset{m=k+1}{\overset{r_\infty-i-2}{\prod}}(2r_\infty-2m-5)}{2^{r_\infty-i-k-2}(r_\infty-i-k-2)!}T_1^{r_\infty-i-k-2}T_2^{\frac{2i-1}{2}}\partial_{t_{\infty,2i-1}}\eeq
and
\footnotesize{\bea \partial_{T_{\infty,i}}&=&\partial_{t_{\infty,2i}} \,\,,\,\, \forall \, i\in \llbracket 1,r_\infty-1\rrbracket\cr
\partial_{T_1}&=&(2r_\infty-5)T_2^{\frac{2r_\infty-5}{2}}\partial_{t_{\infty,2r_\infty-5}}\cr
&&+2\sum_{k=1}^{r_\infty-3}T_2^{\frac{2k-1}{2}}\left(\sum_{p=1}^{r_\infty-k-3}  \frac{\underset{m=p+1}{\overset{r_\infty-k-2}{\prod}}(2r_\infty-2m-5)}{2^{r_\infty-k-p-2}(r_\infty-k-p-3)!}T_1^{r_\infty-k-p-3}\tau_p + T_1^{r_\infty-k-2}\frac{\underset{m=0}{\overset{r_\infty-k-2}{\prod}}(2r_\infty-2m-5)}{2^{r_\infty-1-k}(r_\infty-k-2)!}\right)\partial_{t_{\infty,2k-1}}\cr
\partial_{T_2}&=&  (2r_\infty-3)T_2^{\frac{2r_\infty-5}{2}} \partial_{t_{\infty,2r_\infty-3}}+\frac{(2r_\infty-5)^2}{2}T_1\,T_2^{\frac{2r_\infty-7}{2}}\partial_{t_{\infty,2r_\infty-5}}\cr
&&+ \sum_{k=1}^{r_\infty-3} (2k-1)T_2^{\frac{2k-3}{2}} \left(\sum_{p=1}^{r_\infty-k-2}  \frac{T_1^{r_\infty-k-p-2}\tau_p\underset{m=p+1}{\overset{r_\infty-k-2}{\prod}}(2r_\infty-2m-5)}{2^{r_\infty-k-p-2}(r_\infty-k-p-2)!} + \frac{T_1^{r_\infty-1-k}\underset{m=0}{\overset{r_\infty-k-2}{\prod}}(2r_\infty-2m-5)}{2^{r_\infty-1-k}(r_\infty-1-k)!}\right)\partial_{t_{\infty,2k-1}}\cr
&&
\eea}\normalsize{}
\end{proposition}

For clarity, we shall denote $\boldsymbol{\alpha}^{\tau_k}$ the corresponding vector in the tangent space associated to $\partial_{\tau_k}$ for any $k\in \llbracket 1,r_\infty-3\rrbracket$. Its entries are given by
\bea\label{alphaki} \boldsymbol{\alpha}^{\tau_k}_{\infty,2i}&=&0 \,\,,\,\, \forall\, i\in\llbracket 1, r_\infty-1\rrbracket\cr
   \boldsymbol{\alpha}^{\tau_k}_{\infty,2r_\infty-3}&=&0 \cr
	\boldsymbol{\alpha}^{\tau_k}_{\infty,2r_\infty-5}&=&0\cr
\boldsymbol{\alpha}^{\tau_k}_{\infty,2i-1}&=&
\begin{cases}
\displaystyle
\frac{2\prod_{m=k+1}^{r_\infty-i-2}(2r_\infty-2m-5)}
{2^{r_\infty-i-k-2}(r_\infty-i-k-2)!}
T_1^{r_\infty-i-k-2}T_2^{(2i-1)/2},
&1\leq i\leq r_\infty-k-2,\\
0,&r_\infty-k-2<i\leq r_\infty-3.
\end{cases}\cr
	&&
\eea

\begin{remark}\label{Remarknutauk} Note that inserting \eqref{alphaki} into \eqref{RelationNuAlphaInfty} implies that
\beq \forall\,k\in \llbracket 1,g\rrbracket\,\,,\,\, \forall \, i\in \llbracket -1,k-1\rrbracket\,:\, \nu_{\infty,i}^{(\boldsymbol{\alpha}^{\tau_k})}=0\eeq
In particular 
\beq \forall\,k\in \llbracket 1,g\rrbracket\,:\, \nu_{\infty,-1}^{(\boldsymbol{\alpha}^{\tau_k})}=\nu_{\infty,0}^{(\boldsymbol{\alpha}^{\tau_k})} =0\eeq
\end{remark}

\subsection{Properties of trivial and isomonodromic times}

Trivial and non-trivial times are chosen so that they satisfy the following properties.

\begin{proposition}\label{PropTrivialTimes} Use the convention $t_{\infty,2r_\infty}=0$ in the endpoint formula below. For all $k\in \llbracket 1, r_\infty-1\rrbracket$:
\bea \mathcal{L}_{\mathbf{w}_k}[T_{\infty,j}]&=&\hbar\delta_{j,k}\,\,,\,\, \mathcal{L}_{\mathbf{u}_{-1}}[T_{\infty,j}]=\hbar j t_{\infty,2j}\,\,,\,\, \mathcal{L}_{\mathbf{u}_0}[T_{\infty,j}]=\hbar jt_{\infty,2j+2} \,\,,\,\,\forall\, j\in \llbracket 1, r_\infty-1\rrbracket\cr
\mathcal{L}_{\mathbf{w}_k}[T_2]&=&0 \,\,,\,\, \mathcal{L}_{\mathbf{u}_{-1}}[T_2]=\hbar T_2\,\,,\,\,\mathcal{L}_{\mathbf{u}_{0}}[T_2]=0\cr
\mathcal{L}_{\mathbf{w}_k}[T_1]&=&0 \,\,,\,\, \mathcal{L}_{\mathbf{u}_{-1}}[T_1]=0\,\,,\,\,\mathcal{L}_{\mathbf{u}_{0}}[T_1]=\hbar T_2
\eea
\end{proposition}

\begin{proof} Results on $\left(T_{\infty,j}\right)_{1\leq j\leq r_\infty-1}$ follow by straightforward computations using the fact that
\bea \mathcal{L}_{\mathbf{w}_{k}}&=&\hbar \partial_{t_{\infty,2k}} \,\,,\,\, \forall \, k\in \llbracket 1, r_\infty-1\rrbracket\cr
\mathcal{L}_{\mathbf{u}_{-1}}&=&\frac{\hbar}{2}\sum_{r=1}^{2r_\infty-2}rt_{\infty,r} \partial_{t_{\infty,r}}\cr
\mathcal{L}_{\mathbf{u}_{0}}&=&\frac{\hbar}{2}\sum_{r=1}^{2r_\infty-4}rt_{\infty,r+2} \partial_{t_{\infty,r}}
\eea
Results on $T_2$ are also straightforward using the fact that $T_2$ only depends on $t_{\infty,2r_\infty-3}$. Finally since $T_1$ depends only on $t_{\infty,2r_\infty-3}$ and $t_{\infty,2r_\infty-5}$, we get that
\bea \mathcal{L}_{\mathbf{u}_{-1}}[T_1]&=&\frac{\hbar}{2}(2r_\infty-3)t_{\infty,2r_\infty-3}\partial_{t_{\infty,2r_\infty-3}}[T_1]+\frac{\hbar}{2}(2r_\infty-5)t_{\infty,2r_\infty-5}\partial_{t_{\infty,2r_\infty-5}}[T_1]\cr
&=& \frac{\hbar}{2}(2r_\infty-3)t_{\infty,2r_\infty-3}\left(-\frac{ t_{\infty,2r_\infty-5}}{2r_\infty-5}\frac{(2r_\infty-5)}{(2r_\infty-3)}\left(\frac{1}{2}\right)^{-\frac{2r_\infty-5}{2r_\infty-3}} \left(t_{\infty,2r_\infty-3}\right)^{-\frac{2r_\infty-5}{2r_\infty-3} -1}\right)\cr
&&+\frac{\hbar}{2}(2r_\infty-5)t_{\infty,2r_\infty-5} \frac{1}{2r_\infty-5}\left(\frac{1}{2}t_{\infty,2r_\infty-3}\right)^{-\frac{2r_\infty-5}{2r_\infty-3}}\cr
&=&0
\eea
and
\bea
\mathcal{L}_{\mathbf{u}_{0}}[T_1]&=&\frac{\hbar}{2}(2r_\infty-5)t_{\infty,2r_\infty-3}\partial_{t_{\infty,2r_\infty-5}}[T_1]\cr
&=&\frac{\hbar}{2}(2r_\infty-5)t_{\infty,2r_\infty-3}\frac{1}{2r_\infty-5}\left(\frac{1}{2}t_{\infty,2r_\infty-3}\right)^{-\frac{2r_\infty-5}{2r_\infty-3}}\cr
&=&\hbar \left(\frac{1}{2}t_{\infty,2r_\infty-3}\right)^{\frac{2}{2r_\infty-3}}=\hbar T_2
\eea
\end{proof}

\subsection{Shifted Darboux coordinates}
Theorem \ref{TheoSplitTangentSpace} indicates that deformations $\mathcal{L}_{\mathbf{u}_{-1}}$ and $\mathcal{L}_{\mathbf{u}_{0}}$ do not act trivially on the Darboux coordinates $(q_j,p_j)_{1\leq j\leq g}$. However, since the action is very simple, we may easily perform a symplectic transformation on the Darboux coordinates to obtain ``shifted Darboux coordinates'' for which the action of $\mathcal{L}_{\mathbf{u}_{-1}}$ and $\mathcal{L}_{\mathbf{u}_{0}}$ becomes trivial.

\begin{definition}\label{ShiftDarbouxCoordinates} The shifted Darboux coordinates $(\check{q}_j,\check{p}_j)_{1\leq j\leq g}$ are defined by
\bea \check{q}_j&=&T_2 q_j+T_1\cr
\check{p}_j&=&T_2^{-1}\left(p_j-\frac{1}{2}\td{P}_1(q_j)\right) \,,\,\, \forall \, j\in \llbracket 1,g\rrbracket
\eea
\end{definition}

Using Theorem \ref{TheoSplitTangentSpace} and Proposition \ref{PropTrivialTimes}, we get that the shifted Darboux coordinates satisfy the following proposition.

\begin{proposition}\label{PropDarbouxCoordinates} For all $j\in \llbracket 1,g\rrbracket$:
\bea  \mathcal{L}_{\mathbf{w}_k}[\check{q}_j]&=&\mathcal{L}_{\mathbf{w}_k}[\check{p}_j]=0 \,\, ,\,\, \forall\, k\in \llbracket 1,r_\infty-1\rrbracket \cr
\mathcal{L}_{\mathbf{u}_{-1}}[\check{q}_j]&=&\mathcal{L}_{\mathbf{u}_{-1}}[\check{p}_j]=0\cr
\mathcal{L}_{\mathbf{u}_{0}}[\check{q}_j]&=&\mathcal{L}_{\mathbf{u}_{0}}[\check{p}_j]=0
\eea
In other words, for any $\boldsymbol{\alpha}_0\in \mathcal{U}_{\text{trivial}}$: $\mathcal{L}_{\boldsymbol{\alpha}_0}[\check{q}_j]=\mathcal{L}_{\boldsymbol{\alpha_0}}[\check{p}_j]=0$, hence the terminology ``trivial deformations'' and ``trivial subspace''.
\end{proposition}

\begin{proof}The proof directly follows from Theorem \ref{TheoSplitTangentSpace} and Proposition \ref{PropTrivialTimes} but we detail it in Appendix \ref{AppendixD}.
\end{proof}

Note also that the change of coordinates $\left(q_j,p_j\right)_{1\leq j\leq g} \,\leftrightarrow \, (\check{q}_j,\check{p}_j)_{1\leq j\leq g}$ is symplectic in the sense that $\underset{j=1}{\overset{g}{\sum}} dq_j\wedge dp_j= \underset{j=1}{\overset{g}{\sum}} d\check{q}_j\wedge d\check{p}_j$, with the deformation times held fixed.

\medskip

We finally get to our second main theorem.

\begin{theorem}[Local independence of the shifted coordinates from trivial times]\label{TheoReduction}
On the local branch fixed in Definition \ref{Times}, and along an isomonodromic solution with fixed generalized monodromy data, the shifted Darboux coordinates $(\check q_j,\check p_j)_{1\leq j\leq g}$ are independent of $(T_{\infty,1},\ldots,T_{\infty,r_\infty-1},T_1,T_2)$ and depend only on $(\tau_1,\ldots,\tau_g)$. More generally, the local solutions $f(t_{\infty,1},\ldots,t_{\infty,2r_\infty-2})$ of
\beq
\mathcal L_{\mathbf w_k}[f]=0\quad(1\leq k\leq r_\infty-1),
\qquad \mathcal L_{\mathbf u_{-1}}[f]=0,
\qquad \mathcal L_{\mathbf u_0}[f]=0
\eeq
are exactly the functions $f=u(\tau_1,\ldots,\tau_g)$.
\end{theorem}

\begin{proof}The proof is presented in Appendix \ref{AppendixE}.
\end{proof}

Finally let us mention the following observation.

\begin{proposition}\label{PropositionTrace}
Let $\check A_j=\check A_{\boldsymbol\alpha^{\tau_j}}$ and $\td A_j=\td A_{\boldsymbol\alpha^{\tau_j}}$. Their traces are independent of $\lambda$. The one-forms
\beq
\check\eta=\sum_{j=1}^g\Tr(\check A_j)\,d\tau_j,
\qquad
\td\eta=\sum_{j=1}^g\Tr(\td A_j)\,d\tau_j
\eeq
are closed. On a simply connected local time domain choose primitives $\check F$ and $\td F$, and put
\beq
\check G=\exp\left(-\frac{\check F}{2\hbar}\right)I_2,
\qquad
\td G=\exp\left(-\frac{\td F}{2\hbar}\right)I_2.
\eeq
The gauges $\check\Psi_{\mathrm n}=\check G\check\Psi$ and $\td\Psi_{\mathrm n}=\td G\td\Psi$ make all their respective auxiliary matrices traceless simultaneously and leave $\check L$ and $\td L$ unchanged.
\end{proposition}
\begin{proof}
The coefficients of $\Tr\check L=\Tr\td L=\td P_1(\lambda)$ are trivial times. Taking the trace of spatial compatibility gives $\partial_\lambda\Tr\check A_j=\partial_\lambda\Tr\td A_j=0$. Time compatibility reads
\beq\label{EqGauge}
\hbar\partial_{\tau_j}\check A_i
=\hbar\partial_{\tau_i}\check A_j+[\check A_j,\check A_i].
\eeq
Its trace proves that $\check\eta$ is closed, and similarly for $\td\eta$. The scalar gauge changes $\check A_j$ to
$\check A_j-\tfrac12(\partial_{\tau_j}\check F)I_2$, which is traceless. Since the gauges are scalar and independent of $\lambda$, they do not change the spatial matrices.
\end{proof}

The last proposition shall be useful when $\td{L}$ and $\check{L}$ are traceless. In this case, it is interesting to perform this additional gauge transformation in order to obtain a Lax pair that belongs to $\mathfrak{sl}_2(\mathbb{C})$ rather than $\mathfrak{gl}_2(\mathbb{C})$. In particular, this is always possible for the canonical choice of trivial times that shall be proposed in Section \ref{SectionMainResult}.

\section{Canonical choice of trivial times and simplification of the Hamiltonian systems}\label{SectionMainResult}
The purpose of this section is to select some specific values of the trivial times in order to obtain simpler form of the Hamiltonian evolutions of Theorem \ref{HamTheorem}. Indeed, the last section indicates (Theorem \ref{TheoReduction}) that the shifted Darboux coordinates are independent of the values of the trivial times so that we may choose them without affecting the Hamiltonian evolutions. As it turns out, there exists a natural choice of the trivial times for which the Hamiltonian evolutions drastically simplify.

\subsection{Canonical choice of the trivial times and main theorem}

\begin{definition}[Canonical choice of the trivial times]\label{TrivialTimesChoice} We define the ``canonical choice of trivial times'' by choosing
\bea T_{\infty,k}&=&0 \,\,\,,\,\, \forall \, k\in \llbracket 1,r_\infty-1\rrbracket, \cr
T_{1}&=&0 ,\cr
T_{2}&=&1.
\eea 
\end{definition}

In the rest of the paper, we shall always set the trivial times to their canonical values. The canonical choice of trivial times implies that
\begin{itemize}\item All even irregular times are set to $0$: for all $k\in \llbracket 1, r_\infty-1\rrbracket$: $t_{\infty,2k}=0$.
\item $t_{\infty,2r_\infty-3}=2$ and $t_{\infty,2r_\infty-5}=0$.
\item $\td{P}_1$ is identically null. This is equivalent to say that $\td{L}$ and $\check{L}$ are traceless. Hence, Proposition \ref{PropositionTrace} implies that under a potential additional trivial gauge transformation, we may choose a gauge in which $\td{L}$, $\check{L}$,$\td{A}_{\boldsymbol{\alpha}^{\tau}}$ and $\check{A}_{\boldsymbol{\alpha}^{\tau}}$ are traceless for any isomonodromic time $\tau\in \mathcal{T}_{\text{iso}}$.
\item The shifted Darboux coordinates are identical to the initial Darboux coordinates:
\beq \forall \, j\in \llbracket 1,g\rrbracket\,:\, \check{q}_j=q_j \,\text{ and }\, \check{p}_j=p_j\eeq
\item The isomonodromic times $\tau_k$ identify with an irregular time:
\beq \label{Identifisoirreg}\forall\, k\in \llbracket 1,g\rrbracket\,:\, \tau_k= \frac{1}{2}t_{\infty,2r_\infty-2k-5} \,\, \Leftrightarrow\,\, \frac{1}{2}t_{\infty,2k-1}=\tau_{r_\infty-k-2}\eeq
\item $\td{P}_2$ reduces to $\td{P}_2(\lambda)=-\lambda$ if $r_\infty=3$ or for $r_\infty\geq 4$:
\bea \label{ReducedtdP2}\td{P}_2(\lambda)&=&-\lambda^{2r_\infty-5}-\sum_{k=r_\infty-2}^{2r_\infty-7}\left(2\tau_{2r_\infty-k-6}+\sum_{m=k-r_\infty+6}^{r_\infty-3}\tau_{r_\infty-m-2}\tau_{r_\infty-k+m-5}\right)\lambda^k\cr
&&-\left(2\tau_{r_\infty-3}+\sum_{m=3}^{r_\infty-3}\tau_{r_\infty-m-2}\tau_{m-2}\right)\lambda^{r_\infty-3}
\eea
In other words, for $r_\infty\geq 4$, we have
\small{\bea\label{ReducedtdP2Coeffs} \td{P}^{(2)}_{\infty,2r_\infty-4}&=&0\cr
\td{P}^{(2)}_{\infty,2r_\infty-5}&=&-1\cr
\td{P}^{(2)}_{\infty,2r_\infty-6}&=&0\cr
\td{P}^{(2)}_{\infty,k}&=&-\left(2\tau_{2r_\infty-k-6}+\sum_{m=k-r_\infty+6}^{r_\infty-3}\tau_{r_\infty-m-2}\tau_{r_\infty-k+m-5}\right) \,\,,\,\, \forall \,k\in\llbracket r_\infty-2, 2r_\infty-7\rrbracket\cr
\td{P}^{(2)}_{\infty,r_\infty-3}&=&-\left(2\tau_{r_\infty-3}+\sum_{m=3}^{r_\infty-3}\tau_{r_\infty-m-2}\tau_{m-2}\right)
\eea} 
\item \normalsize{Coefficients} $\left(c^{(\boldsymbol{\alpha}_{\tau})}_{\infty,k}\right)_{1\leq k\leq r_\infty-1}$ are vanishing for any isomonodromic deformation $\tau\in \mathcal{T}_{\text{iso}}$.
\item The gauge matrices $G_1(\lambda)$ and $J(\lambda)$ of Proposition \ref{PropExplicitGaugeTransfo} simplifies to
\beq G_1(\lambda)=I_2 \,\,,\,\, J(\lambda)=\begin{pmatrix} 1 &0 \\ 
-\underset{i=1}{\overset{g}{\sum}} \frac{\check{p}_i}{\lambda-\check{q}_i} \underset{j\neq i}{\prod} \frac{1}{\check{q}_i-\check{q}_j}& \frac{1}{\underset{j=1}{\overset{g}{\prod}}(\lambda-\check{q}_j)} \end{pmatrix} \eeq
In particular, $\check{L}(\lambda)=\td{L}(\lambda)$ and for all $\tau\in \mathcal{T}_{\text{iso}}$, $\check{A}^{(\boldsymbol{\alpha}^{\tau})}(\lambda)=\td{A}^{(\boldsymbol{\alpha}^{\tau})}(\lambda)$.
\end{itemize}

We also get the explicit expression
\begin{proposition}\label{CheckLAEquationsReduced}Under the canonical choice of trivial times given by Definition \ref{TrivialTimesChoice}, the Lax matrices $L$  is given by
\bea \label{CheckLEquationsReduced}\td{L}_{1,1}(\lambda,\hbar)&=&-Q(\lambda,\hbar),\cr
\td{L}_{1,2}(\lambda,\hbar)&=&\underset{j=1}{\overset{g}{\prod}}(\lambda-\check{q}_j),\cr
\td{L}_{2,2}(\lambda,\hbar)&=&Q(\lambda,\hbar),\cr
\td{L}_{2,1}(\lambda,\hbar)&=& \hbar \frac{ \partial \bigg(\frac{Q(\lambda,\hbar)}{\underset{j=1}{\overset{g}{\prod}} (\lambda-\check{q}_j)}\bigg)}{\partial \lambda} +\frac{L_{2,1}(\lambda,\hbar) }{\underset{j=1}{\overset{g}{\prod}} (\lambda - \check{q}_j)}  - \frac{Q(\lambda,\hbar)^2}{\underset{j=1}{\overset{g}{\prod}}(\lambda-\check{q}_j)}
\eea
with $L_{2,1}(\lambda,\hbar)=-\td{P}_2(\lambda)+\underset{k=0}{\overset{r_\infty-4}{\sum}}H_{\infty,k}\lambda^k-\underset{j=1}{\overset{g}{\sum}}\frac{\hbar \check{p}_j}{\lambda-\check{q}_j}$ and $Q(\lambda,\hbar)=-\underset{i=1}{\overset{g}{\sum}} \check{p}_i \underset{j\neq i}{\prod}\frac{\lambda-\check{q}_j}{\check{q}_i-\check{q}_j}$ and $L_{2,2}(\lambda,\hbar)=\underset{j=1}{\overset{g}{\sum}}\frac{\hbar}{\lambda-\check{q}_j}$.
Similarly, the matrix $A_{\boldsymbol{\alpha}^{\tau}}(\lambda,\hbar)$ is given by
\bea  [A_{\boldsymbol{\alpha}^{\tau}}(\lambda,\hbar)]_{1,1}&=&-\sum_{j=1}^g \frac{\mu^{(\boldsymbol{\alpha}^{\tau})}_j \check{p}_j}{\lambda-\check{q}_j}\cr
[A_{\boldsymbol{\alpha}^{\tau}}(\lambda,\hbar)]_{1,2}&=&\sum_{j=1}^g \frac{\mu^{(\boldsymbol{\alpha}^{\tau})}_j}{\lambda-\check{q}_j}\cr
[A_{\boldsymbol{\alpha}^{\tau}}(\lambda,\hbar)]_{2,1}&=&
-\hbar\sum_{j=1}^g\sum_{i\neq j} \frac{\mu^{(\boldsymbol{\alpha}^{\tau})}_j\check p_i}{(\lambda-\check{q}_j)(\lambda-\check{q}_i)}+\left(\sum_{j=1}^g \frac{\mu^{(\boldsymbol{\alpha}^{\tau})}_j}{\lambda-\check{q}_j}\right)\left(-\td{P}_2(\lambda)+\underset{k=0}{\overset{r_\infty-4}{\sum}}H_{\infty,k}\lambda^k\right)\cr
[A_{\boldsymbol{\alpha}^{\tau}}(\lambda,\hbar)]_{2,2}&=&
-\sum_{j=1}^g \frac{\mu^{(\boldsymbol{\alpha}^{\tau})}_j \check{p}_j}{\lambda-\check{q}_j}+\hbar\sum_{j=1}^g\sum_{i\neq j} \frac{\mu^{(\boldsymbol{\alpha}^{\tau})}_j}{(\lambda-\check{q}_j)(\lambda-\check{q}_i)}
\eea
\end{proposition}

We may also simplify Propositions \ref{PropTdL} and \ref{PropTdA}.

\begin{proposition}\label{ProptdLtdAReduced}Under the canonical choice of trivial times given by Definition \ref{TrivialTimesChoice}, the Lax matrices $\td{L}$ and $\td{A}_{\boldsymbol{\alpha}}$ may be expressed in terms of symmetric Darboux coordinates as follow. For any $\tau\in \mathcal{T}_{\text{iso}}$:
\footnotesize{\bea \td{L}_{1,1}(\lambda)&=&-\sum_{j=0}^{g-1}(-1)^{j-1}\left(\sum_{i=j+1}^{g} P_i Q_{i-j-1}\right)\lambda^j\cr
\td{L}_{1,2}(\lambda)&=& \sum_{m=0}^{g}(-1)^{g-m} Q_{g-m} \lambda^m\cr
\td{L}_{2,2}(\lambda)&=&\sum_{j=0}^{g-1}(-1)^{j-1}\left(\sum_{i=j+1}^{g} P_i Q_{i-j-1}\right)\lambda^j\cr 
\td{L}_{2,1}(\lambda)&=&-\sum_{i=0}^{r_\infty-2}\sum_{j=g+i}^{2r_\infty-5}\left(\td{P}_{\infty,j}^{(2)}\,h_{j-g-i}(\{\mathbf{\check{q}}\})\right)\lambda^i\cr
&&-\sum_{i=0}^{g-2}\left( \sum_{j_1=i+1}^{g-1}\sum_{j_2=g+i-j_1}^{g-1} (-1)^{j_1+j_2}\left(\sum_{i_1=j_1+1}^{g} P_{i_1} Q_{i_1-j_1-1}\right)\left(\sum_{i_2=j_2+1}^{g} P_{i_2} Q_{i_2-j_2-1}\right) h_{j_1+j_2-g-i}(\{\mathbf{\check{q}}\})\right)\lambda^i\cr
[\td{A}_{\boldsymbol{\alpha}^{\tau}}(\lambda)]_{1,1}&=&-\sum_{i=0}^{g-2} \sum_{m=1}^{g-1-i}(-1)^{i+m-1}\nu^{(\boldsymbol{\alpha})}_{\infty,m}\left(\sum_{r=i+m+1}^{g} P_r Q_{r-i-m-1}\right)\lambda^{i}\cr
[\td{A}_{\boldsymbol{\alpha}^{\tau}}(\lambda)]_{1,2}&=&\sum_{j=0}^{g-1}\left(\sum_{m=1}^{g-j}(-1)^{g-j-m}\nu^{(\boldsymbol{\alpha})}_{\infty,m}Q_{g-j-m}\right)\lambda^j \cr
[\td{A}_{\boldsymbol{\alpha}^{\tau}}(\lambda)]_{2,2}&=&-[\td{A}^{(\boldsymbol{\alpha})}(\lambda)]_{1,1}\cr
[\td{A}_{\boldsymbol{\alpha}^{\tau}}(\lambda)]_{2,1}&=&-\sum_{i=0}^{g}\sum_{j=\text{Max}(g,g+i-1)}^{2r_\infty-5}\sum_{m=1}^{j-g-i}\nu^{(\boldsymbol{\alpha})}_{\infty,m}h_{j-g-m-i}(\{\mathbf{\check{q}}\})\td{P}_{\infty,j}^{(2)}\lambda^i\cr
&&-\sum_{i=0}^g\sum_{j_1=0}^{g-1}\sum_{j_2=0}^{g-1}\sum_{m=1}^{j_1+j_2-g-i}(-1)^{j_1+j_2}\nu^{(\boldsymbol{\alpha})}_{\infty,m}h_{j_1+j_2-g-m-i}(\{\mathbf{\check{q}}\})\sum_{r_1=j_1+1}^{g}\sum_{r_2=j_2+1}^{g}P_{r_1}P_{r_2}Q_{r_1-j_1-1}Q_{r_2-j_2-1}\lambda^i\cr
&&
\eea}
\normalsize{where} $\left(\td{P}_{\infty,k}^{(2)}\right)_{r_\infty-3\leq k\leq 2r_\infty-4}$ are determined by \eqref{ReducedtdP2Coeffs} and $\left(h_k(\{\mathbf{\check{q}}\})\right)_{k\geq 0}$ are expressed in terms of symmetric Darboux coordinates by $h_0(\{\mathbf{\check{q}}\})=1$ and
\beq \label{RelationheReduced}
h_k(\{\mathbf{\check{q}}\})=\sum_{j=1}^k (-1)^{j}\sum_{\substack{(b_1,\dots,b_j)\in \llbracket 1,k\rrbracket^j \\ b_1+\dots+b_j=k}}\,\,\prod_{m=1}^j (-1)^{b_m}Q_{b_m}\,\,,\,\, \forall \, k\geq1\eeq
Coefficients $\left(\nu_{\infty,k}^{(\boldsymbol{\alpha}^\tau)}\right)_{1\leq k\leq r_\infty-3}$ shall be given by Proposition \ref{Propnureduced} depending on the isomonodromic deformation $\tau\in \mathcal{T}_{\text{iso}}$ and $\nu^{(\boldsymbol{\alpha})}_{\infty,r_\infty-2}=\underset{k=1}{\overset{g}{\sum}}(-1)^{g-k}\nu^{(\boldsymbol{\alpha})}_{\infty,k}Q_{g+1-k}$.
\end{proposition}

We shall now apply Theorem \ref{HamTheorem} for the canonical values of the trivial times and obtain our third main theorem. 

\begin{theorem}[Hamiltonian representation for the canonical choice of trivial times]\label{HamTheoremReduced} The canonical choice of the trivial times given by Definition \ref{TrivialTimesChoice} and the definition of trivial times (Definition \ref{Times}) imply that for any isomonodromic time $\tau\in \mathcal{T}_{\text{iso}}$:
\beq \label{DefHamReduced} \text{Ham}^{(\boldsymbol{\alpha}^\tau)}(\check{\mathbf{q}},\check{\mathbf{p}})=\sum_{k=0}^{r_\infty-4} \nu_{\infty,k+1}^{(\boldsymbol{\alpha}^\tau)}H_{\infty,k}\eeq
In other words, the Hamiltonian is a (time-dependent) linear combination of the companion/interpolation coefficients $(H_{\infty,k})_{0\leq k\leq r_\infty-4}$ that are determined by
\beq  \label{ReducedDefCi2}\begin{pmatrix}1&\check{q}_1 &\dots &\dots &\check{q}_1^{g-1}\\
1& \check{q}_2&\dots &\dots&\check{q}_2^{g-1} \\
\vdots & & & & \vdots\\
\vdots & & & & \vdots\\
1& \check{q}_{g} &\dots & \dots& \check{q}_{g}^{g-1}\end{pmatrix}\begin{pmatrix} H_{\infty,0}\\ \vdots\\ \vdots\\ H_{\infty,r_\infty-4}\end{pmatrix}=
\begin{pmatrix} \check{p}_1^2+\td{P}_2(\check{q}_1)+\hbar \underset{i\neq 1}{\sum}\frac{\check{p}_i-\check{p}_1}{\check{q}_1-\check{q}_i}\\
\vdots\\ \vdots\\
\check{p}_g^2 +\td{P}_2(\check{q}_g)+\hbar \underset{i\neq g}{\sum}\frac{\check{p}_i-\check{p}_g}{\check{q}_g-\check{q}_i}
\end{pmatrix}
\eeq
where $\td{P_2}$ is given by \eqref{ReducedtdP2}. Coefficients $\left(\nu_{\infty,k}^{(\boldsymbol{\alpha}^\tau)}\right)_{1\leq k\leq r_\infty-3}$ shall be given by Proposition \ref{Propnureduced} depending on the isomonodromic deformation $\tau\in \mathcal{T}_{\text{iso}}$. In terms of symmetric Darboux coordinates, the Hamiltonian is given by: 
\bea\text{Ham}^{(\boldsymbol{\alpha}^\tau)}(\mathbf{Q},\mathbf{P})&=&-\hbar\sum_{i=1}^g\nu^{(\boldsymbol{\alpha}^\tau)}_{\infty,i}\sum_{k=i+1}^g\left((-1)^{i}(g-i)P_k Q_{k-1-i}+\sum_{m=i+1}^{k-1}(-1)^{m}P_k Q_{k-1-m}S_{m-i}(\{\mathbf{\check{q}}\})\right)\cr
&&+\sum_{i=1}^g\nu^{(\boldsymbol{\alpha}^\tau)}_{\infty,i}\sum_{k_1=1}^g\sum_{k_2=1}^g P_{k_1}P_{k_2}\Big[(-1)^{i-1}\sum_{r_1=\text{Max}(0,i-k_2)}^{\text{Min}(k_1-1,i-1)}Q_{k_1-1-r_1}Q_{k_2-i+r_1}\cr
&&+ \displaystyle \sum_{\substack{0\leq r_1\leq k_1-1 \\ 0\leq r_2\leq k_2-1 \\ r_1+r_2\geq g}}(-1)^{r_1+r_2}Q_{k_1-1-r_1}Q_{k_2-1-r_2}\sum_{m=i}^g (-1)^{g-m}Q_{g-m}h_{r_1+r_2+m-i-g+1}(\{\mathbf{\check{q}}\})\Big]\cr
&&+\sum_{i=1}^g\nu^{(\boldsymbol{\alpha}^\tau)}_{\infty,i}\sum_{r=g}^{2r_\infty-5}\sum_{m=i}^{g} (-1)^{g-m} \td{P}_{\infty,r}^{(2)} Q_{g-m}h_{r+m-i-g+1}(\{\mathbf{\check{q}}\})
\eea
where $\left(S_k(\{\mathbf{\check{q}}\})\right)_{k\geq0}$ and $\left(h_k(\{\mathbf{\check{q}}\})\right)_{k\geq0}$ are determined by $h_0(\{\mathbf{\check{q}}\})=1$, $S_0(\{\mathbf{\check{q}}\})=g$ and for all $k\geq1$:
\bea\label{RelationSeReduced} 
h_k(\{\mathbf{\check{q}}\})&=&\sum_{j=1}^k (-1)^{j}\sum_{\substack{(b_1,\dots,b_j)\in \llbracket 1,k\rrbracket^j \\ b_1+\dots+b_j=k}}\,\,\prod_{m=1}^j (-1)^{b_m}Q_{b_m}\,\,,\,\, \forall \, k\geq1 \cr
S_k(\{\mathbf{\check{q}}\})
&=&(-1)^k k\sum_{\substack{b_1+2b_2+\dots+kb_k=k\\ b_1\geq 0,\dots,b_k\geq 0}} \frac{(-1)^{b_1+\dots+b_k}}{(b_1+\dots+b_k)} \binom{b_1+\dots +b_k}{b_1,\dots,b_k } \prod_{i=1}^k Q_i^{b_i}
\eea
\end{theorem}

\begin{proof}The proof is obvious since the canonical choice of trivial times implies that the coefficients $\left(c^{(\boldsymbol{\alpha}_{\tau})}_{\infty,k}\right)_{1\leq k\leq r_\infty-1}$ are vanishing for any isomonodromic deformation. Moreover, Remark \ref{Remarknutauk} implies that $\nu_{\infty,-1}^{(\boldsymbol{\alpha}^{\tau})}=\nu_{\infty,0}^{(\boldsymbol{\alpha}^{\tau})} =0$.
\end{proof}

Only the coefficients $\nu_{\infty,k}^{(\boldsymbol\alpha^\tau)}$ depend on the choice of deformation direction. The interpolation functions $H_{\infty,k}$ can vary along a deformation and include explicit $\hbar$ corrections. At fixed symmetric Darboux coordinates on the canonical slice, $\td L$ has no explicit $\hbar$ dependence; its characteristic coefficients are the $\hbar=0$ parts $H_{\infty,k}^{(0)}$, not the full interpolation coefficients. For example, at $g=2$,
\beq
H_{\infty,0}=H_{\infty,0}^{(0)}+\hbar P_2,
\qquad H_{\infty,1}=H_{\infty,1}^{(0)},
\qquad
\{H_{\infty,0},H_{\infty,1}\}
=-\hbar(3Q_1^2-2Q_2+2\tau_1),
\eeq
where $\{Q_i,P_j\}=\delta_{ij}$. Thus the full $H_{\infty,k}$ are not in general an involutive isospectral family. We now determine the coefficients $\nu$ explicitly.

\begin{proposition}[Expression of $\nu_{\infty,k}^{(\boldsymbol{\alpha}^{\tau_j})}$]\label{Propnureduced} For any $j\in \llbracket 1,g\rrbracket$, the coefficients $\left(\nu_{\infty,k}^{(\boldsymbol{\alpha}^{\tau_j})}\right)_{1\leq k\leq r_\infty-3}$ are determined under the canonical choice of trivial times of Definition \ref{TrivialTimesChoice} by
\beq \begin{pmatrix}1&0&\dots&\dots&\dots&0\\
0&1&0&\ddots&&0\\
\tau_1 &0&1&0&\ddots &\vdots\\
\vdots&\ddots &\ddots&\ddots&\ddots &\vdots\\
\vdots&\ddots &\ddots&\ddots&\ddots &\vdots\\
\tau_{g-2}&\tau_{g-3}&\dots& \tau_1&0&1 
\end{pmatrix}\begin{pmatrix}\nu_{\infty,1}^{(\boldsymbol{\alpha}^{\tau_j})}\\ \vdots\\ \vdots\\\nu_{\infty,r_\infty-3}^{(\boldsymbol{\alpha}^{\tau_j})}\end{pmatrix}=\frac{2}{2r_\infty-2j-5} \mathbf{e}_{j}
\eeq
\end{proposition}

One may also obtain a simplified expression for $\left(\mu_i^{(\boldsymbol{\alpha}^{\tau_j})}\right)_{1\leq i\leq g}$ from \eqref{RelationNuMuMatrixForm}. We find for the canonical choice of trivial times:
\footnotesize{\beq\label{ReducedExpremu} \begin{pmatrix}1&0&\dots&\dots&\dots&0\\
0&1&0&\ddots&&0\\
\tau_1 &0&1&0&\ddots &\vdots\\
\vdots&\ddots &\ddots&\ddots&\ddots &\vdots\\
\vdots&\ddots &\ddots&\ddots&\ddots &\vdots\\
\tau_{g-2}&\tau_{g-3}&\dots& \tau_1&0&1 
\end{pmatrix}\begin{pmatrix}1&1 &\dots &\dots &1\\
\check{q}_1& \check{q}_2&\dots &\dots& \check{q}_{g}\\
\vdots & & & & \vdots\\
\vdots & & & & \vdots\\
\check{q}_1^{r_\infty-4}& \check{q}_2^{r_\infty-4} &\dots & \dots& \check{q}_{g}^{r_\infty-4}\end{pmatrix}\begin{pmatrix}\mu_{1}^{(\boldsymbol{\alpha}^{\tau_j})}\\ \vdots\\ \vdots\\\mu_{g}^{(\boldsymbol{\alpha}^{\tau_j})}\end{pmatrix}=\frac{2}{2r_\infty-2j-5} \mathbf{e}_{j}
\eeq}
\normalsize{for all} $j\in \llbracket 1,g\rrbracket$. 

\medskip

Finally, the evolution equations under this canonical choice reduce to
\bea \label{ReducedEvolutions}
\hbar \partial_{\tau_j}\check{q}_m&=&2\mu^{(\boldsymbol{\alpha}^{\tau_j})}_m \check{p}_m -\hbar \sum_{i\neq m}\frac{\mu^{(\boldsymbol{\alpha}^{\tau_j})}_m+\mu^{(\boldsymbol{\alpha}^{\tau_j})}_i}{\check{q}_m-\check{q}_i},\cr
\hbar \partial_{\tau_j}\check{p}_m&=&\hbar \sum_{i\neq m}\frac{(\mu^{(\boldsymbol{\alpha}^{\tau_j})}_i+\mu^{(\boldsymbol{\alpha}^{\tau_j})}_m)(\check{p}_i-\check{p}_m)}{(\check{q}_m-\check{q}_i)^2} +\mu^{(\boldsymbol{\alpha}^{\tau_j})}_m\left(-\td{P}_2'(\check{q}_m)+\sum_{k=1}^{r_\infty-4}kH_{\infty,k}\check{q}_m^{k-1}\right)\cr
&&
\eea
for all $(j,m)\in \llbracket 1,g\rrbracket^2$.

\medskip

Let us now underline a few aspects of Theorem \ref{HamTheoremReduced}:

\begin{itemize}\item Note that Theorem \ref{HamTheoremReduced} is valid for any choice of the trivial times and not only the canonical choice of Definition \ref{TrivialTimesChoice}. Indeed $(\check{q}_j,\check{p}_j)_{1\leq j \leq g}$ are independent of the choice of trivial times so that we may choose any value to compute the Hamiltonian system. However, for other choices of trivial times, the connection between Darboux coordinates and shifted Darboux coordinates and the relation between irregular times and isomonodromic times may be more complex and is given by Definitions \ref{ShiftDarbouxCoordinates} and \ref{Times}. This observation was already made in the untwisted case in \cite{MarchalOrantinAlameddine2022}.
\item Note that the Hamiltonian system for $\left(\check{q}_j,\check{p}_j\right)_{1\leq j \leq g}$ does not depend on $\td{P}_1$. Since $\td{P}_1=\Tr \td{L}$, this means that the non-trivial isomonodromic evolutions are the same in the study of isomonodromic deformations of twisted connections on $\mathfrak{gl}_2(\mathbb{C})$ or $\mathfrak{sl}_2(\mathbb{C})$. However, in the case of $\mathfrak{gl}_2(\mathbb{C})$, apparent singularities are no longer the right Darboux coordinates and a shift by $-\frac{1}{2}\td{P}_1(q_j)$ becomes necessary (Cf. Definition \ref{ShiftDarbouxCoordinates}). This observation was already made in the untwisted case in \cite{MarchalOrantinAlameddine2022}.
\item Hamiltonian evolutions only depend on $\tau_g$ and $\tau_{g-1}$ through $\td{P}_2$ (because the matrix in Proposition \ref{Propnureduced} does not depend neither on $\tau_{g-1}$ nor $\tau_g$) so that the explicit dependence of the Hamiltonians in $\tau_g$ and $\tau_{g-1}$ is linear. The explicit dependence of the Hamiltonians in $(\tau_{i})_{1\leq i\leq g-2}$ is polynomial and the corresponding degrees are given by \eqref{InvertLowerTriangularToeplitz}. 
\end{itemize}

\subsection{Explicit expressions for the inverse of the matrices} 
One may invert the Vandermonde matrix in Theorem \ref{HamTheoremReduced} in order to have some explicit expressions for $\left(H_{\infty,k}\right)_{0\leq k\leq r_\infty-4}$. For all $i\in \llbracket 1, r_\infty-3\rrbracket$ we find
\beq \label{Cinvert} H_{\infty, i-1}=\sum_{m=1}^{r_\infty-3}\frac{(-1)^{g-i}e_{g-i}(\{\check{q}_1,\dots,\check{q}_g\}\setminus\{\check{q}_m\})}{\underset{r\neq m}{\prod}(\check{q}_m-\check{q}_r)}\left(\check{p}_m^2+\td{P}_2(\check{q}_m)+\hbar \sum_{s\neq m}\frac{\check{p}_s -\check{p}_m}{\check{q}_m-\check{q}_s}\right)
\eeq
where we have defined the elementary symmetric functions by
\beq \label{DefElementarySymmetricFunctions} e_{m}(\{y_1,\dots y_k\})=\sum_{1\leq j_1<\dots<j_m\leq k}y_{j_1}\dots y_{j_m}\,\,,\,\, \forall \, m\geq 0, \,k\geq 1\eeq
Similarly, one may invert the lower triangular Toeplitz matrix of Proposition \ref{Propnureduced} in order to have an explicit expression for $\nu_{\infty,k}^{(\boldsymbol{\alpha}^{\tau_j})}$. We find
\beq\label{Reducednu}
\nu_{\infty,k}^{(\boldsymbol\alpha^{\tau_j})}
=\frac{2}{2r_\infty-2j-5}
\begin{cases}
1,&k=j,\\
F_{k-j-1}(\tau_1,\ldots,\tau_{k-j-1}),&k\geq j+2,\\
0,&\text{otherwise},
\end{cases}
\qquad 1\leq j,k\leq g.
\eeq
where we have defined:
\small{\beq \begin{pmatrix}1&0&\dots&\dots&\dots&0\\
0&1&0&\ddots&&0\\
\tau_1 &0&1&0&\ddots &\vdots\\
\vdots&\ddots &\ddots&\ddots&\ddots &\vdots\\
\vdots&\ddots &\ddots&\ddots&\ddots &\vdots\\
\tau_{g-2}&\tau_{g-3}&\dots& \tau_1&0&1 
\end{pmatrix}^{-1} = \begin{pmatrix}1&0&\dots&\dots&\dots&0\\
0&1&0&\ddots&&0\\
F_1(\tau_1) &0&1&0&\ddots &\vdots\\
\vdots&\ddots &\ddots&\ddots&\ddots &\vdots\\
\vdots&\ddots &\ddots&\ddots&\ddots &\vdots\\
F_{g-2}(\tau_1,\dots,\tau_{g-2})&F_{g-3}(\tau_1,\dots,\tau_{g-3})&\dots& F_1(\tau_1)&0&1 
\end{pmatrix}
\eeq}
\normalsize{with}
\beq \label{InvertLowerTriangularToeplitz} F_i(\tau_{1},\dots,\tau_i)=\sum_{\begin{subarray}{l}
(b_1,\dots,b_i)\in \mathbb{Z}_{\geq0}^i\\ \underset{j=1}{\overset{i}{\sum}}(j+1)b_j=i+1
\end{subarray}} \binom{\underset{j=1}{\overset{i}{\sum}}b_j}{b_1,\dots,b_i} (-1)^{\underset{j=1}{\overset{i}{\sum}}b_j} \,\tau_1^{b_1}\dots \tau_{i}^{b_i} \,\,,\,\, \forall \, i\geq 1
\eeq
For example, the first values of $\left(F_i(\tau_1,\dots,\tau_i)\right)_{1\leq i\leq 5}$ are
\bea F_1(\tau_1)&=&-\tau_1\cr
F_2(\tau_1,\tau_2)&=&-\tau_2\cr
F_3(\tau_1,\tau_2,\tau_3)&=&\tau_1^2-\tau_3\cr
F_4(\tau_1,\tau_2,\tau_3,\tau_4)&=&2\tau_1\tau_2-\tau_4\cr
F_5(\tau_1,\dots,\tau_5)&=&-\tau_1^3+2\tau_1\tau_3+\tau_2^2-\tau_5
\eea

Finally, set $R(\lambda)=\prod_{j=1}^g(\lambda-\check q_j)$. For $1\leq i,k\leq g$ the residues are
\bea\label{muinvert}
\mu_i^{(\boldsymbol\alpha^{\tau_k})}
&=&\frac{2}{(2r_\infty-2k-5)R'(\check q_i)}
\Bigg[(-1)^{g-k}e_{g-k}(\{\check q_1,\ldots,\check q_g\}\setminus\{\check q_i\})\cr
&&+\sum_{s=k+2}^g(-1)^{g-s}
e_{g-s}(\{\check q_1,\ldots,\check q_g\}\setminus\{\check q_i\})
F_{s-k-1}(\tau_1,\ldots,\tau_{s-k-1})\Bigg].
\eea

\section{Topological type property and formal WKB solutions}
On the canonical slice, the compatible system is
\beq
\hbar\partial_\lambda\td\Psi=\td L\td\Psi,
\qquad
\hbar\partial_{\tau_j}\td\Psi=\td A_{\boldsymbol\alpha^{\tau_j}}\td\Psi,
\qquad 1\leq j\leq g.
\eeq
Its algebraic construction does not require a particular class of solutions. In this section, however, we restrict to formal power-series branches
$Q_i=\underset{m\geq 0}{\sum}Q_i^{(m)}\hbar^m$ and
$P_i=\underset{m\geq 0}{\sum}P_i^{(m)}\hbar^m$.
Locally away from turning points a WKB fundamental matrix has the form
\beq
\td\Psi=
\left(\sum_{m\geq0}F_m(\lambda,\boldsymbol\tau)\hbar^m\right)
\exp\left(\hbar^{-1}\operatorname{diag}(S_1,S_2)\right),
\eeq
with invertible $F_0$.

For $g\geq1$, the last auxiliary matrix is
\beq
\td A_{\boldsymbol\alpha^{\tau_g}}
=\begin{pmatrix}0&2\\2(\lambda+2Q_1)&0\end{pmatrix}.
\eeq
Writing $\td L=\left(\begin{smallmatrix}U&R\\C&-U\end{smallmatrix}\right)$, leading-order compatibility implies
\beq
[\td L^{(0)},\td A_{\boldsymbol\alpha^{\tau_g}}^{(0)}]=0,
\qquad U^{(0)}=0,
\qquad C^{(0)}=(\lambda+2Q_1^{(0)})R^{(0)}.
\eeq
Consequently, the leading spectral curve is
\beq
y^2=(R^{(0)}(\lambda))^2(\lambda+2Q_1^{(0)}),
\eeq
whose normalization has genus zero. For ordinary topological recursion we restrict further to a noncritical branch, in particular
$R^{(0)}(-2Q_1^{(0)})\neq0$.

The reconstruction theorem for systems of topological type in \cite{BergereBorotEynard,MOsl2} applies once its hypotheses have been matched to the chosen formal branch and normalization. In particular, the determinantal correlators must have the required $\hbar$ expansion, parity, leading orders and pole property. The degree-one auxiliary matrix above is useful in verifying these conditions, but is not by itself a proof of all of them. We therefore state the reconstruction here under these topological-type hypotheses.

Write $W_n$ for the connected determinantal correlation functions and $W_{h,n}$ for their coefficients at order $\hbar^{2h-2+n}$. For stable indices $2h-2+n>0$, the reconstruction identifies the differentials by
\beq
W_{h,n}(\lambda_1,\ldots,\lambda_n)
\,d\lambda_1\cdots d\lambda_n
=\omega_{h,n}(\lambda_1,\ldots,\lambda_n).
\eeq
The unstable one- and two-point terms use the matching spectral-curve and Bergman-kernel conventions, including the Cauchy subtraction in the connected two-point function when that convention is used. With the corresponding normalization of the free energies, the formal tau-function satisfies
\beq
\log\tau_{\mathrm{JMU}}(\boldsymbol\tau,\hbar)
=\sum_{h\geq0}\hbar^{2h-2}\omega_{h,0}(\boldsymbol\tau)+C(\hbar),
\eeq
where $C(\hbar)$ is independent of every deformation time. These conditional formal statements do not assert an analytic summation result or reconstruction for arbitrary transseries or critical branches.

\section{Examples}\label{SectionExamples}
Let us now apply the general theory to the first cases of the Painlev\'{e} $1$ hierarchy.
\subsection{The Airy case: $r_\infty=3$}
The Airy case corresponds to $r_\infty=3$ so that $g=0$. The canonical choice of trivial times corresponds to $t_{\infty,4}=t_{\infty,2}=0$, $t_{\infty,3}=2$ and $t_{\infty,1}=0$ so that $\td{P}_2(\lambda)=-\lambda$. There is no Darboux coordinates and any Hamiltonian evolutions. However, one may still write the Lax matrices $L$ and $\td{L}$. They are given by
\beq L(\lambda)=\begin{pmatrix} 0&1\\\lambda&0 \end{pmatrix}=\td{L}(\lambda)\eeq 
giving the Airy spectral curve: $y^2=\lambda$. Since $g=0$, the only interesting result provided by the paper is that the wave function $\Psi$ may be reconstructed by topological recursion after the introduction of the formal parameter $\hbar$. This is of course in agreement with known results about the Airy spectral curve \cite{Kontsevich,EORev,DM18}.

\subsection{Painlev\'{e} $1$ case: $r_\infty=4$}
Let us consider $r_\infty=4$, i.e. $g=1$ corresponding to the Painlev\'{e} $1$ case. In this setup, the canonical choice of trivial times corresponds to $t_{\infty,6}=t_{\infty,4}=t_{\infty,2}=0$, $t_{\infty,5}=2$ and $t_{\infty,3}=0$. The only non-trivial isomonodromic time is $\tau:=\tau_1=\frac{1}{2}t_{\infty,1}$.  Since $g=1$, we shall drop the useless index in this case (i.e. $q:=q_1$, $p:=p_1$, etc.). Application of the general results to this case is straightforward and give under the choice of trivial times made in Definition \ref{TrivialTimesChoice}:
\bea \td{P}_2(\lambda)&=&-\lambda^3 -2\tau\lambda\cr
H_{\infty,0}&=&\check{p}^2+\td{P}_2(\check{q})=\check{p}^2 -\check{q}^3 -2\tau \check{q}\cr
Q(\lambda,\hbar)&=&-\check{p}\cr
\nu_{\infty,1}^{(\boldsymbol{\alpha}^{\tau})}&=&2 \,\,,\,\, \mu^{(\boldsymbol{\alpha}^{\tau})}=2
\eea
Thus, we get that the Hamiltonian is
\beq \text{Ham}^{(\boldsymbol{\alpha}^{\tau})}(\check{q},\check{p})=2H_{\infty,0}=2\check{p}^2 -2\check{q}^3 -4\tau \check{q} \eeq
It corresponds to the ordinary differential equations
\beq \hbar \partial_\tau \check{q}=4\check{p}\,\,\,,\,\,\, \hbar \partial_\tau \check{p}=6\check{q}^2 +4\tau
\eeq
so that $\check{q}(\tau)$ satisfies a Painlev\'{e} $1$ like equation:
\beq \hbar^2\frac{d^2 \check{q}}{d\tau^2}=24\check{q}^2+16\tau  \eeq
The associated Lax pairs are given by
\bea L(\lambda,\hbar)&=&\begin{pmatrix}0&1\\ \lambda^3 +2\tau\lambda+\check{p}^2 -\check{q}^3 -2\tau \check{q} -\frac{\hbar \check{p}}{\lambda-\check{q}}& \frac{\hbar}{\lambda-\check{q}}\end{pmatrix} \,,\cr
 A_{\boldsymbol{\alpha^{\tau}}}(\lambda,\hbar)&=&\begin{pmatrix}-\frac{2\check{p}}{\lambda-\check{q}}&\frac{2}{\lambda-\check{q}}\\
\frac{2}{\lambda-\check{q}}\left(\lambda^3 +2\tau\lambda+\check{p}^2 -\check{q}^3 -2\tau \check{q}\right)&-\frac{2\check{p}}{\lambda-\check{q}} \end{pmatrix}\eea
or equivalently
\beq \td{L}(\lambda,\hbar)=\begin{pmatrix}\check{p}&\lambda-\check{q}\\ \lambda^2+\check{q}\lambda+\check{q}^2+2\tau& -\check{p} \end{pmatrix} \,,\, \td{A}_{\boldsymbol{\alpha^{\tau}}}(\lambda,\hbar)=\begin{pmatrix}0&2\\2(\lambda+2\check{q})& 0 \end{pmatrix}\eeq

\begin{remark}
If we perform $t=2^{\frac{6}{5}}\tau$, $\td{q}=2^{-\frac{2}{5}}\check{q}$, $\td{p}=2^{\frac{2}{5}}\check{p}$ we find that $\td{q}(t)$ satisfies the normalized Painlev\'{e} $1$ equation:
\beq \hbar^2\frac{d^2 \td{q}}{dt^2}=6\td{q}^2+t\eeq
Moreover, one may recover the Jimbo-Miwa Lax pair (eq. $C.2$ of \cite{JimboMiwa}): 
\beq L_{\text{JMU}}(x)=\begin{pmatrix} -z(t) & x^2+y(t)x +y(t)^2+\frac{t}{2} \\4(x-y(t))  & z(t)\end{pmatrix}\,,\,\, A_{\text{JMU}}(x)=\begin{pmatrix}0 &\frac{x}{2}+y(t)  \\2 & 0 \end{pmatrix}\eeq
\beq \Psi_{\text{JMU}}(x)=\begin{pmatrix} 0& 1\\ 2^{\frac{6}{5}}&0\end{pmatrix}\td{\Psi}(\lambda) \,\,,\,\, x=2^{-\frac{2}{5}}\lambda \,\,,\,\, t=2^{\frac{6}{5}}\tau\,\,,\,\, y(t)=\td{q}(t)
\,\,,\,\, z(t)=\td{p}(t)
 \eeq
\end{remark}

\subsection{Second element of the Painlev\'{e} $1$ hierarchy: $r_\infty=5$}
Let us consider $r_\infty=5$, i.e. $g=2$ corresponding to the second element of the Painlev\'{e} $1$ hierarchy. In this setup, the canonical choice of trivial times corresponds to $t_{\infty,8}=t_{\infty,6}=t_{\infty,4}=t_{\infty,2}=0$, $t_{\infty,7}=2$ and $t_{\infty,5}=0$. The only non-trivial isomonodromic times are $\tau_1=\frac{1}{2}t_{\infty,3}$ and $\tau_2=\frac{1}{2}t_{\infty,1}$. We have also
\bea \td{P}_2(\lambda)&=&-\lambda^5-2\tau_1\lambda^3 -2\tau_2\lambda^2\cr
Q(\lambda,\hbar)&=&-\frac{\check{p}_1(\lambda-\check{q}_2)}{\check{q}_1-\check{q}_2}-\frac{\check{p}_2(\lambda-\check{q}_1)}{\check{q}_2-\check{q}_1}=\frac{(\check{p}_2-\check{p}_1)\lambda+\check{p}_1\check{q}_2-\check{p}_2\check{q}_1}{\check{q}_1-\check{q}_2}
\eea
Coefficients $(H_{\infty,0},H_{\infty,1})$ are determined by \eqref{ReducedDefCi2}:
\bea H_{\infty,0}&=&\frac{(\check{q}_1\check{p}_2^2-\check{q}_2\check{p}_1^2)}{(\check{q}_1-\check{q}_2)}-\hbar\frac{(\check{p}_1-\check{p}_2)}{(\check{q}_1-\check{q}_2)} +(\check{q}_1+\check{q}_2)\check{q}_1\check{q}_2(\check{q}_1^2+\check{q}_2^2+2\tau_1)+2\tau_2\check{q}_1\check{q}_2 \cr
H_{\infty,1}&=&\frac{\check{p}_1^2-\check{p}_2^2}{\check{q}_1-\check{q}_2} - 2\tau_1(\check{q}_1^2+\check{q}_1\check{q}_2+\check{q}_2^2) -2\tau_2(\check{q}_1+\check{q}_2) -\check{q}_1^4-\check{q}_1^3\check{q}_2- \check{q}_1^2\check{q}_2^2-\check{q}_1\check{q}_2^3-\check{q}_2^4\cr
&&
\eea
Coefficients $\left(\nu_{\infty,1}^{(\boldsymbol{\alpha}^{\tau_1})},\nu_{\infty,2}^{(\boldsymbol{\alpha}^{\tau_1})},\nu_{\infty,1}^{(\boldsymbol{\alpha}^{\tau_2})},\nu_{\infty,2}^{(\boldsymbol{\alpha}^{\tau_2})}\right)$ are determined by Proposition \ref{Propnureduced} whose l.h.s. is trivial for $g=2$ so that
\bea
\nu_{\infty,1}^{(\boldsymbol{\alpha}^{\tau_1})}&=&\frac{2}{3} \,\,,\,\, \nu_{\infty,2}^{(\boldsymbol{\alpha}^{\tau_1})}=0 \cr
\nu_{\infty,1}^{(\boldsymbol{\alpha}^{\tau_2})}&=&0 \,\,,\,\,\nu_{\infty,2}^{(\boldsymbol{\alpha}^{\tau_2})}=2
\eea
Coefficients $\left(\mu_{1}^{(\boldsymbol{\alpha}^{\tau_1})},\mu_{2}^{(\boldsymbol{\alpha}^{\tau_1})},\mu_{1}^{(\boldsymbol{\alpha}^{\tau_2})},\mu_{2}^{(\boldsymbol{\alpha}^{\tau_2})}\right)$ are determined by \eqref{ReducedExpremu}:
\bea
\mu_1^{(\boldsymbol{\alpha}^{\tau_1})}&=&-\frac{2\check{q}_2}{3(\check{q}_1-\check{q}_2)}\,\,,\,\,\mu_2^{(\boldsymbol{\alpha}^{\tau_1})}=\frac{2\check{q}_1}{3(\check{q}_1-\check{q}_2)}\cr
\mu_1^{(\boldsymbol{\alpha}^{\tau_2})}&=&\frac{2}{\check{q}_1-\check{q}_2}\,\,,\,\,\mu_2^{(\boldsymbol{\alpha}^{\tau_2})}=-\frac{2}{\check{q}_1-\check{q}_2}
\eea

The Hamiltonians are
\bea
\text{Ham}^{(\boldsymbol\alpha^{\tau_1})}
&=&\frac23 H_{\infty,0}\cr
&=&\frac23\Bigl[P_1^2+2Q_1P_1P_2+(Q_1^2-Q_2)P_2^2+\hbar P_2
+Q_1Q_2(Q_1^2-2Q_2+2\tau_1)+2\tau_2Q_2\Bigr],\cr
\text{Ham}^{(\boldsymbol\alpha^{\tau_2})}
&=&2H_{\infty,1}\cr
&=&2\Bigl[-2P_1P_2-Q_1P_2^2-Q_1^4+3Q_1^2Q_2-Q_2^2
+2(Q_2-Q_1^2)\tau_1-2Q_1\tau_2\Bigr].
\eea
where
\beq Q_1=\check{q}_1+\check{q}_2\,\,,\,\, Q_2=\check{q}_1\check{q}_2\,\,,\,\, P_1=\frac{\check{q}_1\check{p}_1-\check{q}_2\check{p}_2}{\check{q}_1-\check{q}_2}\,\,,\,\, P_2=-\frac{\check{p}_1-\check{p}_2}{\check{q}_1-\check{q}_2}\eeq

The Lax matrices are
\small{\bea L(\lambda,\hbar)&=&\begin{pmatrix} 0&1\\ \lambda^5+2\tau_1\lambda^3 +2\tau_2\lambda^2+H_{\infty,1}\lambda+H_{\infty,0}-\frac{\hbar\check p_1}{\lambda-\check q_1}-\frac{\hbar\check p_2}{\lambda-\check q_2}& \frac{\hbar}{\lambda-\check{q}_1}+\frac{\hbar}{\lambda-\check{q}_2}\end{pmatrix}\cr
[A_{\boldsymbol{\alpha}^{\tau_1}}(\lambda,\hbar)]_{1,1}&=&\frac{2\check{p}_1\check{q}_2}{3(\check{q}_1-\check{q}_2)(\lambda-\check{q}_1)}-\frac{2\check{q}_1\check{p}_2}{3(\check{q}_1-\check{q}_2)(\lambda-\check{q}_2)}\cr
[A_{\boldsymbol{\alpha}^{\tau_1}}(\lambda,\hbar)]_{1,2}&=&-\frac{2\check{q}_2}{3(\check{q}_1-\check{q}_2)(\lambda-\check{q}_1)}+\frac{2\check{q}_1}{3(\check{q}_1-\check{q}_2)(\lambda-\check{q}_2)}\cr
[A_{\boldsymbol{\alpha}^{\tau_1}}(\lambda,\hbar)]_{2,1}&=&2\left(-\frac{\check{q}_2}{3(\check{q}_1-\check{q}_2)(\lambda-\check{q}_1)}+ \frac{\check{q}_1}{3(\check{q}_1-\check{q}_2)(\lambda-\check{q}_2)}\right)\left(\lambda^5+2\tau_1\lambda^3 +2\tau_2\lambda^2+H_{\infty,1}\lambda+H_{\infty,0}\right)\cr
&&-\hbar\left(\frac{2P_1}{3(\lambda-\check{q}_1)(\lambda-\check{q}_2)}\right)\cr
[A_{\boldsymbol{\alpha}^{\tau_1}}(\lambda,\hbar)]_{2,2}&=&\frac{2\check{q}_2\check{p}_1}{3(\check{q}_1-\check{q}_2)(\lambda-\check{q}_1)}-\frac{2\check{q}_1\check{p}_2}{3(\check{q}_1-\check{q}_2)(\lambda-\check{q}_2)}+\hbar\left(\frac{2}{3(\lambda-\check{q}_1)(\lambda-\check{q}_2)}\right)\cr
[A_{\boldsymbol{\alpha}^{\tau_2}}(\lambda,\hbar)]_{1,1}&=&2\left(-\frac{\check{p}_1}{(\check{q}_1-\check{q}_2)(\lambda-\check{q}_1)}+\frac{\check{p}_2}{(\check{q}_1-\check{q}_2)(\lambda-\check{q}_2)}\right)\cr
[A_{\boldsymbol{\alpha}^{\tau_2}}(\lambda,\hbar)]_{1,2}&=&\frac{2}{(\check{q}_1-\check{q}_2)(\lambda-\check{q}_1)}-\frac{2}{(\check{q}_1-\check{q}_2)(\lambda-\check{q}_2)}\cr
[A_{\boldsymbol{\alpha}^{\tau_2}}(\lambda,\hbar)]_{2,1}&=&2\left(\frac{1}{(\check{q}_1-\check{q}_2)(\lambda-\check{q}_1)}- \frac{1}{(\check{q}_1-\check{q}_2)(\lambda-\check{q}_2)}\right)\left(\lambda^5+2\tau_1\lambda^3 +2\tau_2\lambda^2+H_{\infty,1}\lambda+H_{\infty,0}\right)-\frac{2\hbar P_2}{(\lambda-\check q_1)(\lambda-\check q_2)}\cr
[A_{\boldsymbol{\alpha}^{\tau_2}}(\lambda,\hbar)]_{2,2}&=&2\left(-\frac{\check{p}_1}{(\check{q}_1-\check{q}_2)(\lambda-\check{q}_1)}+\frac{\check{p}_2}{(\check{q}_1-\check{q}_2)(\lambda-\check{q}_2)}\right)
\eea}
\normalsize{or} equivalently
\footnotesize{\bea\td{L}(\lambda)&=& \begin{pmatrix}\frac{\check{p}_1-\check{p}_2}{\check{q}_1-\check{q}_2}\lambda+\frac{\check{q}_1\check{p}_2-\check{q}_2\check{p}_1}{\check{q}_1-\check{q}_2}& (\lambda-\check{q}_1)(\lambda-\check{q}_2)\\\lambda^3+(\check{q}_1+\check{q}_2)\lambda^2+(\check{q}_1^2+\check{q}_1\check{q}_2+\check{q}_2^2+2\tau_1)\lambda-\frac{(\check{p}_1-\check{p}_2)^2}{(\check{q}_1-\check{q}_2)^2}+(\check{q}_1^2+\check{q}_2^2+2\tau_1)(\check{q}_1+\check{q}_2)+2\tau_2&-\frac{\check{p}_1-\check{p}_2}{\check{q}_1-\check{q}_2}\lambda-\frac{\check{q}_1\check{p}_2-\check{q}_2\check{p}_1}{\check{q}_1-\check{q}_2}\end{pmatrix}\cr
&=&\begin{pmatrix} -P_2\lambda+Q_1P_2+P_1& \lambda^2-Q_1\lambda+Q_2\\
\lambda^3+Q_1\lambda^2+(Q_1^2-Q_2+2\tau_1)\lambda -P_2^2+Q_1(Q_1^2-2Q_2)+2Q_1\tau_1+2\tau_2&  P_2\lambda-Q_1P_2-P_1
\end{pmatrix}
\eea}
\normalsize{\bea
\td{A}_{\boldsymbol{\alpha}^{\tau_1}}(\lambda)&=&\begin{pmatrix}\frac{2(\check{p}_1-\check{p}_2)}{3(\check{q}_1-\check{q}_2)}&\frac{2}{3}(\lambda-\check{q}_1-\check{q}_2)\\
\frac{2}{3}\left(\lambda^2+(\check{q}_1+\check{q}_2)\lambda+\check{q}_1^2+\check{q}_2^2+2\tau_1\right)& -\frac{2(\check{p}_1-\check{p}_2)}{3(\check{q}_1-\check{q}_2)}  \end{pmatrix}\cr
&=&\begin{pmatrix}-\frac{2}{3}P_2& \frac{2}{3}(\lambda-Q_1)\\
\frac{2}{3}\left(\lambda^2+Q_1\lambda+Q_1^2-2Q_2+2\tau_1\right)& \frac{2}{3}P_2\end{pmatrix} \cr
\td{A}_{\boldsymbol{\alpha}^{\tau_2}}(\lambda)&=&\begin{pmatrix}0&2\\ 2(\lambda+2(\check{q}_1+\check{q}_2))&0 \end{pmatrix}=\begin{pmatrix}0&2\\ 2(\lambda+2Q_1)&0 \end{pmatrix}
\eea}

\subsection{Third element of the Painlev\'{e} $1$ hierarchy: $r_\infty=6$}
Let us consider $r_\infty=6$, i.e. $g=3$ corresponding to the third element of the Painlev\'{e} $1$ hierarchy. In particular, this case is the first case where the Hamiltonians are non-trivial linear combinations of the coefficients $\left(H_{\infty,k}\right)_{0\leq k\leq r_\infty-4}$, this is due to the fact that the matrix of Proposition \ref{Propnureduced} is no longer diagonal. Indeed we have
\bea \text{Ham}^{(\boldsymbol{\alpha}^{\tau_1})}&=&\frac{2}{5}H_{\infty,0}-\frac{2}{5}\tau_1H_{\infty,2}\cr
\text{Ham}^{(\boldsymbol{\alpha}^{\tau_2})}&=&\frac{2}{3}H_{\infty,1}\cr
\text{Ham}^{(\boldsymbol{\alpha}^{\tau_3})}&=&2H_{\infty,2}
\eea
In this setup, the canonical choice of trivial times corresponds to $t_{\infty,10}=t_{\infty,8}=t_{\infty,6}=t_{\infty,4}=t_{\infty,2}=0$, $t_{\infty,9}=2$ and $t_{\infty,7}=0$. The only non-trivial isomonodromic times are $\tau_1=\frac{1}{2}t_{\infty,5}$, $\tau_2=\frac{1}{2}t_{\infty,3}$ and $\tau_3=\frac{1}{2}t_{\infty,1}$. For compactness, we shall only present results expressed in terms of the symmetric Darboux coordinates $(Q_1,Q_2,Q_3,P_1,P_2,P_3)$. One may recover the expression in terms of shifted Darboux coordinates using
\bea Q_1&=&\check{q}_1+\check{q}_2+\check{q}_3\,\,,\,\, Q_2=\check{q}_1\check{q}_2+\check{q}_1\check{q}_3+\check{q}_2\check{q}_3\,\,,\,\, Q_3=\check{q}_1\check{q}_2\check{q}_3\cr
P_1&=&\frac{\check{q}_1^2(\check{q}_2-\check{q}_3)\check{p}_1- \check{q}_2^2(\check{q}_1-\check{q}_3)\check{p}_2+\check{q}_3^2(\check{q}_1-\check{q}_2)\check{p}_3}{(\check{q}_1-\check{q}_2)(\check{q}_1-\check{q}_3)(\check{q}_2-\check{q}_3)}\cr
P_2&=&-\frac{\check{q}_1(\check{q}_2-\check{q}_3)\check{p}_1- \check{q}_2(\check{q}_1-\check{q}_3)\check{p}_2+\check{q}_3(\check{q}_1-\check{q}_2)\check{p}_3}{(\check{q}_1-\check{q}_2)(\check{q}_1-\check{q}_3)(\check{q}_2-\check{q}_3)}\cr
P_3&=&\frac{(\check{q}_2-\check{q}_3)\check{p}_1- (\check{q}_1-\check{q}_3)\check{p}_2+(\check{q}_1-\check{q}_2)\check{p}_3}{(\check{q}_1-\check{q}_2)(\check{q}_1-\check{q}_3)(\check{q}_2-\check{q}_3)}
\eea
We have
\bea \td{P}_2(\lambda)&=&-\lambda^7-2\tau_1\lambda^5 -2\tau_2\lambda^4-(\tau_1^2+2\tau_3)\lambda^3\cr
Q(\lambda)&=&-P_3\lambda^2+(P_3Q_1+P_2)\lambda-P_3Q_2-P_2Q_1-P_1
\eea
The Hamiltonians are
\footnotesize{\bea \text{Ham}^{(\boldsymbol{\alpha}^{\tau_1})}&=&\frac{2}{5}\Big[ (-Q_1Q_3+Q_2^2)P_3^2+2(Q_1+Q_2)P_1P_3+Q_1^2P_2^2+2Q_1P_1P_2+P_1^2+2(Q_1Q_2-Q_3)P_2P_3-2Q_1P_1P_3\cr
&&-Q_3(Q_1^4-3Q_1^2Q_2+2Q_1Q_3+Q_2^2)+2(Q_1^2-Q_2)\tau_1\tau_2+2Q_1\tau_1\tau_3+Q_1\tau_1^3+(2Q_1^3-4Q_1Q_2+Q_3)\tau_1^2\cr
&&-(-Q_1^5+4Q_1^3Q_2+2Q_1P_2P_3+Q_2P_3^2-Q_1^2Q_3-3Q_1Q_2^2+2P_1P_3+P_2^2)\tau_1-2Q_1Q_3\tau_2-2Q_3\tau_3\cr
&&+\hbar(Q_1P_3+2P_2)\Big]\cr
\text{Ham}^{(\boldsymbol{\alpha}^{\tau_2})}&=&\frac{2}{3}\Big[ -2Q_1^2P_2P_3-2Q_1P_1P_3-2P_1P_2+(Q_3-Q_1Q_2)P_3^2-2Q_1P_2^2+4Q_1Q_2Q_3+Q_1^4Q_2-3Q_1^2Q_2^2\cr
&&-Q_1^3Q_3+Q_2^3-Q_3^2+Q_2\tau_1^2+2(Q_1^2Q_2-Q_1Q_3-Q_2^2)\tau_1+2(Q_1Q_2-Q_3)\tau_2+2Q_2\tau_3-\hbar P_3\Big]\cr
\text{Ham}^{(\boldsymbol{\alpha}^{\tau_3})}&=&2\Big[2Q_1P_2P_3+Q_2P_3^2+2P_1P_3+P_2^2-Q_1^5+4Q_1^3Q_2-3Q_1^2Q_3-3Q_1Q_2^2+2Q_2Q_3\cr
&&-Q_1\tau_1^2+2(2Q_1Q_2-Q_1^3-Q_3)\tau_1+2(Q_2-Q_1^2)\tau_2-2Q_1\tau_3\Big]
\eea}
\normalsize{The} Lax matrices are
\bea \td{L}_{1,1}(\lambda)&=& P_3\lambda^2-(P_2+Q_1P_3)\lambda+P_1+Q_2P_3+Q_1P_2\cr
\td{L}_{1,2}(\lambda)&=&\lambda^3-Q_1\lambda^2+Q_2\lambda-Q_3\cr
\td{L}_{2,1}(\lambda)&=&\lambda^4+Q_1\lambda^3+(Q_1^2-Q_2+2\tau_1)\lambda^2+(-P_3^2+Q_1^3+Q_3-2Q_1Q_2 +2Q_1\tau_1+2\tau_2 )\lambda\cr
&&+2P_2P_3+Q_1P_3^2+Q_1^4-3Q_2Q_1^2+2Q_3Q_1+Q_2^2+\tau_1^2+2(Q_1^2-Q_2)\tau_1+2Q_1\tau_2+2\tau_3\cr
\td{L}_{2,2}(\lambda)&=& -P_3\lambda^2+(P_2+Q_1P_3)\lambda-P_1-Q_2P_3-Q_1P_2
\eea
and
\footnotesize{\bea 
\td{A}_{\boldsymbol{\alpha}^{\tau_1}}(\lambda)&=&\begin{pmatrix}\frac{2}{5}\left(P_3\lambda-Q_1P_3-P_2\right)&\frac{2}{5}\left(\lambda^2-Q_1\lambda+Q_2-\tau_1 \right)\\
\frac{2}{5}\left(\lambda^3+Q_1\lambda^2+(Q_1^2-Q_2+\tau_1)\lambda-P_3^2+Q_1^3-2Q_1Q_2+2Q_3+2\tau_2\right)& -\frac{2}{5}\left(P_3\lambda-Q_1P_3-P_2\right)\end{pmatrix}\cr
\td{A}_{\boldsymbol{\alpha}^{\tau_2}}(\lambda)&=&\begin{pmatrix} \frac{2}{3}P_3&\frac{2}{3}\left(\lambda-Q_1\right)\\ \frac{2}{3}\left(\lambda^2+Q_1\lambda+Q_1^2-2Q_2+2\tau_1\right)& -\frac{2}{3}P_3\end{pmatrix}\cr
 \td{A}_{\boldsymbol{\alpha}^{\tau_3}}(\lambda)&=&\begin{pmatrix} 0& 2\\ 2(\lambda+2Q_1)&0\end{pmatrix}
\eea}
\normalsize{}

\section{Outlooks}
In this article, we complemented the results of \cite{MarchalOrantinAlameddine2022} by dealing with the generic one-pole ramified stratum of meromorphic connections in $\mathfrak{gl}_2(\mathbb{C})$ obtaining explicit expressions of the Hamiltonians and Lax pairs in various sets of Darboux coordinates. Moreover, we provided a reduction of the initial space of irregular times (of dimension $2g+4$) to a set of non-trivial isomonodromic times of dimension $g$. On this stratum and locally in the deformation parameters, we recover the fact that meromorphic connections in $\mathfrak{gl}_2(\mathbb{C})$ are equivalent at the level of Hamiltonian systems to meromorphic connections in $\mathfrak{sl}_2(\mathbb{C})$, a point that was already raised in \cite{MarchalOrantinAlameddine2022} in the untwisted case. The method used in the present article opens the way to several generalizations:
\begin{itemize}\item This article and \cite{MarchalOrantinAlameddine2022} give explicit constructions on the corresponding generic ramified and untwisted strata. Extensions to lower formal-type strata and more general pole configurations remain to be studied. Another natural question is extension to $\mathfrak{gl}_d(\mathbb C)$ with $d\geq3$. In principle, a similar strategy shall be used but it is unclear if all technical issues might be overcome when the dimension is arbitrary, especially in the twisted case where the underlying geometric construction is far less understood. A natural case that should be closely related to the present work are meromorphic connections in $\mathfrak{gl}_d$ with a unique pole at infinity and whose leading order is diagonal with distinct eigenvalues except for a lower triangular Jordan block of size $2$ with a double eigenvalue. Under the analogous generic nonvanishing condition on that block, the local fundamental form is expected to use the local variable $z=\lambda^{\frac{1}{2}}$ and with a combination of irregular times similar to the one introduced in this paper and untwisted times for the diagonalizable part. However, since the rank of the connection is higher, obtaining the form of the Lax matrices and solving the compatibility equations remains a challenging issue. We let this very exciting question of generalizing to arbitrary type of singularities for future works that would in particular open the way to explicit formulas for $(p,q)$ models that are of great importance in random matrix models and in theoretical physics.
\item This article deals with meromorphic connections in $\mathfrak{gl}_2(\mathbb{C})$. However, one may be interested in a more general abstract setup with any Lie algebra and not only a matrix representation of it. In this case, we believe that most of the results shall be generalized upon the adequate quantities and terminology. 
\item As mentioned in Section \ref{SectionTR} and similarly to \cite{MarchalOrantinAlameddine2022}, this article makes the connection with formal WKB expansions and $\hbar$-transseries via topological recursion. At the level of isomonodromic deformations, the introduction of the formal parameter $\hbar$ is done by a simple rescaling of irregular times (Section \ref{SectionIntrohbar}). A better understanding of the role of $\hbar$ and its limit to $0$ in the context of meromorphic connections would be interesting in order to address the issue of resumation and analytical properties associated to the formal $\hbar$ (trans)series.
\item In \cite{marchal2023isomonodromicisospectraldeformationsmeromorphic} the authors linked, in the untwisted setting, the isomonodromic deformations with the isospectral approach motivated by the Montr\'{e}al school \cite{BertolaHarnadHurtubise2022} and Yamakawa \cite{Yamakawa2019FundamentalTwoForms} through a time dependent change of Darboux coordinates. It would be interesting to extend this work to more general settings covering the twisted case considered in this article.
\end{itemize}

\section*{Statements and Declarations}
The authors declare that they have no conflict of interest.\\
Data sharing not applicable to this article as no datasets were generated or analysed during the current study.

\sloppy{\section*{Acknowledgements} The authors would like to thank deeply N. Orantin for suggesting many improvements of the present paper. The authors would also like to thank G. Rembado, J. Dou\c{c}ot for fruitful discussions.}

\newpage

\appendix
\renewcommand{\theequation}{\thesection-\arabic{equation}}

\section{Proof of Proposition \ref{PropPsiAsymp}}\label{AppendixAsymptoticsWaveFunctions}
Put $z=\lambda^{1/2}$ and $t_*=t_{\infty,2r_\infty-3}\neq0$. For a representative normalized by \eqref{NormalizationInfty}, first apply the shearing $S(z)=\operatorname{diag}(z,1)$. The leading non-scalar coefficient of the sheared system is
\beq
\begin{pmatrix}0&1\\\frac{t_*^2}{4}&0\end{pmatrix}z^{2g+1}.
\eeq
The ordered left eigenvectors corresponding to $-\frac{t_*}{2}$ and $\frac{t_*}{2}$ give
\beq
H_0=\begin{pmatrix}1&-\frac{2}{t_*}\\1&\frac{2}{t_*}\end{pmatrix},
\qquad H_0^{-1}=\begin{pmatrix}\frac{1}{2}&\frac{1}{2}\\-\frac{t_*}{4}&\frac{t_*}{4}\end{pmatrix}.
\eeq
The subsequent regular formal diagonalization in Proposition \ref{PropDiago} may therefore be chosen in the form
\beq
G_\infty(z)=H(z)S(z),\qquad H(z)=H_0+O(z^{-1}).
\eeq
In particular, its leading coefficients and those of its inverse are
\beq
G_{\infty,-1}=\begin{pmatrix}1&0\\1&0\end{pmatrix},
\qquad
[z^0]G_\infty^{-1}=\begin{pmatrix}0&0\\-\frac{t_*}{4}&\frac{t_*}{4}\end{pmatrix}.
\eeq
The two entries on the bottom row of the latter matrix have opposite signs.

After introducing $\hbar$, set
\beq
\Theta_1=-\frac1\hbar\sum_{k=1}^{2r_\infty-2}\frac{t_{\infty,k}}kz^k,
\qquad
\Theta_2=-\frac1\hbar\sum_{k=1}^{2r_\infty-2}\frac{(-1)^kt_{\infty,k}}kz^k.
\eeq
Using the formal normalization $\Psi_\infty^{(\mathrm{reg})}=I_2$ gives
\beq\label{PsiPsi}
\td\Psi
=S^{-1}H^{-1}\operatorname{diag}
\left(z^{\frac{1}{2}}e^{\Theta_1},z^{\frac{1}{2}}e^{\Theta_2}\right).
\eeq
Thus its first row, which is also the first row of $\Psi$ and $\check\Psi$, satisfies
\bea\label{AsympttdPsi}
\psi_1&=&\bigl(\frac{1}{2}+O(z^{-1})\bigr)z^{-\frac{1}{2}}e^{\Theta_1},\cr
\psi_2&=&\bigl(\frac{1}{2}+O(z^{-1})\bigr)z^{-\frac{1}{2}}e^{\Theta_2}.
\eea
Taking formal logarithms yields exactly \eqref{PsiAsymptotics0}, since $-\tfrac12\log z=-\tfrac14\log\lambda$. Constant column normalizations only change the $O(1)$ terms. This also justifies differentiating those terms as formal series: their nonconstant parts are $O(z^{-1})$, whose $\lambda$ derivatives are $O(\lambda^{-3/2})$.

\section{Proof of Proposition \ref{PropLaxMatrix}}\label{AppendixLForm}
From Proposition \ref{PropPsiAsymp}, the formal wave functions obey
\bea\label{Asymptpsi1psi2}
\psi_1(\lambda)&=&\exp\left(-\frac1\hbar\sum_{k=1}^{2r_\infty-2}
\frac{t_{\infty,k}}{k}\lambda^{k/2}-\frac14\log\lambda+O(1)\right),\cr
\psi_2(\lambda)&=&\exp\left(-\frac1\hbar\sum_{k=1}^{2r_\infty-2}
\frac{(-1)^kt_{\infty,k}}{k}\lambda^{k/2}-\frac14\log\lambda+O(1)\right).
\eea
With $R(\lambda)=\underset{j=1}{\overset{g}{\prod}}(\lambda-q_j)$, the Wronskian is
\beq\label{Wronksian}
W(\lambda)=W_0R(\lambda)
\exp\left(\frac1\hbar\int_0^\lambda\td P_1(s)\,ds\right).
\eeq
Thus $L_{2,2}=\hbar\partial_\lambda\log W=\td P_1+\hbar \frac{R'}{R}$.
Put $Y_i=\hbar\partial_\lambda\log\psi_i$. The two Riccati equations
$Y_i^2+\hbar Y_i'=L_{2,1}+L_{2,2}Y_i$ imply
\beq
L_{2,1}=-Y_1Y_2+
\hbar\frac{Y_2Y_1'-Y_1Y_2'}{Y_2-Y_1}.
\eeq
Their formal expansions are
\bea\label{Y1Y2}
Y_1&=&-\frac12\sum_{k=1}^{2r_\infty-2}t_{\infty,k}\lambda^{\frac{k}{2}-1}
-\frac{\hbar}{4\lambda}+O\left(\lambda^{-\frac{3}{2}}\right),\cr
Y_2&=&-\frac12\sum_{k=1}^{2r_\infty-2}(-1)^kt_{\infty,k}\lambda^{\frac{k}{2}-1}
-\frac{\hbar}{4\lambda}+O\left(\lambda^{-\frac{3}{2}}\right).
\eea
Let $a=t_{\infty,2r_\infty-2}$. At order $\lambda^g$, the product term contributes $-\frac{\hbar a\lambda^g}{4}$, and
\beq
\hbar\frac{Y_2Y_1'-Y_1Y_2'}{Y_2-Y_1}
=-\frac{\hbar a}{4}\lambda^g+O(\lambda^{g-1/2}).
\eeq
The polynomial terms of degree at least $g$ in the product of the time-dependent sums give $-\td P_2^{(0)}$. Since $L_{2,1}$ is rational in $\lambda$, its expansion has only integer powers. Consequently
\beq
L_{2,1}=-\td P_2^{(0)}(\lambda)
-\frac{\hbar a}{2}\lambda^g+O(\lambda^{g-1})
=-\td P_2(\lambda,\hbar)+O(\lambda^{g-1}).
\eeq
For the finite poles, the companion-gauge formula gives
\beq
L_{2,1}=-\det\td L+\hbar\partial_\lambda\td L_{1,1}
-\hbar\td L_{1,1}\frac{R'}R.
\eeq
The roots of $R$ are simple and $\td L_{1,1}(q_j)=p_j$, so the residue at $q_j$ is $-\hbar p_j$. There are no other finite poles. Subtracting these principal parts and using the asymptotics yields
\beq
L_{2,1}=-\td P_2(\lambda,\hbar)
+\sum_{k=0}^{g-1}H_{\infty,k}\lambda^k
-\hbar\sum_{j=1}^g\frac{p_j}{\lambda-q_j}.
\eeq
For $g=0$, both sums are empty. This proves Proposition \ref{PropLaxMatrix}.

\section{Proof of Proposition \ref{PropAsymptoticExpansionA12}}\label{AppendixExpansionA}
Let us first observe that the entry $\left[A_{\boldsymbol{\alpha}}(\lambda)\right]_{1,2}$ is given by
\beq
\left[A_{\boldsymbol{\alpha}}(\lambda)\right]_{1,2}=\frac{W_{\boldsymbol{\alpha}}(\lambda)}{W(\lambda)}= \frac{Z_{\boldsymbol{\alpha},2}(\lambda)-Z_{\boldsymbol{\alpha},1}(\lambda) }{Y_2(\lambda)-Y_1(\lambda)}
\eeq
where we have defined
\bea 
Z_{\boldsymbol{\alpha},i}(\lambda)&=& \frac{\mathcal{L}_{\boldsymbol{\alpha}}[\psi_i(\lambda)]}{\psi_i(\lambda)} \,\,,\,\, \forall \, i\in\llbracket 1,2\rrbracket\cr
W_{\boldsymbol{\alpha}}(\lambda)&=& \mathcal{L}_{\boldsymbol{\alpha}}[\psi_2(\lambda)] \psi_1(\lambda) - \mathcal{L}_{\boldsymbol{\alpha}}[\psi_1(\lambda)] \psi_2(\lambda)
\eea 
and from Proposition \ref{PropPsiAsymp} we have:
\bea \label{Z1Z2} Z_{\boldsymbol{\alpha},1}(\lambda)&=&-\sum_{k=1}^{2r_\infty-2}\frac{\alpha_{\infty,k}}{k}\lambda^{\frac{k}{2}}+O(1)\cr
Z_{\boldsymbol{\alpha},2}(\lambda)&=&-\sum_{k=1}^{2r_\infty-2}(-1)^k\frac{\alpha_{\infty,k}}{k}\lambda^{\frac{k}{2}}+O(1)
\eea
so that 
\beq Z_{\boldsymbol{\alpha},2}(\lambda)-Z_{\boldsymbol{\alpha},1}(\lambda)=2\sum_{j=1}^{r_\infty-1}\frac{\alpha_{\infty,2j-1}}{2j-1}\lambda^{j-\frac{1}{2}}+O(1)\eeq
Thus, since
\beq \left[A_{\boldsymbol{\alpha}}(\lambda)\right]_{1,2}\left(Y_2(\lambda)-Y_1(\lambda)\right)=Z_{\boldsymbol{\alpha},2}(\lambda)-Z_{\boldsymbol{\alpha},1}(\lambda)\eeq
using \eqref{Y1Y2}
\beq \label{DiffY2Y1}Y_2(\lambda)-Y_1(\lambda)=\sum_{j=1}^{r_\infty-1}  t_{\infty,2j-1} \lambda^{j-\frac{3}{2}}+O\left(\lambda^{-\frac{3}{2}}\right)\eeq
we obtain the form of the entry
\beq\label{A12Asymptotic} A_{\boldsymbol{\alpha}}(\lambda,\hbar)_{1,2}=\sum_{j=-1}^{r_\infty-3} \nu^{(\boldsymbol{\alpha})}_{\infty,j}\lambda^{-j} +O\left(\lambda^{-(r_\infty-2)}\right)\eeq
where the coefficients $\left(\nu^{(\boldsymbol{\alpha})}_{\infty,k}\right)_{-1\leq k\leq r_\infty-3}$ are recursively determined by
\beq \left(\sum_{j=-1}^{r_\infty-3}\nu^{(\boldsymbol{\alpha})}_{\infty,j}\lambda^{-j}+O\left(\lambda^{-r_\infty+2}\right)\right)\left(\sum_{k=1}^{r_\infty-1}t_{\infty,2k-1}\lambda^{k-1}+O\left(\lambda^{-1}\right)\right)=2\sum_{k=1}^{r_\infty-1}\frac{\alpha_{\infty,2k-1}}{2k-1} \lambda^k +O\left(1\right).\eeq
These relations may be rewritten in a $(r_\infty-1) \times (r_\infty-1)$ lower triangular Toeplitz matrix form:
\beq \label{nualpha} M_\infty \begin{pmatrix} \nu^{(\boldsymbol{\alpha})}_{\infty,-1}\\ \vdots \\ \nu^{(\boldsymbol{\alpha})}_{\infty, r_\infty-3}\end{pmatrix}= \begin{pmatrix}\frac{2\alpha_{\infty,2r_\infty-3}}{(2r_\infty-3)}\\ \frac{2\alpha_{\infty,2r_\infty-5}}{(2r_\infty-5)}\\ \vdots \\ \frac{2\alpha_{\infty,1}}{1} \end{pmatrix} \, \text{ with }\,
M_\infty=\begin{pmatrix}t_{\infty,2r_\infty-3}&0&\dots& &0\\
t_{\infty,2r_\infty-5}& t_{\infty,2r_\infty-3}&0& \ddots& \vdots\\
\vdots& \ddots& \ddots&&\vdots \\
t_{\infty,3}& & \ddots &\ddots&0 \\
t_{\infty,1}& t_{\infty,3} & \dots & &t_{\infty,2r_\infty-3}\end{pmatrix}
\eeq

\section{Proof of Proposition \ref{Propcalpha}}\label{AppendixA11}
A straightforward computation shows that
\beq \left[A_{\boldsymbol{\alpha}}(\lambda)\right]_{1,1}= \frac{Z_{\boldsymbol{\alpha},1}(\lambda) Y_2(\lambda) -Z_{\boldsymbol{\alpha},2}(\lambda) Y_1(\lambda)}{Y_2(\lambda) -Y_1(\lambda)}\eeq
which we rewrite as
\beq \left[A_{\boldsymbol{\alpha}}(\lambda)\right]_{1,1}(Y_2(\lambda)-Y_1(\lambda))= Z_{\boldsymbol{\alpha},1}(\lambda) Y_2(\lambda) -Z_{\boldsymbol{\alpha},2}(\lambda) Y_1(\lambda)\eeq
We proceed using  \eqref{Y1Y2} and \eqref{Z1Z2} that give
\bea  Z_{\boldsymbol{\alpha},1}(\lambda) Y_2(\lambda) -Z_{\boldsymbol{\alpha},2}(\lambda) Y_1(\lambda)&=&\frac{1}{2}\sum_{j=1}^{2r_\infty-2}\sum_{i=1}^{2r_\infty-2}\frac{\alpha_{\infty,j}}{j}t_{\infty,i}( (-1)^{i}-(-1)^{j})\lambda^{\frac{i+j}{2}-1} +O\left(\lambda^{r_\infty-2}\right)\cr
&=&\sum_{s=1}^{r_\infty-1}\sum_{m=1}^{r_\infty-1}\left(\frac{\alpha_{\infty,2s-1}}{2s-1}t_{\infty,2m}-\frac{\alpha_{\infty,2s}}{2s}t_{\infty,2m-1}\right)\lambda^{m+s-\frac{3}{2}} +O\left(\lambda^{r_\infty-2}\right)\cr
&&
\eea
We make use \eqref{DiffY2Y1} also in order to obtain
\beq \label{A11Asymptotic} A_{\boldsymbol{\alpha}}(\lambda,\hbar)_{1,1}=\sum_{j=1}^{r_\infty-1} c^{(\boldsymbol{\alpha})}_{\infty,j}\lambda^{j} +O\left(1\right)\eeq
where the coefficients $\left(c^{(\boldsymbol{\alpha})}_{\infty,k}\right)_{1\leq k\leq r_\infty-1}$ are recursively determined by
\footnotesize{\beq \left(\sum_{j=1}^{r_\infty-1}c^{(\boldsymbol{\alpha})}_{\infty,j}\lambda^{j}+O\left(1\right)\right)\left(\sum_{k=1}^{r_\infty-1}t_{\infty,2k-1}\lambda^{k-\frac{3}{2}}+O\left(\lambda^{-\frac{3}{2}}\right)\right)=\sum_{s=1}^{r_\infty-1}\sum_{m=1}^{r_\infty-1}\left(\frac{\alpha_{\infty,2s-1}}{2s-1}t_{\infty,2m}-\frac{\alpha_{\infty,2s}}{2s}t_{\infty,2m-1}\right)\lambda^{m+s-\frac{3}{2}} +O\left(\lambda^{r_\infty-2}\right)\eeq}
\normalsize{These} relations may be rewritten in a $(r_\infty-1) \times (r_\infty-1)$ lower triangular Toeplitz matrix form:
\beq M_\infty \begin{pmatrix} c^{(\boldsymbol{\alpha})}_{\infty,r_\infty-1}\\ \vdots\\c^{(\boldsymbol{\alpha})}_{\infty,k}  \\\vdots \\ c^{(\boldsymbol{\alpha})}_{\infty, 1}\end{pmatrix}=\begin{pmatrix}\frac{\alpha_{\infty,2r_\infty-3}}{2r_\infty-3}t_{\infty,2r_\infty-2}-\frac{\alpha_{\infty,2r_\infty-2}}{2r_\infty-2}t_{\infty,2r_\infty-3}\\ \vdots\\ \underset{m=k}{\overset{r_\infty-1}{\sum}}\left(\frac{\alpha_{\infty,2k+2r_\infty-2m-3}}{2k+2r_\infty-2m-3}t_{\infty,2m}-\frac{\alpha_{\infty,2k+2r_\infty-2m-2}}{2k+2r_\infty-2m-2}t_{\infty,2m-1}\right)  \\ \vdots\\ \underset{m=1}{\overset{r_\infty-1}{\sum}}\left(\frac{\alpha_{\infty,2r_\infty-2m-1}}{2r_\infty-2m-1}t_{\infty,2m}-\frac{\alpha_{\infty,2r_\infty-2m}}{2r_\infty-2m}t_{\infty,2m-1}\right) \end{pmatrix}
\eeq

Finally, the coefficients $\left(\rho^{(\boldsymbol{\alpha})}_j\right)_{1\leq j\leq g}$ are obtained by looking at order $(\lambda-q_j)^{-3}$ of $\mathcal{L}_{\boldsymbol{\alpha}}[L_{2,1}(\lambda)]$.

\section{Proof of Theorem \ref{HamTheorem}}\label{AppendixHamiltonian}
This appendix is devoted for the proof of Theorem \ref{HamTheorem}. 

\subsection{Preliminary results}
We start with the following lemma:

\begin{lemma}\label{PropsumC} For all $j\in \llbracket 1, g\rrbracket$:
\beq 
\sum_{k=0}^{r_\infty-4}\sum_{i=1}^gH_{\infty,k}q_i^k\partial_{q_j} \mu^{\boldsymbol{(\alpha)}}_i=-\mu^{\boldsymbol{(\alpha)}}_j\sum_{k=0}^{r_\infty-4}kH_{\infty,k} q_j^{k-1}
\eeq
\end{lemma}

\begin{proof}The proof follows from the expression relating the coefficients $(\nu^{(\boldsymbol{\alpha})}_{p,k})$ and $(\mu^{(\boldsymbol{\alpha})}_{p,k})$ given by \eqref{RelationNuMuMatrixForm}. Taking the derivative relatively to $q_j$ and using the fact that $(\nu^{(\boldsymbol{\alpha})}_{\infty,k})_{-1\leq k\leq r_\infty-3}$ are independent of $q_j$ gives:
\beq  \forall\, k\in \llbracket 0, r_\infty-4\rrbracket\,:\, \sum_{i=1}^g (\partial_{q_j}\mu^{(\boldsymbol{\alpha})}_i) q_i^{k}=
-k \mu^{(\boldsymbol{\alpha})}_j q_j^{k-1}\eeq
Thus
\beq \sum_{k=0}^{r_\infty-4}\sum_{i=1}^gH_{\infty,k}q_i^k\partial_{q_j} \mu^{\boldsymbol{(\alpha)}}_i=-\mu^{(\boldsymbol{\alpha})}_j \left(\sum_{k=0}^{r_\infty-4} kH_{\infty,k}q_j^{k-1}\right)
\eeq
so that the lemma is proved.
\end{proof}

We may now provide an alternative expression for $\mathcal{L}_{\boldsymbol{\alpha}}[p_j]$:

\begin{proposition}\label{PropLpjbis} Let $j\in \llbracket 1,g\rrbracket$, we have an alternative expression for $\mathcal{L}_{\boldsymbol{\alpha}}[p_j]$:
\bea \label{Lpjbis} \mathcal{L}_{\boldsymbol{\alpha}}[p_j]&=&\hbar \sum_{i\neq j}\frac{(\mu^{(\boldsymbol{\alpha})}_i+\mu^{(\boldsymbol{\alpha})}_j)(p_i-p_j)}{(q_j-q_i)^2} +\frac{\hbar}{2}\displaystyle{\sum_{\substack{(r,s)\in \llbracket 1,g\rrbracket^2 \\ r\neq s }}} \frac{(p_s-p_r)(\partial_{q_j}\mu^{\boldsymbol{(\alpha)}}_r+\partial_{q_j}\mu^{\boldsymbol{(\alpha)}}_s)}{q_s-q_r}\cr
&&-\mu^{(\boldsymbol{\alpha})}_j\left(\td{P}_2'(q_j)-p_j \td{P}_1'(q_j)\right)-\sum_{r=1}^g (\partial_{q_j} \mu^{(\boldsymbol{\alpha})}_r)\left( \td{P}_2(q_r)+ p_r^2-\td{P}_1(q_r)p_r\right)\cr
&&+\hbar \nu^{(\boldsymbol{\alpha})}_{\infty,-1}p_j+\hbar \sum_{k=1}^{r_\infty-1}kc^{(\boldsymbol{\alpha})}_{\infty,k}q_j^{k-1}
\eea 
\end{proposition}

\begin{proof}Using Lemma \ref{PropsumC}, the expression \eqref{Lpj} for $\mathcal{L}[p_j]$ becomes:
\bea \mathcal{L}_{\boldsymbol{\alpha}}[p_j]&=&\hbar \sum_{i\neq j}\frac{(\mu^{(\boldsymbol{\alpha})}_i+\mu^{(\boldsymbol{\alpha})}_j)(p_i-p_j)}{(q_j-q_i)^2} +\mu^{(\boldsymbol{\alpha})}_j\left(p_j \td{P}_1'(q_j)-\td{P}_2'(q_j) \right)+\hbar \nu^{(\boldsymbol{\alpha})}_{\infty,-1}p_j\cr
&&+\hbar \sum_{k=1}^{r_\infty-1}kc^{(\boldsymbol{\alpha})}_{\infty,k}q_j^{k-1}-\sum_{i=1}^g (\partial_{q_j}\mu^{\boldsymbol{(\alpha)}}_i)\sum_{k=0}^{r_\infty-4}H_{\infty,k}q_i^k
\eea
We now use \eqref{DefCi} to get
\bea \mathcal{L}_{\boldsymbol{\alpha}}[p_j]&=&\hbar \sum_{i\neq j}\frac{(\mu^{(\boldsymbol{\alpha})}_i+\mu^{(\boldsymbol{\alpha})}_j)(p_i-p_j)}{(q_j-q_i)^2} +\mu^{(\boldsymbol{\alpha})}_j\left(p_j \td{P}_1'(q_j)-\td{P}_2'(q_j)\right)\cr
&&+\hbar \nu^{(\boldsymbol{\alpha})}_{\infty,-1}p_j+\hbar \sum_{k=1}^{r_\infty-1}kc^{(\boldsymbol{\alpha})}_{\infty,k}q_j^{k-1}\cr
&&-\sum_{i=1}^g (\partial_{q_j}\mu^{\boldsymbol{(\alpha)}}_i)\Big[p_i^2-\td{P}_1(q_i)p_i +\td{P}_2(q_i)+\hbar \sum_{r\neq i}\frac{p_r-p_i}{q_i-q_r}\Big]\cr
&=&\hbar \sum_{i\neq j}\frac{(\mu^{(\boldsymbol{\alpha})}_i+\mu^{(\boldsymbol{\alpha})}_j)(p_i-p_j)}{(q_j-q_i)^2} +\mu^{(\boldsymbol{\alpha})}_j\left(p_j \td{P}_1'(q_j)-\td{P}_2'(q_j)\right)\cr
&&+\hbar \nu^{(\boldsymbol{\alpha})}_{\infty,-1}p_j+\hbar \sum_{k=1}^{r_\infty-1}kc^{(\boldsymbol{\alpha})}_{\infty,k}q_j^{k-1}-\sum_{i=1}^g (\partial_{q_j}\mu^{\boldsymbol{(\alpha)}}_i)\left(p_i^2-\td{P}_1(q_i)p_i+\td{P}_2(q_i)\right)\cr
&&+\hbar\sum_{i=1}^g (\partial_{q_j}\mu^{\boldsymbol{(\alpha)}}_i)\sum_{r\neq i}\frac{p_r-p_i}{q_r-q_i}
\eea
For each pair $i\neq r$, split
\beq
\partial_{q_j}\mu_i^{(\boldsymbol\alpha)}
=\frac12\bigl(\partial_{q_j}\mu_i^{(\boldsymbol\alpha)}
-\partial_{q_j}\mu_r^{(\boldsymbol\alpha)}\bigr)
+\frac12\bigl(\partial_{q_j}\mu_i^{(\boldsymbol\alpha)}
+\partial_{q_j}\mu_r^{(\boldsymbol\alpha)}\bigr).
\eeq
The ratio $(p_r-p_i)/(q_r-q_i)$ is symmetric under exchange of $i$ and $r$, so the antisymmetric contribution vanishes in the double sum. We obtain
\bea \mathcal{L}_{\boldsymbol{\alpha}}[p_j]&=&\hbar \sum_{i\neq j}\frac{(\mu^{(\boldsymbol{\alpha})}_i+\mu^{(\boldsymbol{\alpha})}_j)(p_i-p_j)}{(q_j-q_i)^2}-\mu^{(\boldsymbol{\alpha})}_j\left(\td{P}_2'(q_j)-p_j \td{P}_1'(q_j)\right)\cr
&&+\hbar \nu^{(\boldsymbol{\alpha})}_{\infty,-1}p_j+\hbar \sum_{k=1}^{r_\infty-1}kc^{(\boldsymbol{\alpha})}_{\infty,k}q_j^{k-1}-\sum_{i=1}^g (\partial_{q_j}\mu^{\boldsymbol{(\alpha)}}_i)\left(p_i^2-\td{P}_1(q_i)p_i+\td{P}_2(q_i)\right)\cr
&&+\frac{\hbar}{2}\sum_{i=1}^g\sum_{r\neq i} \frac{(p_r-p_i)(\partial_{q_j}\mu^{\boldsymbol{(\alpha)}}_i+\partial_{q_j}\mu^{\boldsymbol{(\alpha)}}_r)}{q_r-q_i}
\eea
proving Proposition \ref{PropLpjbis}.
\end{proof}

\subsection{Proof of the Theorem \ref{HamTheorem}}

We may now proceed to the proof of Theorem \ref{HamTheorem}. We recall that the Hamiltonian is given by:

\bea\label{HamComputation}\text{Ham}^{(\boldsymbol{\alpha})}(\mathbf{q},\mathbf{p})&=&-\frac{\hbar}{2}\displaystyle{\sum_{\substack{(i,j)\in \llbracket 1,g\rrbracket^2 \\ i\neq j }}} \frac{(\mu^{(\boldsymbol{\alpha})}_i+\mu^{(\boldsymbol{\alpha})}_j)(p_i-p_j)}{q_i-q_j} -\hbar \sum_{j=1}^{g} (\nu^{(\boldsymbol{\alpha})}_{\infty,0} p_j+\nu^{(\boldsymbol{\alpha})}_{\infty,-1}q_jp_j) \cr
&&+\sum_{j=1}^{g}\mu^{(\boldsymbol{\alpha})}_j\left(p_j^2-\td{P}_1(q_j)p_j+\td{P}_2(q_j)\right)-\hbar \sum_{j=1}^g\sum_{k=1}^{r_\infty-1}c^{(\boldsymbol{\alpha})}_{\infty,k}q_j^{k}
\eea

 A straightforward computation from \eqref{DefHam} and from the fact that the $\left(\nu^{(\boldsymbol{\alpha})}_{\infty,k}\right)_{-1\leq k\leq r_\infty-3}$ and $\left(c^{(\boldsymbol{\alpha})}_{\infty,k}\right)_{1\leq k\leq r_\infty-1}$ are independent of $q_j$ provides
\bea -\frac{\partial \text{Ham}^{(\boldsymbol{\alpha})}(\mathbf{q},\mathbf{p})}{\partial q_j}&=&\hbar\displaystyle{\sum_{\substack{i\in \llbracket 1,g\rrbracket \\ i\neq j }}} \frac{(\mu^{(\boldsymbol{\alpha})}_i+\mu^{(\boldsymbol{\alpha})}_j)(p_i-p_j)}{(q_i-q_j)^2}+\frac{\hbar}{2}\displaystyle{\sum_{\substack{(r,s)\in \llbracket 1,g\rrbracket^2 \\ r\neq s }}} \frac{(\partial_{q_j}\mu^{(\boldsymbol{\alpha})}_r+\partial_{q_j}\mu^{(\boldsymbol{\alpha})}_s)(p_r-p_s)}{q_r-q_s}\cr
&&+\hbar \nu^{(\boldsymbol{\alpha})}_{\infty,-1}p_j -\sum_{i=1}^{g}\partial_{q_j}(\mu^{(\boldsymbol{\alpha})}_i)\left(p_i^2-\td{P}_1(q_i)p_i+\td{P}_2(q_i) \right)\cr
&&-\mu^{(\boldsymbol{\alpha})}_j\left(\td{P}_2'(q_j)-p_j \td{P}_1'(q_j) \right)+\hbar\sum_{k=1}^{r_\infty-1}kc^{(\boldsymbol{\alpha})}_{\infty,k}q_j^{k-1}\cr
&\overset{\text{Prop. \eqref{PropLpjbis}}}{=}&\mathcal{L}_{\boldsymbol{\alpha}}[p_j]
\eea
Similarly a direct computation using the fact that $\left(\nu^{(\boldsymbol{\alpha})}_{\infty,k}\right)_{-1\leq k\leq r_\infty-3}$ and $\left(c^{(\boldsymbol{\alpha})}_{\infty,k}\right)_{1\leq k\leq r_\infty-1}$ and $\left(\mu^{(\boldsymbol{\alpha})}_i\right)_{1\leq i\leq g}$ are independent of $p_j$ gives:
\beq \frac{\partial \text{Ham}^{(\boldsymbol{\alpha})}(\mathbf{q},\mathbf{p})}{\partial p_j}=-\hbar\displaystyle{\sum_{\substack{i\in \llbracket 1,g\rrbracket \\ i\neq j }}} \frac{\mu^{(\boldsymbol{\alpha})}_i+\mu^{(\boldsymbol{\alpha})}_j}{q_j-q_i} -\hbar \nu^{(\boldsymbol{\alpha})}_{\infty,0} -\hbar\nu^{(\boldsymbol{\alpha})}_{\infty,-1}q_j +\mu^{(\boldsymbol{\alpha})}_j\left(2p_j -\td{P}_1(q_j)\right)
\eeq
which is exactly $\mathcal{L}_{\boldsymbol{\alpha}}[q_j]$ given by \eqref{Lqj}.

\medskip

The last step is to verify that from Propositions \ref{PropA12Form} and \ref{PropDefCi2}:
\bea &&\sum_{j=1}^g \mu_j^{\boldsymbol{(\alpha)}}\left(p_j^2- \td{P}_1(q_j)p_j +\td{P}_2(q_j)+\hbar \underset{i\neq j}{\sum}\frac{p_i-p_j}{q_j-q_i}\right) \cr
&&=\begin{pmatrix}\mu_1^{\boldsymbol{(\alpha)}},\dots,\mu_{g}^{\boldsymbol{(\alpha)}}\end{pmatrix} \begin{pmatrix} p_1^2- \td{P}_1(q_1)p_1 +\td{P}_2(q_1)+\hbar \underset{i\neq 1}{\sum}\frac{p_i-p_1}{q_1-q_i}\\
\vdots\\
p_g^2- \td{P}_1(q_g)p_g +\td{P}_2(q_g)+\hbar \underset{i\neq g}{\sum}\frac{p_i-p_g}{q_g-q_i}\end{pmatrix}\cr
&&=\begin{pmatrix}\mu_1^{\boldsymbol{(\alpha)}},\dots,\mu_{g} ^{\boldsymbol{(\alpha)}}\end{pmatrix}(V_\infty)^t\begin{pmatrix} H_{\infty,0}\\ \vdots\\ H_{\infty,r_\infty-4}\end{pmatrix}=\begin{pmatrix} \nu^{\boldsymbol{(\alpha)}}_{\infty,1}&\dots& \nu^{\boldsymbol{(\alpha)}}_{\infty,r_\infty-3}\end{pmatrix} \begin{pmatrix} H_{\infty,0}\\ \vdots\\ H_{\infty,r_\infty-4}\end{pmatrix}\cr
&&=\sum_{k=0}^{r_\infty-4} \nu_{\infty,k+1}^{\boldsymbol{(\alpha)}}H_{\infty,k}
\eea
so that \eqref{HamComputation} becomes
\beq\text{Ham}^{(\boldsymbol{\alpha})}(\mathbf{q},\mathbf{p})=\sum_{k=0}^{r_\infty-4} \nu_{\infty,k+1}^{\boldsymbol{(\alpha)}}H_{\infty,k}-\hbar \sum_{j=1}^g\sum_{k=1}^{r_\infty-1}c^{(\boldsymbol{\alpha})}_{\infty,k}q_j^{k}-\hbar \nu^{(\boldsymbol{\alpha})}_{\infty,0}\sum_{j=1}^{g} p_j-\hbar \nu^{(\boldsymbol{\alpha})}_{\infty,-1}\sum_{j=1}^g q_jp_j
\eeq

\section{Proofs of the identities involving elementary symmetric polynomials}\label{AppendixH}
Let us prove Lemma \ref{LemmaESP1}. By definition we have
\beq \prod_{j=1}^g (\lambda-x_j)=\sum_{i=0}^g (-1)^i e_i(\{x_1,\dots,x_g\}) \lambda^{g-i}\eeq
Taking the derivative relatively to $x_m$ provides:
\beq\prod_{j\neq m}(\lambda-x_j)=\sum_{i=0}^{g} (-1)^{i-1}\frac{\partial e_i(\{x_1,\dots,x_g\})}{\partial x_m} \lambda^{g-i}\eeq
Thus, $(-1)^{i-1}\frac{\partial e_i(\{x_1,\dots,x_g\})}{\partial x_m}$ is the coefficient of order $\lambda^{g-i}$ of $\underset{j\neq m}{\prod}(\lambda-x_j)$. However, the last quantity is also given by
\bea \prod_{j\neq m}(\lambda-x_j)&=&\frac{\underset{j=1}{\overset{g}{\prod}}(\lambda-x_j)}{\lambda-x_m}=\left(\sum_{r=0}^g (-1)^r e_r(\{x_1,\dots,x_g\}) \lambda^{g-r}\right)\left(\sum_{s=0}^{\infty} x_m^{s}\lambda^{-s-1}\right) \cr
&=&\sum_{r=0}^g\sum_{s=0}^{\infty} (-1)^r e_r(\{x_1,\dots,x_g\})x_m^s \lambda^{g-r-s-1}
\eea
Identifying the coefficient of order $\lambda^{g-i}$ provides
\beq (-1)^{i-1}\frac{\partial e_i(\{x_1,\dots,x_g\})}{\partial x_m}=\sum_{r=0}^{i-1} (-1)^r e_r(\{x_1,\dots,x_g\})x_m^{i-1-r}\eeq
proving the lemma.

Let us now prove Proposition \ref{PropESP1}. From, Lemma \ref{LemmaESP1}, we get that both polynomials take the same value at $x_m$ (the value is $\frac{\partial e_i(\{x_1,\dots,x_g\})}{\partial x_m}$) for any $m\in \llbracket 1, g\rrbracket$. Since both sides are obviously polynomials of order at most $g-1$, we immediately get that they are equal.

Finally, let us prove Corollary \ref{QCorollary}. By definition
\bea Q(\lambda)&=&-\sum_{i=1}^g p_i \prod_{j\neq i}\frac{\lambda-q_j}{q_i-q_j}=-\sum_{i=1}^gP_i\left(\sum_{k=1}^g \frac{\partial e_i(\{q_1,\dots,q_g\})}{\partial q_k} \prod_{j\neq k} \frac{\lambda-q_j}{q_k-q_j}\right)\cr
&\overset{\text{Prop }\ref{PropESP1}}{=}&-\sum_{i=1}^gP_i\sum_{j=0}^{i-1}(-1)^jQ_{i-j-1}\lambda^j\cr
&=&\sum_{j=0}^{g-1} (-1)^{j-1}\left(\sum_{i=j+1}^g P_i Q_{i-j-1}\right)\lambda^j
\eea
ending the proof.
\medskip

To prove \eqref{IdentityPoly1}, let $X=\{x_1,\ldots,x_n\}$ and fix $j$. Polynomial division gives
\bea
\prod_{k\neq j}(\lambda-x_k)
&=&\sum_{i=1}^n(-1)^{n-i}e_{n-i}(X\setminus\{x_j\})\lambda^{i-1}\cr
&=&\sum_{i=1}^n\left(\sum_{m=i}^n
(-1)^{n-m}e_{n-m}(X)x_j^{m-i}\right)\lambda^{i-1}.
\eea
Comparing coefficients proves
\beq
e_{n-i}(X\setminus\{x_j\})
=\sum_{m=i}^n(-1)^{m-i}e_{n-m}(X)x_j^{m-i},
\qquad 1\leq i\leq n.
\eeq
The second identity follows by evaluating the characteristic polynomial at $\lambda=x_j$:
\beq
0=\prod_{k=1}^n(x_j-x_k)
=\sum_{m=0}^n(-1)^{n-m}e_{n-m}(X)x_j^m.
\eeq

To prove Lemma \ref{LemmaInversionVandermonde1}, write $X=\{x_1,\ldots,x_n\}$ and expand the finite numerator before expanding the denominator:
\bea
\frac1{\lambda-x_i}
&=&\frac{\prod_{k\neq i}(\lambda-x_k)}{\prod_{k=1}^n(\lambda-x_k)}\cr
&=&\left(\sum_{j=1}^n\sum_{m=j}^n
(-1)^{n-m}e_{n-m}(X)x_i^{j-1}\lambda^{m-j}\right)
\left(\sum_{r\geq0}h_r(X)\lambda^{-n-r}\right).
\eea
The numerator identity follows from \eqref{IdentityPoly1}, with its finite sums reordered. Comparing the coefficient of $\lambda^{-M-1}$ gives
\beq
x_i^M=\sum_{j=1}^n\sum_{m=\max(j,j+n-1-M)}^n
(-1)^{n-m}e_{n-m}(X)h_{M+m-j-n+1}(X)x_i^{j-1}.
\eeq
Here $m\geq j$ comes from the finite numerator, while $m\geq j+n-1-M$ comes from the nonnegative index of $h_r$. This proves the lemma.

Let us now prove Proposition \ref{PropInversionVandermonde2}. Let $M\geq 0$, the system
\beq \begin{pmatrix}1&x_1&\dots&x_1^{n-1}\\ 1&x_2&\dots&x_2^{n-1}\\\vdots &\vdots& &\vdots\\ 1&x_n&\dots&x_n^{n-1}\end{pmatrix}\begin{pmatrix}C_1\\ C_2\\ \vdots \\ C_n\end{pmatrix}=\begin{pmatrix} x_1^M\\ x_2^M\\ \vdots\\ x_n^M\end{pmatrix}\eeq
is equivalent to say, by inverting the Vandermonde matrix, that for all $i\in \llbracket 1,n\rrbracket$:
\beq \label{ProofPropVandermonde2}C_i=\sum_{j=1}^n \frac{(-1)^{n-i}e_{n-i}(\{x_1,\dots,x_n\}\setminus\{x_j\})}{\underset{m\neq j}{\prod}(x_j-x_m)}x_j^M\eeq
The system is also equivalent to, for all $i\in \llbracket 1,n\rrbracket$:
\beq \label{ProofPropVandermonde22} \sum_{j=1}^n x_i^{j-1}C_j= x_i^M\eeq
Lemma \ref{LemmaInversionVandermonde1} implies that the last identity is equivalent to, for all $j\in \llbracket 1,n\rrbracket$
\beq \label{muSymmetric}C_j=\sum_{m=\text{Max}(j,j+n-1-M)}^n (-1)^{n-m}e_{n-m}(\{x_1,\dots,x_n\})h_{M+m-j-n+1}(\{x_1,\dots,x_n\})\eeq
Identifying \eqref{ProofPropVandermonde2} and \eqref{ProofPropVandermonde22} finally proves Proposition \ref{PropInversionVandermonde2}.

\section{Proof of Theorem \ref{HamiltonianSymmetricCoordinates}}\label{ProofTermsHamiltonians}
The proof is based on the computation of each term appearing in the Hamiltonian. Let us compute for $i\in \llbracket 1,g\rrbracket$:

\bea A_i&:=&\hbar \sum_{j=1}^g\sum_{r\neq j} \frac{(-1)^{g-i}e_{g-i}(\{q_1,\dots,q_g\}\setminus\{q_j\})}{\underset{a\neq j}{\prod}(q_j-q_a)}\frac{p_r-p_j}{q_j-q_r}\cr
&=&-\hbar \sum_{j=1}^g\sum_{r\neq j} \frac{(-1)^{g-i}e_{g-i}(\{q_1,\dots,q_g\}\setminus\{q_j\})}{\underset{a\neq j}{\prod}(q_j-q_a)}\frac{p_r-p_j}{q_r-q_j}\cr
&\overset{\eqref{LemmaESP1}}{=}&-\hbar \sum_{j=1}^g\sum_{r\neq j} \frac{(-1)^{g-i}e_{g-i}(\{q_1,\dots,q_g\}\setminus\{q_j\})}{\underset{a\neq j}{\prod}(q_j-q_a)}\sum_{k=1}^gP_k\sum_{m=0}^{k-1}(-1)^{m} e_{k-1-m}(\{q_1,\dots,q_g\})\frac{q_r^m-q_j^m}{q_r-q_j}\cr
&=&-\hbar \sum_{j=1}^g\sum_{r\neq j} \frac{(-1)^{g-i}e_{g-i}(\{q_1,\dots,q_g\}\setminus\{q_j\})}{\underset{a\neq j}{\prod}(q_j-q_a)}\sum_{k=1}^gP_k\sum_{m=0}^{k-1}(-1)^{m} e_{k-1-m}(\{q_1,\dots,q_g\})\sum_{s=0}^{m-1} q_r^s q_j^{m-1-s}\cr
&=&-\hbar \sum_{j=1}^g\sum_{r=1}^g \frac{(-1)^{g-i}e_{g-i}(\{q_1,\dots,q_g\}\setminus\{q_j\})}{\underset{a\neq j}{\prod}(q_j-q_a)}\sum_{k=1}^gP_k\sum_{m=0}^{k-1}(-1)^{m} e_{k-1-m}(\{q_1,\dots,q_g\})\sum_{s=0}^{m-1} q_r^s q_j^{m-1-s}\cr
&&+\hbar \sum_{j=1}^g\frac{(-1)^{g-i}e_{g-i}(\{q_1,\dots,q_g\}\setminus\{q_j\})}{\underset{a\neq j}{\prod}(q_j-q_a)}\sum_{k=1}^gP_k\sum_{m=0}^{k-1}(-1)^{m} e_{k-1-m}(\{q_1,\dots,q_g\})\sum_{s=0}^{m-1} q_j^{m-1}\cr 
&=&-\hbar  \sum_{k=1}^g\sum_{m=0}^{k-1}\sum_{s=0}^{m-1}(-1)^{m}P_k e_{k-1-m}(\{q_1,\dots,q_g\})\sum_{r=1}^gq_r^s\sum_{j=1}^g  \frac{(-1)^{g-i}e_{g-i}(\{q_1,\dots,q_g\}\setminus\{q_j\})}{\underset{a\neq j}{\prod}(q_j-q_a)}q_j^{m-1-s}\cr
&&+\hbar \sum_{k=1}^g\sum_{m=0}^{k-1}\sum_{s=0}^{m-1} (-1)^{m}P_k e_{k-1-m}(\{q_1,\dots,q_g\}) \sum_{j=1}^g\frac{(-1)^{g-i}e_{g-i}(\{q_1,\dots,q_g\}\setminus\{q_j\})}{\underset{a\neq j}{\prod}(q_j-q_a)}q_j^{m-1}\cr 
&\overset{\eqref{PropInversionVandermonde2}}{=}&-\hbar  \sum_{k=1}^g\sum_{m=0}^{k-1}\sum_{s=0}^{m-1}(-1)^{m}P_k e_{k-1-m}(\{q_1,\dots,q_g\})\sum_{r=1}^gq_r^s \delta_{i,m-s}\cr
&&+\hbar \sum_{k=1}^g\sum_{m=0}^{k-1}\sum_{s=0}^{m-1} (-1)^{m}P_k e_{k-1-m}(\{q_1,\dots,q_g\}) \delta_{i,m}\cr 
&=&-\hbar  \sum_{k=1}^g\sum_{m=0}^{k-1}\sum_{s=0}^{m-1}(-1)^{m}P_k e_{k-1-m}(\{q_1,\dots,q_g\})S_s(\{q_1,\dots,q_g\}) \delta_{s,m-i}\cr
&&+\hbar \sum_{k=1}^g\sum_{m=0}^{k-1} (-1)^{m}m P_k e_{k-1-m}(\{q_1,\dots,q_g\}) \delta_{i,m}\cr 
&=&-\hbar  \sum_{k=i+1}^g\sum_{m=i}^{k-1}(-1)^{m}P_k e_{k-1-m}(\{q_1,\dots,q_g\})S_{m-i}(\{q_1,\dots,q_g\}) \cr
&&+\hbar \sum_{k=1}^g\delta_{k-1\geq i} (-1)^{i}i P_k e_{k-1-i}(\{q_1,\dots,q_g\}) \cr 
&=&-\hbar  \sum_{k=i+1}^g\sum_{m=i}^{k-1}(-1)^{m}P_k e_{k-1-m}(\{q_1,\dots,q_g\})S_{m-i}(\{q_1,\dots,q_g\}) \cr
&&+\hbar \sum_{k=i+1}^g(-1)^{i}i P_k e_{k-1-i}(\{q_1,\dots,q_g\}) \cr 
&=&-\hbar  \sum_{k=i+1}^g\sum_{m=i+1}^{k-1}(-1)^{m}P_k e_{k-1-m}(\{q_1,\dots,q_g\})S_{m-i}(\{q_1,\dots,q_g\})\cr
&&-\hbar  \sum_{k=i+1}^g(-1)^{i}gP_k e_{k-1-i}(\{q_1,\dots,q_g\})\cr
&&+\hbar \sum_{k=i+1}^g(-1)^{i}i P_k e_{k-1-i}(\{q_1,\dots,q_g\}) \cr
&=&-\hbar  \sum_{k=i+1}^g\sum_{m=i+1}^{k-1}(-1)^{m}P_k Q_{k-1-m}S_{m-i}(\{q_1,\dots,q_g\})-\hbar  \sum_{k=i+1}^g(-1)^{i}(g-i)P_k Q_{k-1-i} 
\eea

\bea B_i&:=&-\sum_{j=1}^g \frac{(-1)^{g-i}e_{g-i}(\{q_1,\dots,q_g\}\setminus\{q_j\})}{\underset{a\neq j}{\prod}(q_j-q_a)}\td{P}_1(q_j)p_j \cr
&=&\sum_{j=1}^g\sum_{s=0}^{g+1}\sum_{k=1}^g P_k\sum_{r=0}^{k-1}(-1)^re_{k-1-r}(\{q_1,\dots,q_g\}) \frac{(-1)^{g-i}e_{g-i}(\{q_1,\dots,q_g\}\setminus\{q_j\})}{\underset{a\neq j}{\prod}(q_j-q_a)}t_{\infty,2s+2} q_j^{s+r} \cr
&=&\sum_{k=1}^g \sum_{r=0}^{k-1}\sum_{s=0}^{g+1} (-1)^r P_k e_{k-1-r}(\{q_1,\dots,q_g\}) t_{\infty,2s+2} \sum_{j=1}^g\frac{(-1)^{g-i}e_{g-i}(\{q_1,\dots,q_g\}\setminus\{q_j\})}{\underset{a\neq j}{\prod}(q_j-q_a)} q_j^{s+r} \cr
&=&\sum_{k=1}^g \sum_{r=0}^{k-1}\sum_{s=0}^{g-1-r} (-1)^r P_k e_{k-1-r}(\{q_1,\dots,q_g\}) t_{\infty,2s+2} \sum_{j=1}^g\frac{(-1)^{g-i}e_{g-i}(\{q_1,\dots,q_g\}\setminus\{q_j\})}{\underset{a\neq j}{\prod}(q_j-q_a)} q_j^{s+r} \cr
&&+\sum_{k=1}^g \sum_{r=0}^{k-1}\sum_{s=g-r}^{g+1} (-1)^r P_k e_{k-1-r}(\{q_1,\dots,q_g\}) t_{\infty,2s+2} \sum_{j=1}^g\frac{(-1)^{g-i}e_{g-i}(\{q_1,\dots,q_g\}\setminus\{q_j\})}{\underset{a\neq j}{\prod}(q_j-q_a)} q_j^{s+r} \cr
&\overset{\eqref{PropInversionVandermonde2}}{=}&\sum_{k=1}^g \sum_{r=0}^{k-1}\sum_{s=0}^{g-1-r} (-1)^r P_k Q_{k-1-r} t_{\infty,2s+2} \delta_{s+r+1,i} \cr
&&+\sum_{k=1}^g \sum_{r=0}^{k-1}\sum_{s=g-r}^{g+1} (-1)^r P_k Q_{k-1-r} t_{\infty,2s+2} \sum_{m=i}^{g} (-1)^{g-m}Q_{g-m}h_{r+s+m-i-g+1}(\{q_1,\dots,q_g\})\cr
&=&\sum_{k=1}^g \sum_{r=0}^{\text{Min}(k-1,i-1)} (-1)^r t_{\infty,2i-2r} P_k Q_{k-1-r}   \cr
&&+\sum_{k=1}^g \sum_{r=0}^{k-1}\sum_{s=g-r}^{g+1}\sum_{m=i}^{g} (-1)^{g+r-m}t_{\infty,2s+2}  P_k Q_{k-1-r}  Q_{g-m}h_{r+s+m-i-g+1}(\{q_1,\dots,q_g\})\cr
&&
\eea

\bea
C_i&:=&\sum_{j=1}^g \frac{(-1)^{g-i}e_{g-i}(\{q_1,\dots,q_g\}\setminus\{q_j\})}{\underset{a\neq j}{\prod}(q_j-q_a)}\td{P}_2(q_j) \cr
&=&\sum_{j=1}^g\sum_{r=g}^{2r_\infty-4}\td{P}_{\infty,r}^{(2)} \frac{(-1)^{g-i}e_{g-i}(\{q_1,\dots,q_g\}\setminus\{q_j\})}{\underset{a\neq j}{\prod}(q_j-q_a)}q_j^{r} \cr
&\overset{\eqref{PropInversionVandermonde2}}{=}&\sum_{r=g}^{2r_\infty-4}\sum_{m=i}^{g} (-1)^{g-m} \td{P}_{\infty,r}^{(2)} e_{g-m}(\{q_1,\dots,q_g\})h_{r+m-i-g+1}(\{q_1,\dots,q_g\})\cr
&=&\sum_{r=g}^{2r_\infty-4}\sum_{m=i}^{g} (-1)^{g-m} \td{P}_{\infty,r}^{(2)} Q_{g-m}h_{r+m-i-g+1}(\{q_1,\dots,q_g\})
\eea

\footnotesize{
\bea
E_i&:=&\sum_{j=1}^g \frac{(-1)^{g-i}e_{g-i}(\{q_1,\dots,q_g\}\setminus\{q_j\})}{\underset{a\neq j}{\prod}(q_j-q_a)}p_j^2 \cr
&=&\sum_{j=1}^g \sum_{k_1=1}^g\sum_{k_2=1}^g P_{k_1}P_{k_2}\sum_{r_1=0}^{k_1-1}\sum_{r_2=0}^{k_2-1}(-1)^{r_1+r_2}Q_{k_1-1-r_1}Q_{k_2-1-r_2}\frac{(-1)^{g-i}e_{g-i}(\{q_1,\dots,q_g\}\setminus\{q_j\})}{\underset{a\neq j}{\prod}(q_j-q_a)}q_j^{r_1+r_2} \cr
&=& \sum_{k_1=1}^g\sum_{k_2=1}^g P_{k_1}P_{k_2}\sum_{r_1=0}^{k_1-1}\sum_{r_2=0}^{k_2-1}(-1)^{r_1+r_2}Q_{k_1-1-r_1}Q_{k_2-1-r_2}\sum_{j=1}^g\frac{(-1)^{g-i}e_{g-i}(\{q_1,\dots,q_g\}\setminus\{q_j\})}{\underset{a\neq j}{\prod}(q_j-q_a)}q_j^{r_1+r_2}\delta_{r_1+r_2\leq g-1} \cr
&&+ \sum_{k_1=1}^g\sum_{k_2=1}^g P_{k_1}P_{k_2}\sum_{r_1=0}^{k_1-1}\sum_{r_2=0}^{k_2-1}(-1)^{r_1+r_2}Q_{k_1-1-r_1}Q_{k_2-1-r_2}\sum_{j=1}^g\frac{(-1)^{g-i}e_{g-i}(\{q_1,\dots,q_g\}\setminus\{q_j\})}{\underset{a\neq j}{\prod}(q_j-q_a)}q_j^{r_1+r_2}\delta_{r_1+r_2\geq  g} \cr
&\overset{\eqref{PropInversionVandermonde2}}{=}&\sum_{k_1=1}^g\sum_{k_2=1}^g P_{k_1}P_{k_2}\sum_{r_1=0}^{k_1-1}\sum_{r_2=0}^{k_2-1}(-1)^{r_1+r_2}Q_{k_1-1-r_1}Q_{k_2-1-r_2}\delta_{i,r_1+r_2+1}\delta_{r_1+r_2\leq g-1} \cr
&&+ \sum_{k_1=1}^g\sum_{k_2=1}^g P_{k_1}P_{k_2}\sum_{r_1=0}^{k_1-1}\sum_{r_2=0}^{k_2-1}(-1)^{r_1+r_2}Q_{k_1-1-r_1}Q_{k_2-1-r_2}\delta_{r_1+r_2\geq  g}\sum_{m=i}^g (-1)^{g-m}Q_{g-m}h_{r_1+r_2+m-i-g+1}(\{q_1,\dots,q_g\})  \cr
&=&(-1)^{i-1}\sum_{k_1=1}^g\sum_{k_2=1}^g P_{k_1}P_{k_2}\sum_{r_1=\text{Max}(0,i-k_2)}^{\text{Min}(k_1-1,i-1)}Q_{k_1-1-r_1}Q_{k_2-i+r_1} \cr
&&+ \sum_{k_1=1}^g\sum_{k_2=1}^g P_{k_1}P_{k_2}\displaystyle \sum_{\substack{0\leq r_1\leq k_1-1 \\ 0\leq r_2\leq k_2-1 \\ r_1+r_2\geq g}}(-1)^{r_1+r_2}Q_{k_1-1-r_1}Q_{k_2-1-r_2}\sum_{m=i}^g (-1)^{g-m}Q_{g-m}h_{r_1+r_2+m-i-g+1}(\{q_1,\dots,q_g\})  \cr
&&
\eea}
\normalsize{}
We also have:
\bea -\hbar\sum_{i=1}^g p_i&=& -\hbar\sum_{i=1}^g\sum_{m=1}^g P_m\frac{\partial e_m(\{q_1,\dots,q_g\})}{\partial q_i}\cr
&\overset{\text{Lemma }\eqref{LemmaESP1}}{=}&-\hbar \sum_{m=1}^g\sum_{i=1}^g\sum_{j=0}^{m-1}(-1)^j P_m e_{m-1-j}(\{q_1,\dots,q_g\}) q_i^j\cr
&=&-\hbar\sum_{m=1}^g\sum_{j=0}^{m-1}(-1)^j P_m e_{m-1-j}(\{q_1,\dots,q_g\}) \sum_{i=1}^g q_i^j\cr
&=&-\hbar\sum_{m=1}^gP_m \sum_{j=0}^{m-1}(-1)^j  e_{m-1-j}(\{q_1,\dots,q_g\}) S_j(\{q_1,\dots,q_g\})\cr
&\overset{k=m-1}{=}&-\hbar\sum_{k=0}^{g-1}P_{k+1} \sum_{j=0}^{k}(-1)^j  e_{k-j}(\{q_1,\dots,q_g\}) S_j(\{q_1,\dots,q_g\})\cr
&\overset{\eqref{Identitites}}{=}&-\hbar\sum_{k=0}^{g-1}(g-k)Q_kP_{k+1}
\eea
and
\bea -\hbar\sum_{i=1}^g q_ip_i&=& -\hbar\sum_{i=1}^g\sum_{m=1}^g q_i P_m\frac{\partial e_m(\{q_1,\dots,q_g\})}{\partial q_i}\cr
&\overset{\text{Lemma }\eqref{LemmaESP1}}{=}&-\hbar \sum_{m=1}^g\sum_{i=1}^g\sum_{j=0}^{m-1}(-1)^j P_m e_{m-1-j}(\{q_1,\dots,q_g\}) q_i^{j+1}\cr
&=&-\hbar\sum_{m=1}^g\sum_{j=0}^{m-1}(-1)^j P_m e_{m-1-j}(\{q_1,\dots,q_g\}) \sum_{i=1}^g q_i^{j+1}\cr
&=&-\hbar\sum_{m=1}^gP_m \sum_{j=0}^{m-1}(-1)^j  e_{m-1-j}(\{q_1,\dots,q_g\}) S_{j+1}(\{q_1,\dots,q_g\})\cr
&\overset{k=m-1}{=}&-\hbar\sum_{k=0}^{g-1}P_{k+1} \sum_{j=0}^{k}(-1)^j  e_{k-j}(\{q_1,\dots,q_g\}) S_{j+1}(\{q_1,\dots,q_g\})\cr
&\overset{i=j+1}{=}&-\hbar\sum_{k=0}^{g-1}P_{k+1} \sum_{i=1}^{k+1}(-1)^{i-1}  e_{k+1-i}(\{q_1,\dots,q_g\}) S_{i}(\{q_1,\dots,q_g\})\cr
&\overset{\eqref{Identititesmod}}{=}&-\hbar\sum_{k=0}^{g-1}(k+1)P_{k+1}e_{k+1}(\{q_1,\dots,q_g\})\cr
&=&-\hbar\sum_{k=0}^{g-1}(k+1)P_{k+1}Q_{k+1}\cr
&=&-\hbar\sum_{k=1}^{g}kP_{k}Q_{k}
\eea
where we have used a modified version of \eqref{Identitites}. Indeed, \eqref{Identitites} implies that for any $k\in \llbracket 0,g-1\rrbracket$:
\beq\sum_{i=0}^{k+1}(-1)^ie_{k+1-i}(\{q_1,\dots,q_g\})S_i(\{q_1,\dots,q_g\})=(g-k-1)e_{k+1}(\{q_1,\dots,q_g\})\eeq
Isolating the term $i=0$ and reminding that $S_0(\{q_1,\dots,q_g\})=g$ gives:
\bea -(g-k-1)e_{k+1}(\{q_1,\dots,q_g\})&=&\sum_{i=0}^{k+1}(-1)^{i-1}e_{k+1-i}(\{q_1,\dots,q_g\})S_i(\{q_1,\dots,q_g\})\cr
&=&\sum_{i=1}^{k+1}(-1)^{i-1}e_{k+1-i}(\{q_1,\dots,q_g\})S_i(\{q_1,\dots,q_g\})-ge_{k+1}(\{q_1,\dots,q_g\})\cr
&&
\eea
i.e.
\beq \label{Identititesmod}\sum_{i=1}^{k+1}(-1)^{i-1}e_{k+1-i}(\{q_1,\dots,q_g\})S_i(\{q_1,\dots,q_g\})=(k+1)e_{k+1}(\{q_1,\dots,q_g\})\eeq

\section{Proof of Proposition \ref{PropTdA}}\label{ProofAtildeSymmetric}
Use the notation of Proposition \ref{PropTdA}, and let $D=\mathcal L_{\boldsymbol\alpha}$, which already contains $\hbar$. The gauge from the companion system to the normalized polynomial system is
\beq
S=G_1J=\begin{pmatrix}1&0\\K/R&1/R\end{pmatrix},
\qquad S^{-1}=\begin{pmatrix}1&0\\-K&R\end{pmatrix}.
\eeq
Direct multiplication in $\td A=SAS^{-1}+(DS)S^{-1}$ gives
\bea
\td A_{1,1}&=&A_{1,1}-KA_{1,2},\cr
\td A_{1,2}&=&RA_{1,2},\cr
\td A_{2,1}&=&\frac{A_{2,1}+(A_{1,1}-A_{2,2})K-A_{1,2}K^2+DK}{R},\cr
\td A_{2,2}&=&A_{2,2}+KA_{1,2}-\frac{DR}{R}.
\eea
In particular the exact trace relation is
\beq
\Tr\td A=\Tr A-D\log R.
\eeq
The traces are not generally equal as rational functions. Since $R$ is monic, $(DR)/R=O(\lambda^{-1})$ and their polynomial parts agree.

Put $v=A_{1,2}$. The residues of $vQ$ are $-\mu_jp_j$, so
\beq
A_{1,1}=c+vQ-[vQ]_+.
\eeq
This proves the first two entries of Proposition \ref{PropTdA}. Also
\beq
A_{2,2}=A_{1,1}+\hbar v'+vL_{2,2},
\qquad L_{2,2}=\td P_1+\hbar R'/R.
\eeq
Taking polynomial parts and using $R'/R=g/\lambda+Q_1/\lambda^2+O(\lambda^{-3})$ yields
\beq
\Tr\td A=2c+[v\td P_1]_++\hbar(g+1)\nu_{-1}.
\eeq
Thus, if a zero constant term in this trace is desired, the scalar coefficient is
\beq
c_{\infty,0}^{(\boldsymbol\alpha)}
=-\frac{\hbar}{2}(g+1)\nu_{-1}
+\frac12\sum_{j=0}^{r_\infty-2}t_{\infty,2j+2}\nu_j.
\eeq

For the lower-left entry, use
$A_{2,1}=\hbar A_{1,1}'+vL_{2,1}$ and
$L_{2,1}=-\td P_2+E-\hbar\sum_jp_j/(\lambda-q_j)$.
The strictly proper part of $A_{1,1}'/R$ contributes no polynomial part. Nor does $v\sum_jp_j/((\lambda-q_j)R)$: for $g\geq1$ it is strictly proper, and for $g=0$ the sum is empty. Hence
\bea
\td A_{2,1}
&=&\left[\frac{v(-\td P_2+E-\td P_1K-K^2)+\hbar c'}{R}\right]_+\cr
&&-\hbar\left[\left(v'+v\frac{R'}R\right)\frac KR\right]_+
+\left[\frac{DK}{R}\right]_+.
\eea
The required two remaining polynomial parts are
\bea
\left[\left(v'+v\frac{R'}R\right)\frac KR\right]_+
&=&(g+1)\nu_{-1}b+\frac a2(g\nu_0+\nu_{-1}Q_1),\cr
\left[\frac{DK}{R}\right]_+
&=&\frac\hbar2\bigl(\alpha_a(\lambda+Q_1)+\alpha_d\bigr).
\eea
For the second equality, $DQ/R$ is strictly proper and
$[b\,DR/R]_+=-a\,DQ_1/2$, which cancels the $a\,DQ_1/2$ term in $Dg_0$. For $g=0$ use $Q=Q_1=0$ and $R=1$. These identities prove the claimed lower-left entry and complete the proof.

\section{Proof of Proposition \ref{PropReduction}}\label{AppendixB} Let $k\in \llbracket 1 ,r_\infty-1\rrbracket$ and consider $\mathbf{w}_k$. It is obvious from \eqref{RelationNuAlphaInfty} (whose r.h.s. only implies odd entries of $\boldsymbol{\alpha}$) that $\nu_j^{(\mathbf{w}_k)}=0$ for all $j\in \llbracket -1, r_\infty-3\rrbracket$. Consequently, \eqref{RelationNuMuMatrixForm} implies that $\mu_j^{(\mathbf{w}_k)}=0$ for all $j\in \llbracket 1, g\rrbracket$. For $j\in \llbracket 1, r_\infty-1\rrbracket$, the $j^{\text{th}}$ line of the r.h.s. of \eqref{calphaexpr} is given by $\underset{m=r_\infty-j}{\overset{r_\infty-1}{\sum}}\left(\frac{\alpha_{\infty,4r_\infty-2j-2m-3}}{4r_\infty-2j-2m-3}t_{\infty,2m}-\frac{\alpha_{\infty,4r_\infty-2j-2m-2}}{4r_\infty-2j-2m-2}t_{\infty,2m-1}\right)$. For $\boldsymbol{\alpha}=\mathbf{w}_k$ it reduces only to $-\frac{1}{2k}t_{\infty,4r_\infty-3-2j-2k}\delta_{j\geq r_\infty- k}$
Thus we get:
\footnotesize{\beq \begin{pmatrix}t_{\infty,2r_\infty-3}&0&\dots& &0\\
t_{\infty,2r_\infty-5}& t_{\infty,2r_\infty-3}&0& \ddots& \vdots\\
\vdots& \ddots& \ddots&&\vdots \\
t_{\infty,3}& & \ddots &\ddots&0 \\
t_{\infty,1} &t_{\infty,3}& \dots & &t_{\infty,2r_\infty-3}\end{pmatrix} \begin{pmatrix}c^{(\mathbf{w}_k)}_{\infty,r_\infty-1}\\ c^{(\mathbf{w}_k)}_{\infty,r_\infty-2}\\ \vdots \\ c^{(\mathbf{w}_k)}_{\infty,1} \end{pmatrix}=-\frac{1}{2k}
\begin{pmatrix} 0\\ \vdots\\ 0\\ t_{\infty,2r_\infty-3}\\\vdots\\ t_{\infty,2r_\infty-2k-1}
\end{pmatrix}
\eeq} 
\normalsize{We} recognize that the r.h.s. is $-\frac{1}{2k}$ times the $(r_\infty-k)^{\text{th}}$ column of the matrix on the l.h.s. so that we immediately get
$c_{\infty,j}^{(\mathbf{w}_k)}=-\frac{1}{2k}\delta_{j,k}$ for all $j\in \llbracket 1, r_\infty-1\rrbracket$.

\medskip

Let us now take $k\in \llbracket -1,r_\infty-3\rrbracket$ and consider $\mathbf{u}_k$. The $(k+2)^{\text{th}}$ column of $M_\infty$ is given by $\underset{i=k+2}{\overset{r_\infty-1}{\sum}}t_{\infty,2r_\infty-2i+2k+1} \mathbf{e}_i$. Similarly, the r.h.s. of \eqref{RelationNuAlphaInfty} only implies odd indexes of $\boldsymbol{\mathbf{u}_k}$ and it is given by 
\bea \text{RHS}&=&\underset{i=1}{\overset{r_\infty-1}{\sum}}\frac{2 \alpha^{(\mathbf{u}_k)}_{\infty,2r_\infty-1-2i}}{2r_\infty-1-2i}\mathbf{e}_i\overset{m=r_\infty-i}{=}\underset{m=1}{\overset{r_\infty-1}{\sum}}\frac{2 \alpha^{(\mathbf{u}_k)}_{\infty,2m-1}}{2m-1}\mathbf{e}_{r_\infty-m}\cr
&=&\underset{m=1}{\overset{r_\infty-k-2}{\sum}} t_{\infty,2m+1+2k}\mathbf{e}_{r_\infty-m}\overset{i=r_\infty-m}{=}\underset{i=k+2}{\overset{r_\infty-1}{\sum}}t_{\infty,2r_\infty-2i+1+2k}  \mathbf{e}_{i}
\eea
We recognize the $(k+2)^{\text{th}}$ column of $M_\infty$. Hence, we conclude that $(\nu^{(\mathbf{u}_k)}_{-1},\dots,\nu^{(\mathbf{u}_k)}_{r_\infty-3})^t=\mathbf{e}_{k+2}$, i.e. $\nu^{(\mathbf{u}_k)}_{j}=\delta_{j,k}$ for all $j \in \llbracket -1,r_\infty-3\rrbracket$. 

Let us now take $i\in \llbracket 1, r_\infty-1\rrbracket$ and compute the $i^{\text{th}}$ line of the r.h.s. of \eqref{calphaexpr}. It is given by:
\small{\bea \text{RHS}_i&=&\underset{m=r_\infty-i}{\overset{r_\infty-1}{\sum}}\frac{\alpha^{(\mathbf{u}_k)}_{\infty,4r_\infty-2i-2m-3}}{4r_\infty-2i-2m-3}t_{\infty,2m}-\underset{m=r_\infty-i}{\overset{r_\infty-1}{\sum}}\frac{\alpha^{(\mathbf{u}_k)}_{\infty,4r_\infty-2i-2m-2}}{4r_\infty-2i-2m-2}t_{\infty,2m-1}\cr
&\overset{s=2r_\infty-i-m-1}{=}&\underset{s=r_\infty-i}{\overset{r_\infty-1}{\sum}}\frac{\alpha^{(\mathbf{u}_k)}_{\infty,2s-1}}{2s-1}t_{\infty,4r_\infty-2i-2s-2}-\underset{s=r_\infty-i}{\overset{r_\infty-1}{\sum}}\frac{\alpha^{(\mathbf{u}_k)}_{\infty,2s}}{2s}t_{\infty,4r_\infty-2i-2s-3}\cr
&=&\frac{1}{2}\underset{s=r_\infty-i}{\overset{r_\infty-k-2}{\sum}}t_{\infty,2s+2k+1}t_{\infty,4r_\infty-2i-2s-2}-\frac{1}{2}\underset{s=r_\infty-i}{\overset{r_\infty-k-2}{\sum}}t_{\infty,2s+2k+2}t_{\infty,4r_\infty-2i-2s-3}\cr
&\overset{m=2r_\infty-2-i-s-k}{=}&\frac{1}{2}\underset{m=r_\infty-i}{\overset{r_\infty-k-2}{\sum}}t_{\infty,4r_\infty-2i-2m-3}t_{\infty,2m+2k+2} -\frac{1}{2}\underset{s=r_\infty-i}{\overset{r_\infty-k-2}{\sum}}t_{\infty,2s+2k+2}t_{\infty,4r_\infty-2i-2s-3}\cr
&=&0
\eea}

\normalsize{Thus}, the r.h.s. of \eqref{calphaexpr} is null for $\boldsymbol{\alpha}=\mathbf{u}_k$ so that since $M_\infty$ is invertible, we get $c^{(\mathbf{u}_k)}_{\infty,j}=0$ for all j$\in \llbracket 1, r_\infty-1\rrbracket$.

\section{Proof of Proposition \ref{InverseCoordinates}}\label{AppendixF}
The even times and the two highest odd times are inverted directly by Definition \ref{Times}. Work on its fixed local branch and put
\beq
c_k=\frac{2r_\infty-2k-5}{2},\qquad
y_{-1}=1,\qquad
y_k=\frac12t_{\infty,2r_\infty-5-2k}\,T_2^{-c_k},
\qquad 0\leq k\leq g.
\eeq
Then $y_0=c_0T_1$. Index a vector $Y=(y_{-1},y_0,\ldots,y_g)^t$ by $-1,0,\ldots,g$, and define the strictly lower triangular matrix $N$ by
\beq
(NY)_{-1}=0,\qquad (NY)_k=c_k y_{k-1}\quad(0\leq k\leq g).
\eeq
All matrix exponentials below are finite sums, since $N$ is nilpotent. Expanding $e^{-T_1N}Y$ gives, for $k\geq1$,
\bea
(e^{-T_1N}Y)_k
&=&\sum_{i=0}^{k-1}\frac{(-T_1)^i}{i!}
\left(\prod_{m=k-i+1}^kc_m\right)y_{k-i}\cr
&&+\frac{(-1)^kT_1^{k+1}}{(k+1)(k-1)!}
\prod_{m=0}^kc_m.
\eea
In the last term we combined the contributions involving $y_0=c_0T_1$ and $y_{-1}=1$. Substituting the definitions of $c_m$ and $y_m$ shows that this is exactly $\tau_k$ in Definition \ref{Times}. The components with indices $-1$ and $0$ are $1$ and $0$. Therefore
\beq
e^{-T_1N}Y=(1,0,\tau_1,\ldots,\tau_g)^t,
\qquad
Y=e^{T_1N}(1,0,\tau_1,\ldots,\tau_g)^t.
\eeq
Expanding the second equality gives
\beq
y_j=\sum_{p=1}^j\frac{T_1^{j-p}}{(j-p)!}
\left(\prod_{m=p+1}^jc_m\right)\tau_p
+\frac{T_1^{j+1}}{(j+1)!}\prod_{m=0}^jc_m,
\qquad 1\leq j\leq g.
\eeq
Consequently,
\bea
t_{\infty,2r_\infty-5-2j}
&=&2T_2^{(2r_\infty-5-2j)/2}\Bigg[
\sum_{p=1}^j
\frac{T_1^{j-p}\tau_p}{2^{j-p}(j-p)!}
\prod_{m=p+1}^j(2r_\infty-2m-5)\cr
&&+\frac{T_1^{j+1}}{2^{j+1}(j+1)!}
\prod_{m=0}^j(2r_\infty-2m-5)\Bigg].
\eea
Taking $j=r_\infty-k-2$ proves Proposition \ref{InverseCoordinates}. This finite matrix identity establishes the inverse for arbitrary $r_\infty$, without the intermediate binomial rearrangements.

\section{Proof of Proposition \ref{PropDarbouxCoordinates}}\label{AppendixD}
Let $j\in \llbracket 1,g\rrbracket$. Let $k\in \llbracket 1,r_\infty-1\rrbracket$ and consider $\mathcal{L}_{\mathbf{w}_k}$. Since $\mathcal{L}_{\mathbf{w}_k}[T_1]=\mathcal{L}_{\mathbf{w}_k}[T_2]=\mathcal{L}_{\mathbf{w}_k}[q_j]=0$ from Proposition \ref{PropTrivialTimes}, we immediately get $\mathcal{L}_{\mathbf{w}_k}[\check{q}_j]=0$. Moreover from Proposition \ref{ClassicalSpectralCurve}, we observe that:
\beq \mathcal{L}_{\mathbf{w}_k}[\td{P}_1(\lambda)]=-\sum_{i=0}^{r_\infty-2}\mathcal{L}_{\mathbf{w}_k}[t_{\infty,2i+2}]\lambda^{i}=-\hbar\lambda^{k-1}\eeq
so that Theorem \ref{TheoSplitTangentSpace} provides for all $j\in \llbracket 1,g\rrbracket$:
\beq \mathcal{L}_{\mathbf{w}_k}[\check{p}_j]=T_2^{-1}\left(\mathcal{L}_{\mathbf{w}_k}[p_j]-\frac{1}{2}\mathcal{L}_{\mathbf{w}_k}[\td{P}_1](q_j) \right)
=T_2^{-1}\left(-\frac{\hbar}{2} q_j^{k-1}+\frac{1}{2}\hbar q_j^{k-1}\right)=0
\eeq

\medskip

Let us now consider $\mathcal{L}_{\mathbf{u}_{-1}}$. From Proposition \ref{PropTrivialTimes}, we have $\mathcal{L}_{\mathbf{u}_{-1}}[T_2]=\hbar T_2$ and $\mathcal{L}_{\mathbf{u}_{-1}}[T_1]=0$. We also have from Theorem \ref{TheoSplitTangentSpace} $\mathcal{L}_{\mathbf{u}_{-1}}[q_j]=-\hbar q_j$ and $\mathcal{L}_{\mathbf{u}_{-1}}[p_j]=\hbar p_j$. Finally we observe that
\bea\mathcal{L}_{\mathbf{u}_{-1}}[\td{P}_1(\lambda)]&=&-\sum_{i=0}^{r_\infty-2}\mathcal{L}_{\mathbf{u}_{-1}}[t_{\infty,2i+2}]\lambda^{i}=-\frac{\hbar}{2} \sum_{i=0}^{r_\infty-2}\sum_{r=1}^{2r_\infty-2}r t_{\infty,r} \partial_{t_{\infty,r}}[t_{\infty,2i+2}]\lambda^{i}\cr
&=&-\frac{\hbar}{2} \sum_{i=0}^{r_\infty-2}(2i+2)t_{\infty,2i+2}\lambda^{i}
\eea
Thus, we get for all $j\in \llbracket 1,g\rrbracket$:
\small{\bea \mathcal{L}_{\mathbf{u}_{-1}}[\check{q}_j]&=&\mathcal{L}_{\mathbf{u}_{-1}}[T_2]q_j+ T_2\mathcal{L}_{\mathbf{u}_{-1}}[q_j]+\mathcal{L}_{\mathbf{u}_{-1}}[T_1]=\hbar T_2q_j+T_2(-\hbar q_j)=0\cr
\mathcal{L}_{\mathbf{u}_{-1}}[\check{p}_j]&=&-\frac{\mathcal{L}_{\mathbf{u}_{-1}}[T_2]}{T_2^2}\left(p_j-\frac{1}{2}\td{P}_1(q_j)\right)+T_2^{-1}\left(\mathcal{L}_{\mathbf{u}_{-1}}[p_j]-\frac{1}{2}\mathcal{L}_{\mathbf{u}_{-1}}[\td{P}_1](q_j) -\frac{1}{2}\mathcal{L}_{\mathbf{u}_{-1}}[q_j]\td{P}_1'(q_j)\right)\cr
&=&-\hbar T_2^{-1}\left(p_j-\frac{1}{2}\td{P}_1(q_j)\right)\cr
&&+\hbar T_2^{-1}\left(p_j+\frac{1}{4}\sum_{i=0}^{r_\infty-2}(2i+2)t_{\infty,2i+2}q_j^{i} -\frac{1}{2}q_j \sum_{i=1}^{r_\infty-2}it_{\infty,2i+2}q_j^{i-1}\right)\cr
&=&\hbar T_2^{-1}\left( -\frac{1}{2}\sum_{i=0}^{r_\infty-2}t_{\infty,2i+2}q_j^{i}+\frac{1}{4}\sum_{i=0}^{r_\infty-2}(2i+2)t_{\infty,2i+2}q_j^{i}-\frac{1}{2} \sum_{i=1}^{r_\infty-2}it_{\infty,2i+2}q_j^{i}\right)\cr
&=&0
\eea}

\normalsize{Let} us now consider $\mathcal{L}_{\mathbf{u}_{0}}$. From Proposition \ref{PropTrivialTimes}, we have $\mathcal{L}_{\mathbf{u}_{0}}[T_2]=0$ and $\mathcal{L}_{\mathbf{u}_{0}}[T_1]=\hbar T_2$. We also have from Theorem \ref{TheoSplitTangentSpace} $\mathcal{L}_{\mathbf{u}_{0}}[q_j]=-\hbar$ and $\mathcal{L}_{\mathbf{u}_{0}}[p_j]=0$. Finally we observe that
\bea\mathcal{L}_{\mathbf{u}_{0}}[\td{P}_1(\lambda)]&=&-\sum_{i=0}^{r_\infty-2}\mathcal{L}_{\mathbf{u}_{0}}[t_{\infty,2i+2}]\lambda^{i}\cr
&=&-\frac{\hbar}{2} \sum_{i=0}^{r_\infty-2}\sum_{r=1}^{2r_\infty-4}r t_{\infty,r+2} \partial_{t_{\infty,r}}[t_{\infty,2i+2}]\lambda^{i}\cr
&=&-\frac{\hbar}{2} \sum_{i=0}^{r_\infty-3}(2i+2)t_{\infty,2i+4}\lambda^{i}\overset{i=s-1}{=}-\hbar \sum_{s=1}^{r_\infty-2}st_{\infty,2s+2}\lambda^{s-1}
\eea
Thus, we get for all $j\in \llbracket 1,g\rrbracket$:
\bea \mathcal{L}_{\mathbf{u}_{0}}[\check{q}_j]&=&\mathcal{L}_{\mathbf{u}_{0}}[T_2]q_j+ T_2\mathcal{L}_{\mathbf{u}_{0}}[q_j]+\mathcal{L}_{\mathbf{u}_{0}}[T_1]=- \hbar T_2+\hbar T_2=0\cr
\mathcal{L}_{\mathbf{u}_{0}}[\check{p}_j]&=&-\frac{\mathcal{L}_{\mathbf{u}_{0}}[T_2]}{T_2^2}\left(p_j-\frac{1}{2}\td{P}_1(q_j)\right)+T_2^{-1}\left(\mathcal{L}_{\mathbf{u}_{0}}[p_j]-\frac{1}{2}\mathcal{L}_{\mathbf{u}_{0}}[\td{P}_1](q_j) -\frac{1}{2}\mathcal{L}_{\mathbf{u}_{0}}[q_j]\td{P}_1'(q_j)\right)\cr
&=&0+T_2^{-1}\left(0+\frac{\hbar}{2} \sum_{s=1}^{r_\infty-2}st_{\infty,2s+2}q_j^{s-1} -\frac{\hbar}{2} \sum_{i=1}^{r_\infty-2}it_{\infty,2i+2}q_j^{i-1}\right)\cr
&=&0
\eea

\section{Proof of Theorem \ref{TheoReduction}}\label{AppendixE}
Let us first observe that a function $f(t_{\infty,1},t_{\infty,2},\dots,t_{\infty,2r_\infty-2})$ solution of $\mathcal{L}_{\mathbf{w}_k}[f]=0$ for all $k\in \llbracket 1, r_\infty-1\rrbracket$ is independent of $(t_{\infty,2},\dots,t_{\infty,2r_\infty-2})$. Hence, the function $f$ may only depend on odd irregular times: 
\beq f(t_{\infty,1},\dots,t_{\infty,2r_\infty-2})=g(t_{\infty,1},t_{\infty,3}\dots,t_{\infty,2r_\infty-3})\eeq

Let us now consider $\mathcal{L}_{\mathbf{u}_{-1}}[g]=0$. It is equivalent to
\beq 0=\mathcal{L}_{\mathbf{u}_{-1}}[g]=\frac{\hbar}{2}\sum_{m=1}^{r_\infty-1} (2m-1)t_{\infty,2m-1}\partial_{t_{\infty,2m-1}}[g]\eeq 
whose solutions are arbitrary functions of 
\beq \label{DefIntermiediateyj} y_j=\frac{\frac{1}{2}t_{\infty,2j-1}}{\left(\frac{1}{2}t_{\infty,2r_\infty-3}\right)^{\frac{2j-1}{2r_\infty-3}}} \, \text{ with } \,j \in \llbracket 1,r_\infty-2\rrbracket\eeq
In other words:
\beq 0=\mathcal{L}_{\mathbf{u}_{-1}}[g] \,\Leftrightarrow \, g(t_{\infty,1},t_{\infty,3}\dots,t_{\infty,2r_\infty-3})=h(y_1,\dots,y_{r_\infty-2})\eeq

Let us now translate this result to $\mathcal{L}_{\mathbf{u}_{0}}[h]=0$. We find
\bea 0&=&\mathcal{L}_{\mathbf{u}_{0}}[h(y_1,\dots,y_{r_\infty-2})]=\frac{\hbar}{2}\sum_{m=1}^{r_\infty-2} (2m-1)t_{\infty,2m+1}\partial_{t_{\infty,2m-1}}[h(y_1,\dots,y_{r_\infty-2})]\cr
&=&\frac{\hbar}{2}\sum_{m=1}^{r_\infty-2} (2m-1)t_{\infty,2m+1}\sum_{r=1}^{r_\infty-2}\frac{\partial y_r}{\partial t_{\infty,2m-1}} \partial_{y_r} h(y_1,\dots,y_{r_\infty-2})\cr
&=&\frac{\hbar}{2}(2r_\infty-5)\frac{1}{2}t_{\infty,2r_\infty-3}\left(\frac{1}{2}t_{\infty,2r_\infty-3}\right)^{-\frac{2r_\infty-5}{2r_\infty-3}} \partial_{{y_{r_\infty-2}}} h(y_1,\dots,y_{r_\infty-2})\cr
&&+\frac{\hbar}{2}\sum_{m=1}^{r_\infty-3} (2m-1)\frac{1}{2}t_{\infty,2m+1}\left(\frac{1}{2}t_{\infty,2r_\infty-3}\right)^{-\frac{2m-1}{2r_\infty-3}}\partial_{y_m} h(y_1,\dots,y_{r_\infty-2})
\cr
&=& \frac{\hbar}{2}\left(\frac{1}{2}t_{\infty,2r_\infty-3}\right)^{\frac{2}{2r_\infty-3}}\Big[(2r_\infty-5)\partial_{{y_{r_\infty-2}}} h(y_1,\dots,y_{r_\infty-2})\cr
&&+\sum_{m=1}^{r_\infty-3} (2m-1)\frac{1}{2}t_{\infty,2m+1}\left(\frac{1}{2}t_{\infty,2r_\infty-3}\right)^{-\frac{2m+1}{2r_\infty-3}}\partial_{y_m} h(y_1,\dots,y_{r_\infty-2})\Big]
\cr
&=& \frac{\hbar}{2}\left(\frac{1}{2}t_{\infty,2r_\infty-3}\right)^{\frac{2}{2r_\infty-3}}\Big[(2r_\infty-5)\partial_{{y_{r_\infty-2}}} h(y_1,\dots,y_{r_\infty-2})\cr
&&+\sum_{m=1}^{r_\infty-3} (2m-1)y_{m+1}\partial_{y_m} h(y_1,\dots,y_{r_\infty-2})\Big]
\eea

For $g=0$ (that is, $r_\infty=3$), the last equation is $\partial_{y_1}h=0$, so its local solutions are constant and there are no non-trivial times. For $g\geq1$, we proceed using the following lemma.

\begin{lemma}\label{LemmaTechnical}Assume $r_\infty\geq4$. The general solutions of the differential equation
\beq (2r_\infty-5)\partial_{{y_{r_\infty-2}}}h(y_1,\dots,y_{r_\infty-2})+\sum_{m=1}^{r_\infty-3}(2m-1) y_{m+1}\partial_{y_m}h(y_1,\dots,y_{r_\infty-2})=0\eeq
are arbitrary functions of
\bea
f_1(y_{1},\dots,y_{r_\infty-2})&=&y_{r_\infty-3}-\frac{(2r_\infty-7)}{2(2r_\infty-5)} y_{r_\infty-2}^2\cr
&\vdots&\cr
f_k(y_{1},\dots,y_{r_\infty-2})&=&y_{r_\infty-2-k}+\sum_{i=1}^{k-1}\frac{(-1)^i\left(\underset{s=1}{\overset{i}{\prod}} (2r_\infty-2k+2s-7)\right)  y_{r_\infty-2}^i y_{r_\infty-2-k+i}}{i!(2r_\infty-5)^i}\cr
&&+ \frac{(-1)^{k}\left(\underset{s=1}{\overset{k-1}{\prod}}(2r_\infty-2k+2s-7)\right) (2r_\infty-7) y_{r_\infty-2}^{k+1}}{(k+1)(k-1)!(2r_\infty-5)^{k}}\cr
&=&\sum_{i=0}^{k-1}\frac{(-1)^i\left(\underset{s=1}{\overset{i}{\prod}} (2r_\infty-2k+2s-7)\right)  y_{r_\infty-2}^i y_{r_\infty-2-k+i}}{i!(2r_\infty-5)^i}\cr
&&+ \frac{(-1)^{k}\left(\underset{s=1}{\overset{k}{\prod}}(2r_\infty-2k+2s-7)\right) y_{r_\infty-2}^{k+1}}{(k+1)(k-1)!(2r_\infty-5)^{k}}
\eea 
where $k\in \llbracket 1, r_\infty-3\rrbracket$. 
\end{lemma}

\begin{proof}
 Let $k\in \llbracket 1, r_\infty-3\rrbracket$. We have:
\bea\label{FirstTerm} (2r_\infty-5) \partial_{y_{r_\infty-2}}f_k&=& \sum_{i=1}^{k-1}\frac{(-1)^i\left(\underset{s=1}{\overset{i}{\prod}} (2r_\infty-2k+2s-7)\right) y_{r_\infty-2}^{i-1}y_{r_\infty-2-k+i}}{(i-1)!(2r_\infty-5)^{i-1}}\cr
&&+\frac{(-1)^{k}\left(\underset{s=1}{\overset{k-1}{\prod}}(2r_\infty-2k+2s-7)\right) (2r_\infty-7) y_{r_\infty-2}^{k}}{(k-1)!(2r_\infty-5)^{k-1}}\cr
&=&\sum_{i=1}^{k}\frac{(-1)^i\left(\underset{s=1}{\overset{i}{\prod}} (2r_\infty-2k+2s-7)\right) y_{r_\infty-2}^{i-1}y_{r_\infty-2-k+i}}{(i-1)!(2r_\infty-5)^{i-1}}
 \eea
Moreover, we have:
\footnotesize{\bea \label{SecondTerm}\sum_{m=1}^{r_\infty-3}(2m-1)y_{m+1}\partial_{y_m}f_k&=&\sum_{i=0}^{k-1}\frac{(-1)^i(2r_\infty-2k+2i-5)\left(\underset{s=1}{\overset{i}{\prod}} (2r_\infty-2k+2s-7)\right) y_{r_\infty-2}^{i} y_{r_\infty-1-k+i}}{i!(2r_\infty-5)^{i}}\cr
&\overset{j=i+1}{=}&-\sum_{j=1}^{k}\frac{(-1)^j(2r_\infty-2k+2j-7)\left(\underset{s=1}{\overset{j-1}{\prod}} (2r_\infty-2k+2s-7)\right) y_{r_\infty-2}^{j-1}y_{r_\infty-k+j-2}}{(j-1)!(2r_\infty-5)^{j-1}}\cr
&&
\eea}
\normalsize{since} in the first equality only $m=r_\infty-2-k+i$ provides non-vanishing contributions. We now observe that \eqref{FirstTerm} and \eqref{SecondTerm} provides opposite contributions so that
\beq (2r_\infty-5)\partial_{{y_{r_\infty-2}}}f_k(y_1,\dots,y_{r_\infty-2})+\sum_{m=1}^{r_\infty-3}(2m-1) y_{m+1}\partial_{y_m}f_k(y_1,\dots,y_{r_\infty-2})=0\eeq
To prove that these are all local solutions, set $s=y_{r_\infty-2}$. In the variable order $(s,y_{r_\infty-3},\ldots,y_1)$, the map to $(s,f_1,\ldots,f_g)$ is triangular with unit diagonal: $f_k$ equals $y_{r_\infty-2-k}$ plus a polynomial in the preceding variables. Its Jacobian is therefore invertible. The vector field in the lemma annihilates every $f_k$ and sends $s$ to $2r_\infty-5$; in these coordinates it is
\beq
(2r_\infty-5)\frac{\partial}{\partial s}
\quad\text{with }f_1,\ldots,f_g\text{ held fixed}.
\eeq
Its kernel consists exactly of the local functions of $(f_1,\ldots,f_g)$, completing the proof.
\end{proof}

Combining Lemma \ref{LemmaTechnical} with \eqref{DefIntermiediateyj}, we obtain arbitrary functions of
\small{\bea \tau_k&=&\sum_{i=0}^{k-1}\frac{(-1)^i\left(\underset{s=1}{\overset{i}{\prod}} (2r_\infty-2k+2s-7)\right)  \left(\frac{1}{2}t_{\infty,2r_\infty-5}\right)^i \left(\frac{1}{2}t_{\infty,2r_\infty-3}\right)^{-\frac{(2r_\infty-3)i+2r_\infty-5-2k}{2r_\infty-3}} \frac{1}{2}t_{\infty, 2r_\infty-5-2k+2i} }{i!(2r_\infty-5)^i}\cr
&&+ \frac{(-1)^{k}\left(\underset{s=1}{\overset{k}{\prod}}(2r_\infty-2k+2s-7)\right) \left(\frac{1}{2}t_{\infty, 2r_\infty-5}\right)^{k+1} \left(\frac{1}{2}t_{\infty,2r_\infty-3}\right)^{-\frac{(k+1)(2r_\infty-5)}{2r_\infty-3} }}{(k+1)(k-1)!(2r_\infty-5)^{k}}\cr
&&
\eea}
\normalsize{with} $k\in \llbracket 1,r_\infty-3\rrbracket$.

\newpage

\bibliographystyle{plain}
\bibliography{Biblio}

@article{Darboux_coord93,
  author = {Adams, M. R. and Harnad, J. and Hurtubise, J.},
  title = {{Darboux} coordinates and {Liouville--Arnold} integration in loop algebras},
  journal = {Commun. Math. Phys.},
  fjournal = {Communications in Mathematical Physics},
  volume = {155},
  number = {2},
  pages = {385--413},
  year = {1993},
  doi = {10.1007/BF02097398},
  url = {https://doi.org/10.1007/BF02097398},
  issn = {0010-3616},
  mrclass = {58F07 (14H40 14H42)},
  mrnumber = {1230033},
  mrreviewer = {Emma Previato},
}

@article{HarnadHurtubise1997,
  author = {Adams, M. R. and Harnad, J. and Hurtubise, J.},
  title = {{Darboux} coordinates on coadjoint orbits of {Lie} algebras},
  journal = {Lett. Math. Phys.},
  fjournal = {Letters in Mathematical Physics},
  volume = {40},
  pages = {41--57},
  year = {1997},
  doi = {10.1023/A:1007355508426},
  url = {https://doi.org/10.1023/A:1007355508426},
}

@article{MOsl2,
  author = {O.~Marchal and N.~Orantin},
  title = {Isomonodromic deformations of a rational differential system and reconstruction with the topological recursion: the {$\mathfrak{sl}_2$} case},
  journal = {J. Math. Phys.},
  fjournal = {Journal of Mathematical Physics},
  volume = {61},
  number = {6},
  pages = {061506},
  year = {2020},
  doi = {10.1063/5.0002260},
  url = {https://doi.org/10.1063/5.0002260},
  issn = {0022-2488},
  mrclass = {34M56 (32G34 81Q20)},
}

@article{Quantization_2021,
  author = {B.~Eynard and E.~Garcia-Failde and O.~Marchal and N.~Orantin},
  title = {Quantization of classical spectral curves via topological recursion},
  journal = {Commun. Math. Phys.},
  fjournal = {Communications in Mathematical Physics},
  volume = {405},
  pages = {116},
  year = {2024},
  doi = {10.1007/s00220-024-04997-6},
  url = {https://doi.org/10.1007/s00220-024-04997-6},
  eprint = {2106.04339},
  archiveprefix = {arXiv},
  primaryclass = {math-ph},
}

@article{JimboMiwa,
  author = {M.~Jimbo and T.~Miwa},
  title = {Monodromy preserving deformation of linear ordinary differential equations with rational coefficients. {II}},
  journal = {Phys. D},
  fjournal = {Physica D. Nonlinear Phenomena},
  volume = {2},
  number = {3},
  pages = {407--448},
  year = {1981},
  doi = {10.1016/0167-2789(81)90021-X},
  url = {https://doi.org/10.1016/0167-2789(81)90021-X},
  issn = {0167-2789},
  mrclass = {34A20 (14K25 58A15 58F07 81C05)},
  mrnumber = {625446},
  mrreviewer = {V. A. Golubeva},
}

@article{BergereBorotEynard,
  author = {M.~Berg\`ere and G.~Borot and B.~Eynard},
  title = {Rational differential systems, loop equations, and application to the {$q$th} reductions of {KP}},
  journal = {Ann. Henri Poincar\'{e}},
  fjournal = {Annales Henri Poincar\'{e}. A Journal of Theoretical and Mathematical Physics},
  volume = {16},
  number = {12},
  pages = {2713--2782},
  year = {2015},
  doi = {10.1007/s00023-014-0391-8},
  url = {https://doi.org/10.1007/s00023-014-0391-8},
  issn = {1424-0637},
}

@incollection{DM18,
  author = {O.~Dumitrescu and M.~Mulase},
  title = {Quantization of spectral curves for meromorphic {H}iggs bundles through topological recursion},
  booktitle = {Topological recursion and its influence in analysis, geometry, and topology},
  series = {Proc. Sympos. Pure Math.},
  volume = {100},
  pages = {179--229},
  publisher = {American Mathematical Society},
  address = {Providence, RI},
  year = {2018},
  doi = {10.1090/pspum/100/01780},
  url = {https://doi.org/10.1090/pspum/100/01780},
  mrclass = {14H81 (14H15 33Cxx 34M60 53D37 81T45)},
}

@article{Kontsevich,
  author = {M.~Kontsevich},
  title = {Intersection theory on the moduli space of curves and the matrix {A}iry function},
  journal = {Commun. Math. Phys.},
  volume = {147},
  number = {1},
  pages = {1--23},
  year = {1992},
  doi = {10.1007/BF02099526},
  url = {https://doi.org/10.1007/BF02099526},
}

@article{EORev,
  author = {B.~Eynard and N.~Orantin},
  title = {Topological recursion in enumerative geometry and random matrices},
  journal = {J. Phys. A: Math. Theor.},
  volume = {42},
  number = {29},
  pages = {293001},
  year = {2009},
  doi = {10.1088/1751-8113/42/29/293001},
  url = {https://doi.org/10.1088/1751-8113/42/29/293001},
}

@article{EO07,
  author = {B.~Eynard and N.~Orantin},
  title = {Invariants of algebraic curves and topological expansion},
  journal = {Commun. Number Theory Phys.},
  fjournal = {Communications in Number Theory and Physics},
  volume = {1},
  number = {2},
  pages = {347--452},
  year = {2007},
  doi = {10.4310/CNTP.2007.v1.n2.a4},
  url = {https://doi.org/10.4310/CNTP.2007.v1.n2.a4},
}

@article{Boalch2012,
  author = {P. Boalch},
  title = {Simply-laced isomonodromy systems},
  journal = {Publ. Math. IH\'{E}S},
  fjournal = {Publications Math\'{e}matiques de l'IH\'{E}S},
  volume = {116},
  pages = {1--68},
  year = {2012},
  doi = {10.1007/s10240-012-0044-8},
  url = {https://doi.org/10.1007/s10240-012-0044-8},
}

@article{Boalch2001,
  author = {P. Boalch},
  title = {Symplectic manifolds and isomonodromic deformations},
  journal = {Adv. Math.},
  fjournal = {Advances in Mathematics},
  volume = {163},
  number = {2},
  pages = {137--205},
  year = {2001},
  doi = {10.1006/aima.2001.1998},
  url = {https://doi.org/10.1006/aima.2001.1998},
}

@article{Boalch2022,
  author = {P. Boalch and J. Dou{\c c}ot and G. Rembado},
  title = {Twisted local wild mapping class groups: configuration spaces, fission trees and complex braids},
  journal = {Publ. Res. Inst. Math. Sci.},
  volume = {61},
  number = {3},
  pages = {391--445},
  year = {2025},
  doi = {10.4171/PRIMS/61-3-4},
  url = {https://doi.org/10.4171/PRIMS/61-3-4},
  eprint = {2209.12695},
  archiveprefix = {arXiv},
  primaryclass = {math.AG},
}

@article{schlesinger1912klasse,
  author = {Schlesinger, L.},
  title = {{\"U}ber eine {Klasse} von {Differentialsystemen} beliebiger {Ordnung} mit festen kritischen {Punkten}},
  journal = {J. Reine Angew. Math.},
  volume = {141},
  pages = {96--145},
  year = {1912},
  doi = {10.1515/crll.1912.141.96},
  url = {https://doi.org/10.1515/crll.1912.141.96},
}

@article{Garnier,
  author = {R. Garnier},
  title = {Solution du probl\`{e}me de {Riemann} pour les syst\`{e}mes diff\'{e}rentiels lin\'{e}aires du second ordre},
  journal = {Ann. Sci. \'{E}cole Norm. Sup. (3)},
  volume = {43},
  pages = {177--307},
  year = {1926},
  doi = {10.24033/asens.771},
  url = {https://doi.org/10.24033/asens.771},
}

@article{HarnadHurtubise,
  author = {M. R. Adams and J. Harnad and J. Hurtubise},
  title = {Dual moment maps into loop algebras},
  journal = {Lett. Math. Phys.},
  volume = {20},
  number = {4},
  pages = {299--308},
  year = {1990},
  doi = {10.1007/BF00626526},
  url = {https://doi.org/10.1007/BF00626526},
}

@article{Fuchs,
  author = {R. Fuchs},
  title = {Sur quelques \'{e}quations diff{\'e}rentielles lin\'{e}aires du second ordre},
  journal = {C. R. Acad. Sci. Paris},
  volume = {141},
  pages = {555--558},
  year = {1905},
}

@article{Gambier,
  author = {B. Gambier},
  title = {Sur les \'{e}quations diff\'{e}rentielles du second ordre et du premier degr\'{e} dont l'int\'{e}grale g\'{e}n\'{e}rale est \`{a} points critiques fixes},
  journal = {Acta Math.},
  volume = {33},
  pages = {1--55},
  year = {1910},
  doi = {10.1007/BF02393211},
  url = {https://doi.org/10.1007/BF02393211},
}

@article{Painleve,
  author = {P. Painlev\'{e}},
  title = {Sur les \'{e}quations diff\'{e}rentielles du second ordre et d'ordre sup\'{e}rieur dont l'int\'{e}grale g\'{e}n\'{e}rale est uniforme},
  journal = {Acta Math.},
  volume = {25},
  pages = {1--85},
  year = {1902},
  doi = {10.1007/BF02419020},
  url = {https://doi.org/10.1007/BF02419020},
}

@article{Picard,
  author = {E. Picard},
  title = {M\'{e}moire sur la th\'{e}orie des fonctions alg\'{e}briques de deux variables},
  journal = {J. Math. Pures Appl. (4)},
  volume = {5},
  pages = {135--320},
  year = {1889},
  url = {https://www.numdam.org/item/JMPA_1889_4_5__135_0.pdf},
}

@article{biquard_boalch_2004,
  author = {O. Biquard and P. Boalch},
  title = {Wild non-abelian {Hodge} theory on curves},
  journal = {Compos. Math.},
  fjournal = {Compositio Mathematica},
  volume = {140},
  number = {1},
  pages = {179--204},
  publisher = {London Mathematical Society},
  year = {2004},
  doi = {10.1112/S0010437X03000010},
  url = {https://doi.org/10.1112/S0010437X03000010},
}

@article{HURTUBISE20081394,
  author = {J. Hurtubise},
  title = {On the geometry of isomonodromic deformations},
  journal = {J. Geom. Phys.},
  fjournal = {Journal of Geometry and Physics},
  volume = {58},
  number = {10},
  pages = {1394--1406},
  year = {2008},
  doi = {10.1016/j.geomphys.2008.05.013},
  url = {https://doi.org/10.1016/j.geomphys.2008.05.013},
  issn = {0393-0440},
}

@article{BertolaHarnadHurtubise2022,
  author = {M. Bertola and J. Harnad and J. Hurtubise},
  title = {Hamiltonian structure of rational isomonodromic deformation systems},
  journal = {J. Math. Phys.},
  volume = {64},
  number = {8},
  pages = {083502},
  year = {2023},
  doi = {10.1063/5.0142532},
  url = {https://doi.org/10.1063/5.0142532},
  eprint = {2212.06880},
  archiveprefix = {arXiv},
  primaryclass = {nlin.SI},
}

@article{MartaPaper2022,
  author = {I. Gaiur and M. Mazzocco and V. Rubtsov},
  title = {Isomonodromic deformations: Confluence, reduction and quantisation},
  journal = {Commun. Math. Phys.},
  fjournal = {Communications in Mathematical Physics},
  volume = {400},
  number = {2},
  pages = {1385--1461},
  year = {2023},
  doi = {10.1007/s00220-023-04650-8},
  url = {https://doi.org/10.1007/s00220-023-04650-8},
}

@article{MarchalOrantinAlameddine2022,
  author = {O. Marchal and N. Orantin and M. Alameddine},
  title = {Hamiltonian representation of isomonodromic deformations of general rational connections on {$\mathfrak{gl}_2(\mathbb{C})$}},
  journal = {Nonlinearity},
  volume = {38},
  number = {11},
  pages = {115018},
  year = {2025},
  doi = {10.1088/1361-6544/ae19dd},
  url = {https://doi.org/10.1088/1361-6544/ae19dd},
  eprint = {2212.04833},
  archiveprefix = {arXiv},
  primaryclass = {math-ph},
}

@misc{boalch2015twistedwildcharactervarieties,
  author = {P. Boalch and D. Yamakawa},
  title = {Twisted wild character varieties},
  year = {2015},
  url = {https://arxiv.org/abs/1512.08091},
  eprint = {1512.08091},
  archiveprefix = {arXiv},
  primaryclass = {math.AG},
  note = {arXiv:1512.08091},
}

@article{boalch2012geometrybraidingstokesdata,
  author = {P. Boalch},
  title = {Geometry and braiding of {S}tokes data; {F}ission and wild character varieties},
  journal = {Ann. of Math. (2)},
  fjournal = {Annals of mathematics},
  volume = {179},
  number = {1},
  pages = {301--365},
  year = {2014},
  doi = {10.4007/annals.2014.179.1.5},
  url = {https://doi.org/10.4007/annals.2014.179.1.5},
}

@article{marchal2023isomonodromicisospectraldeformationsmeromorphic,
  author = {O. Marchal and M. Alameddine},
  title = {Isomonodromic and isospectral deformations of meromorphic connections: the {$\mathfrak{sl}_2(\mathbb{C})$} case},
  journal = {Nonlinearity},
  volume = {37},
  number = {11},
  pages = {115006},
  year = {2024},
  doi = {10.1088/1361-6544/ad7b96},
  url = {https://doi.org/10.1088/1361-6544/ad7b96},
  eprint = {2306.07378},
  archiveprefix = {arXiv},
  primaryclass = {math-ph},
}

@article{Yamakawa2019FundamentalTwoForms,
  author = {D. Yamakawa},
  title = {Fundamental two-forms for isomonodromic deformations},
  journal = {J. Integrable Syst.},
  fjournal = {Journal of Integrable Systems},
  volume = {4},
  number = {1},
  pages = {xyz009},
  year = {2019},
  doi = {10.1093/integr/xyz009},
  url = {https://doi.org/10.1093/integr/xyz009},
}

@article{RembadoDoucot2023,
  author = {J. Dou{\c c}ot and G. Rembado},
  title = {Topology of irregular isomonodromy times on a fixed pointed curve},
  journal = {Transform. Groups},
  fjournal = {Transformation Groups},
  volume = {30},
  number = {2},
  pages = {591--631},
  year = {2025},
  doi = {10.1007/s00031-023-09800-9},
  url = {https://doi.org/10.1007/s00031-023-09800-9},
}

@article{Komyo2024,
  author = {A. Komyo},
  title = {Description of generalized isomonodromic deformations of rank-two linear differential equations using apparent singularities},
  journal = {Publ. Res. Inst. Math. Sci.},
  volume = {60},
  number = {1},
  pages = {185--269},
  year = {2024},
  doi = {10.4171/PRIMS/60-1-5},
  url = {https://doi.org/10.4171/PRIMS/60-1-5},
}

@misc{KomyoLoray2023,
  author = {A. Komyo and F. Loray and M.-H. Saito and S. Szab\'{o}},
  title = {Canonical coordinates for moduli spaces of rank two irregular connections on curves},
  year = {2023},
  url = {https://arxiv.org/abs/2309.05012},
  eprint = {2309.05012},
  archiveprefix = {arXiv},
  primaryclass = {math.AG},
  note = {arXiv:2309.05012},
}
\addcontentsline{toc}{section}{References}

\end{document}